\pdfoutput=1
\documentclass[a4paper,12pt,reqno]{article}
\usepackage[markup=underlined]{changes}
\RequirePackage{etex}

\usepackage[T1]{fontenc}
\usepackage{fourier}
\usepackage{newcent}

\DeclareMathSizes{12}{13}{8}{6}

\usepackage{etex}
\usepackage{graphicx}
\usepackage{setspace}
\usepackage{amsmath}
\usepackage{amsfonts}
\usepackage{titlesec}
\usepackage{pictex}
\usepackage{comment}
\usepackage{amsthm}
\usepackage{graphics,epsfig,verbatim,bm,latexsym,url,amsbsy}
\usepackage{rotating}
\usepackage[authoryear,round]{natbib}
\usepackage{float}
\usepackage[position=bottom]{subfig}
\usepackage{mathrsfs}
\usepackage{multirow}
\usepackage{array}
\usepackage{bigints}
\usepackage{bbm}
\usepackage{geometry}
\usepackage{bbm}
\usepackage[ruled,vlined]{algorithm2e}
\usepackage{soul}
\usepackage{mathtools}
\usepackage{enumitem}
\usepackage{amssymb}
\usepackage{diagbox}   
\usepackage{array}     
\usepackage{makecell}  

\DeclarePairedDelimiterX{\inp}[2]{\langle}{\rangle}{#1, #2}

\usepackage{hyperref}
\hypersetup{
	colorlinks=true,
	linkcolor=red!50!black,
	citecolor=blue!65!black,
	filecolor=magenta,      
	urlcolor=cyan,
}

\usepackage[nameinlink,noabbrev,sort,capitalise]{cleveref}

\newcommand{\de}{\mathop{}\!\mathrm{d}}
\newcommand{\e}{\varepsilon}

\renewcommand{\Re}{\mathbb{R}}

\newcommand{\marg}{\mathrm{marg}}
\newcommand{\Gr}{\mathrm{Gr}}

\newcommand{\oF}{\mkern2mu \overline{\mkern-2mu F } \mkern2mu}

\theoremstyle{definition} \newtheorem{example}{Example}
\theoremstyle{definition} 
\theoremstyle{definition} 

\theoremstyle{definition} \newtheorem{claim}{Claim}
\theoremstyle{definition} \newtheorem{definition}{Definition}
\theoremstyle{definition} 
\theoremstyle{definition} 
\theoremstyle{definition} \newtheorem{lemma}{Lemma}
\theoremstyle{definition} \newtheorem{theorem}{Theorem}
\theoremstyle{definition} \newtheorem{assumption}{Assumption}
\theoremstyle{definition} 
\theoremstyle{definition} 
\theoremstyle{definition} 
\theoremstyle{definition}\newtheorem{proposition}{Proposition}
\theoremstyle{definition} 
\theoremstyle{definition} 
\theoremstyle{definition} 
\theoremstyle{definition} 
\theoremstyle{definition} 
\theoremstyle{definition}

\newtheorem{manualtheoreminner}{Theorem}
\newenvironment{manualtheorem}[1]{%
  \renewcommand\themanualtheoreminner{#1}%
  \manualtheoreminner
}{\endmanualtheoreminner}

\theoremstyle{definition} 

\usepackage{mathtools}
\titlespacing*{\paragraph}{0pt}{1.25ex plus 1ex minus .2ex}{0.5em}

\titleformat{\section}
		{\bfseries\center}
         {\thesection}
        {0.5em}
        {}
        []

\titleformat{\subsection}[runin]
        {\normalfont\bfseries}
        {\thesubsection}
        {0.5em}
        {\addperiod}
        []
\newcommand{\addperiod}[1]{#1.}

\title{\LARGE{\scshape{\textbf{Flexible Demand Manipulation}}}}

\date{This version: \today}

\author{\makebox[.2\linewidth]{{Yifan Dai}\thanks{MIT Department of Economics; email: \protect\texttt{yfdai@mit.edu}}}\\{MIT} \and
\makebox[.2\linewidth]{{Andrew Koh}\thanks{MIT Department of Economics; email: \protect\texttt{ajkoh@mit.edu}\\~\\
\emph{First posted version: February 2024.}\\
We are particularly grateful to Drew Fudenberg and Stephen Morris for guidance, support, and many helpful discussions. We also thank Daron Acemoglu, Alessandro Bonatti, Ian Ball, Roberto Corrao, Laura Doval, Matt Elliott, Glenn Ellison, Andrea Galeotti, Nima Haghpanah, Jiangtao Li, Ellen Muir, Sivakorn Sanguanmoo, Anna Sanktjohanser, Yucheng Shang, Teck Yong Tan, Alex Wolitzky, and Jidong Zhou for helpful suggestions. We benefited from collaboration with Shijian Yang in the preliminary stages of this project. We also thank participants at the 2024 Berkeley-Columbia-Duke-MIT-Northwestern IO Theory Conference, the 2023 MIT Summer Theory Lab, and MIT Behavioral, Theory and IO lunches for helpful comments.}}\\{MIT}}

\begin{document}
\maketitle
\thispagestyle{empty}

\begin{abstract}
    We develop a simple framework to analyze how targeted persuasive advertising shapes market power and welfare. A designer flexibly manipulates the demand curve by influencing individual valuations at a cost. A monopolist prices against this manipulated demand curve. We fully characterize the form of optimal advertising plans under \emph{ex-ante} and \emph{ex-post} welfare measures. Flexibility \emph{per se} is powerful, and can substantially harm or benefit consumers vis-a-vis uniform advertising. We discuss implications for regulation, intermediation, and the joint design of manipulation and information.
\end{abstract}

\clearpage

\setcounter{page}{1}

\section{Introduction}
Targeted advertising is a central feature of modern markets.\footnote{For instance, expenditure in the US on targeted advertising has tripled between 2017 and 2022.} Because of their ubiquity and implications for welfare, it has accordingly received considerable attention from policymakers.\footnote{For instance, in 2018 the EU's General Data Protection Regulation (GDPR) went into effect; in the US, California has implemented a GDPR-like law with more states to follow suit in 2023.} Indeed many legislators have called for, or implemented stringent regulations: European regulators recently banned Meta, a large social media company, from using personal data to target ads. In the US, the proposed American Privacy Rights Act---intended to \emph{`establish national consumer data privacy rights and set standards for data
security'}---proposes to prohibit targeted ads based on protected class information and personal data purchased from data brokers.

These policies have been largely driven by privacy concerns---that is, concerns over how firms acquire data to personalize their advertisements and the implications for data protection. However, the implications of targeted advertising \emph{qua} advertising for prices and welfare remain unclear: how does the ability to tailor the volume and content of advertisements to precise segments of the population shape prices? What are the attendant implications for welfare and regulation? 

Motivated by these questions, we develop a simple but novel framework to study the interplay between targeted advertising and market power. Our
notion of advertising is broad. It includes both overt forms e.g., via social media, e-commerce platforms, or tailored search results, as well as more covert forms e.g., the non-informational content of recommender systems (`digital nudges'\footnote{This umbrella term has been used to describe the power for digital systems to systematically shape users' behavior \citep{weinmann2016digital,schneider2018digital}; see a recent survey by \cite{jesse2021digital}. We discuss mechanisms when we detail the model.}) or via interaction with artificial intelligence systems like large-language models.\footnote{Recent work finds strong experimental evidence that `human decision making is susceptible to AI manipulation' \citep{sabour2025human}. There is also increasing incorporation of advertising into large language models; see e.g., \cite{meguellati2024good} and \url{https://hbr.org/2024/05/how-marketers-can-adapt-to-llm-powered-search}.} Importantly, advertising in our model is entirely persuasive in the sense that it increases willingness to pay directly:\footnote{This stands in contrast to ``informative'' advertising in which consumers learn about the existence or quality of a product. See, e.g., \cite{kaldor1950economic} for an early discussion of this distinction.}
 a designer chooses an advertising plan where the cost of increasing individual valuations from $x$ to $y$ is $c(x,y)$. Our model reflects the designer's ability to: 
\begin{itemize}
    \item[(i)] \emph{Target precisely} which corresponds to its ability to shape the valuations of different types differently. This reflects a distinct feature of online advertising where substantial amounts of consumer data are collected and the costs of targeting different kinds and quantities of ads at different consumers are dramatically reduced \citep{goldfarb2014different}.
    \item[(ii)] \emph{Persuade precisely} which corresponds to its ability to increase the valuation of type $x$ to exactly $y$. This is motivated by the increasing persuasiveness of new technologies such as artificial intelligence, especially for marketing consumer products \citep{matz2024potential,schoenegger2025large}.\footnote{\cite{matz2024potential} find experimentally that personalized messages by language models \emph{`exhibit significantly more influence than non-personalized messages... across different domains of persuasion (e.g., marketing of consumer products)'.} \cite{schoenegger2025large} find that language models are \emph{`more persuasive than incentivized human persuaders'}. \cite{salvi2024conversational} offer RCT evidence that large language models are substantially (81.7\%) more likely to be persuasive than humans.} Such newfound possibilities for persuasion have led AI firms to incorporate advertisements into their large language models,\footnote{Google has recently added ads into its `AI mode'; likewise, the Financial Times reported in December 2024 that \emph{`OpenAI explores advertising as it steps up revenue drive'}.} 
    as well as regulatory scrutiny (see Article 5 of the EU AI Act on `subliminal techniques' for manipulation).\footnote{The act prohibits the use of \emph{`...manipulative and deceptive techniques that may result in the distortion of a person’s ability to make an informed decision.'}}
\end{itemize}

Taken together, precision in both targeting and persuasion allows the designer to \emph{flexibly manipulate demand}: any demand curve which first-order stochastically dominates the original one is feasible at a cost. We discipline these costs by assuming that (i) total costs are the sum of individual costs; and (ii) individual costs are submodular. Facing this transformed distribution, a monopolist sets an optimal posted price. We begin with the following simple example to illustrate the power and form of flexible demand manipulation. 

\begin{example}\label{eg:uniform_eg}
    Consider a monopoly environment where each consumer inelastically demands a single unit and initial valuations are uniformly distributed over $[0,1]$. A designer can, by targeting different kinds and quantities of persuasive ads at different consumers, precisely manipulate valuations. Suppose that the cost of increasing the valuation of type $x$ to $y > x$ is proportional to the quadratic distance. We start by considering {ex-ante} welfare---that is, consumer surplus measured with respect to initial rather than manipulated valuations.\footnote{Throughout this paper we use the term `valuations' in a normatively-neutral way to reflect consumers' willingness-to-pay; this can diverge from welfare (and does so under the ex-ante measure).} 

    \cref{fig:uniform_example} illustrates the welfare possibilities from flexibly manipulating the demand curve. Flexibility \emph{per se} affords substantially more power to shape both (i) total welfare; and (ii) the surplus split between producers and consumers. To illustrate this, we also consider the benchmark under which all consumers' valuations are shifted by equal amounts (under the same cost function) which in the spirit of classic analyses of persuasive advertising \citep{dorfman1954optimal,dixit1978advertising}. 
\begin{figure}[H]
\centering
\caption{The power of flexible demand manipulation}
    {\includegraphics[width=0.7\textwidth]{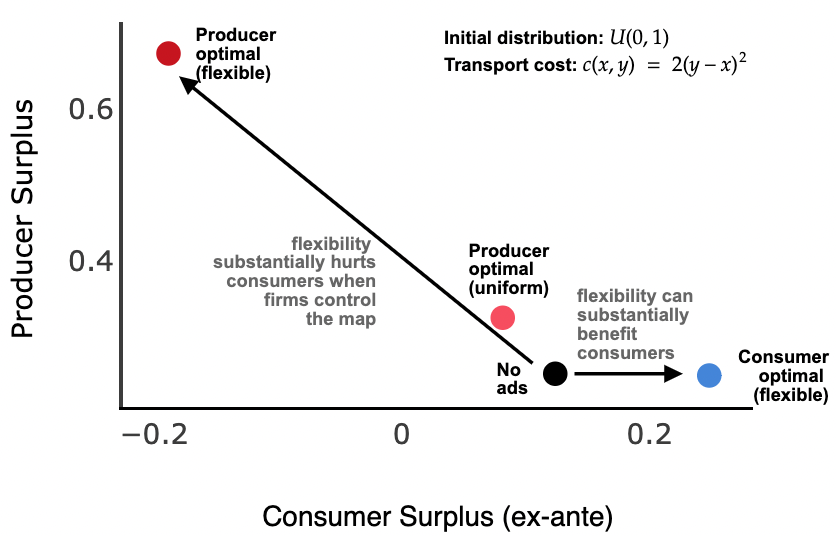}}
    \label{fig:uniform_example}
\end{figure}

\textbf{No advertisements.} First observe that sans advertisements, the monopolist charges an optimal price of $1/2$ and consumer surplus is $1/8$. The corresponding welfare outcome is illustrated by the \textbf{black dot} in \cref{fig:uniform_example}. This is the familiar monopoly pricing problem which trades off the intensive margin against the extensive margin and is illustrated in the first panel of \cref{fig:uniform_exampleB}.

\textbf{Uniform advertising.} Next consider the case of uniform advertising in which all consumers valuations must be shifted by the same amount, and suppose that the designer maximizes producer surplus. This is equivalent to jointly choosing the monopolist's price and the (scalar-valued) volume of uniform advertising directly---the producer-optimal uniform advertising plan is then pinned down by the classic analysis of \cite{dorfman1954optimal}. The uniform shift to the demand curve is illustrated in the second panel of \cref{fig:uniform_exampleB}. This pushes up the monopolist's price to $4/7$ which raises producer surplus and decreases consumer surplus as illustrated by the \textbf{\textcolor{red}{red dot}} in \cref{fig:uniform_example}.

\begin{figure}[H]
\centering
\caption{ The workings of flexible demand manipulation}
\vspace{-1.8em}
  \begin{quote}
  \centering
  {\footnotesize CS (ex-ante): blue minus dark red; PS: red (light and dark)}
  \end{quote}
\includegraphics[width=1\textwidth]{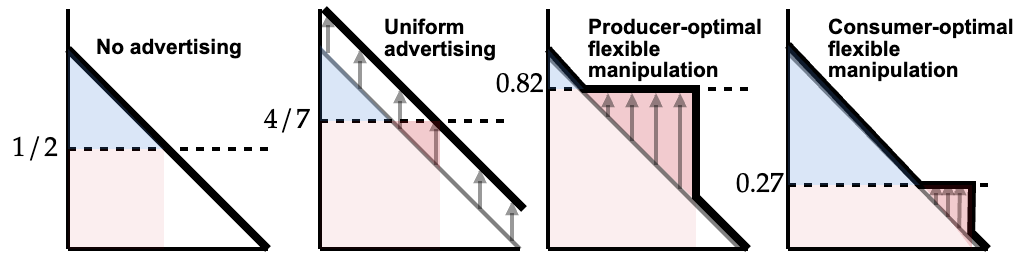}
\label{fig:uniform_exampleB}
\end{figure}

\vspace{-2em}\textbf{Producer-optimal flexible advertising.} Now consider the producer-optimal advertising plan under full flexibility. The corresponding welfare outcome is illustrated by the \textbf{\textcolor{red!60!black}{dark red dot}} in \cref{fig:uniform_example}. This {exacerbates} the losses to consumer welfare: with finer control, producer surplus is maximized at the substantially higher price of $\approx 0.82$ because demand is shored up via flexibly targeted advertisements. This is illustrated by the third panel of \cref{fig:uniform_exampleB}.

\textbf{Consumer-optimal flexible advertising.} Finally, consider the consumer-optimal advertising plan under full flexibility. The corresponding welfare outcome is illustrated by the \textbf{\textcolor{blue}{blue dot}} in \cref{fig:uniform_example}. By carefully choosing the advertising plan, the designer is able to induce the monopolist to set a low price of $\approx 0.27$ and deliver substantially (i) more consumer surplus; and (ii) more gains from trade. How the demand curve is optimally manipulated to induce this is illustrated by the last panel of \cref{fig:uniform_exampleB}.

\begin{figure}[H]
\begin{minipage}[t]{0.68\linewidth}
\textbf{Ex-post consumer-optimal plans.} Now consider, instead, ex-post consumer surplus which measures welfare with respect to \emph{final} valuations. Inframarginal consumers can now be made better-off via advertising---but increasing their valuations too much might tempt the monopolist to raise its prices. 
\cref{fig:uniform_example_intro_expost} illustrates how these forces are optimally traded off: the designer continues to induce a low price of $0.25$ but, different from the ex-ante case, also advertises to inframarginal consumers such as to simultaneously fulfill the (i) \emph{technological tradeoff} given by the advertising technology; and (ii) the monopolist's \emph{pricing constraint}.  $\hfill \diamondsuit$

\end{minipage}%
\hfill%
\begin{minipage}[t]{0.3\textwidth}\vspace{0pt}
\vspace{1em}
\centering
{\includegraphics[width=0.9\textwidth]{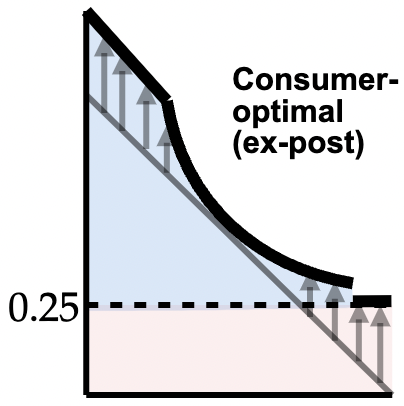}}
    \caption{}
    \footnotesize {CS (ex-post)}: blue \\ {Producer surplus}: red
\label{fig:uniform_example_intro_expost}
\end{minipage}
\end{figure}
\end{example}

 \cref{eg:uniform_eg} worked within a simple environment with uniform demand and quadratic costs. In what follows, we generalize and extend these insights by fully characterizing the form and value of flexible demand manipulation. This offers a novel perspective on a classic view articulated by \cite{johnson2006simple} that the function of persuasive advertising is to \emph{"change the position of a firm’s demand curve"} (a scalar-valued parametric shift) whereas informative advertising \emph{"rotates demand, thereby changing its shape"}. Our model of targeted persuasive advertising simultaneously alters both position and shape\footnote{Our baseline model in which targeted advertising weakly increases valuations already alters both position and shape of the demand curve; we further analyze the joint design of flexible information and manipulation (\cref{sec:joint}) as well as extending our analysis to accommodate both `forward' and `backward' shifts (Online Appendix \ref{appendix:directional}).}---we analyze the implications for the form of optimal advertising plans, welfare, and regulation.

\textbf{Outline of contribution.} As we observed in \cref{eg:uniform_eg}, some care is required in how we think about consumer welfare.\footnote{\cite{dixit1978advertising} also distinguishes between these welfare measures: what we call `ex-ante' they call `preadvertising'; what we call `ex-post' they call `postadvertising'.} We consider two measures: 
\begin{enumerate}[leftmargin = 2em]
    \item[(i)] \textbf{\emph{Ex-ante welfare}: advertising as behavioral manipulation.} A consumer transported from $x$ to $y$ who buys the product at price $p$ obtains surplus of $x - p$ i.e., with respect to her initial valuation. Note that this can be negative if $x < p \leq y$ and is consistent with the view of advertising as \emph{behavioral manipulation}: it drives a wedge between choice (willingness-to-pay) and preferences.  
    \item[(ii)] \textbf{\emph{Ex-post welfare}: advertising as a complementary good.} A consumer transported from $x$ to $y$ who chooses to buy the product at price $p$ obtains surplus $y - p$ i.e., with respect to her {new} valuation. This is always non-negative, and is consistent with the view of advertising as a \emph{complementary good}: it increases valuations directly.   
\end{enumerate}  

We are interested in the form of advertising plans that maximize a general objective that is weakly increasing in consumer and producer surplus, and strictly decreasing in advertising costs. This encapsulates the familiar consumer- and producer-optimal advertising plans, but also allows for richer objectives such as those that arise when the designer must deliver surplus to both sides of a market with endogenous entry, or when the designer wishes to robustly maximize welfare but faces non-Bayesian uncertainty about the appropriate welfare weights. Our main results completely characterize the form of optimal advertising plans under ex-ante (\cref{thrm:exante_general}) and ex-post (\cref{thrm:expost_general}) welfare.

\textbf{Optimal plans under ex-ante welfare.}
All optimal plans under an ex-ante notion of welfare take an \emph{intermediate interval structure} (\cref{thrm:exante_general}). As we saw in \cref{eg:uniform_eg}, this comprises a pair $(\underline p, p^*)$ where $p^*$ is the monopolist's optimal price facing the manipulated distribution, and consumers with valuations $x \in [\underline p, p^*)$ are transported precisely to $p^*$. While all optimal plans take this form, the economic constraints shaping the choices of $\underline p$ and $p^*$ vary with the designer's objective.

When the designer values producer surplus highly (e.g., if the monopolist is in control of the plan), it uses the additional flexibility from advertising to \emph{decouple the intensive-extensive margin} tradeoff typical of monopoly pricing problems. The monopolist is induced to set a high price and, while it would previously have lost too much on the extensive margin, targeted advertising is able to precisely shore up demand. We show analytically and numerically that this tends to raise prices substantially vis-a-vis uniform advertising. Conversely, when the designer values ex-ante consumer surplus highly, it uses the additional flexibility from advertising to make the demand curve \emph{asymmetrically elastic}, thereby tempting the monopolist to lower its price in order to capture additional demand. This improves total consumer surplus by lowering prices for inframarginal consumers, but is to the detriment of marginal consumers who are manipulated into buying.

\textbf{Optimal plans under ex-post welfare.} When our notion of welfare is ex-post, advertising to inframarginal consumers can make them better-off. Thus, the advertising plan must navigate two economically distinct tradeoffs. First, the advertising technology pins down a \emph{technological tradeoff} between the benefits of improving valuations relative to the costs of doing so. But the designer must also contend with the monopolist's incentives, as reflected by a \emph{pricing tradeoff}: if the designer raises valuations excessively, the monopolist becomes tempted to increase its prices which hurts consumer surplus.

Advertising plans that tradeoff these two forces optimally take a \emph{constrained greedy structure} (\cref{thrm:expost_general}). Constrained greedy plans partition consumers according to their initial valuations into three groups: consumers with low initial valuations are excluded; intermediate consumers are transported precisely to a target price $p^*$; and consumers with high valuations are transported to the {lower envelope} of a \emph{locally-greedy map} reflecting the cost and benefits associated with marginal increases in valuations; and a \emph{unit-elastic demand curve} reflecting the need to deter upward deviations. Our solution is simple, and combines (i) the familiar logic of optimal monopoly pricing; with (ii) a functional first-order approach to optimize over the target distribution; and (iii) basic results from optimal transportation on the structure of least-cost advertising plans. 

\textbf{Manipulation vs information.} Flexible demand manipulation via persuasive advertising is distinct from informative advertising in which consumers \emph{"learn about their taste for a product"} \citep{johnson2006simple}. Nonetheless, advertising in practice has elements of both \citep{kaldor1950economic}. To better understand how these instruments interact, we augment our model to study the \emph{joint} design of manipulation and information and analyze the form and value of producer- and consumer-optimal plans (\cref{prop:joint}).
When the designer values producer surplus, manipulation and information jointly deliver more value than each can deliver separately---the optimal plan is \emph{no information and all hype}: consumers are kept in the dark about idiosyncratic match qualities, but their valuations are manipulated via persuasive advertising. When the designer values ex-ante consumer surplus, manipulation has no value over just delivering flexible information \citep{roesler2017buyer}---the optimal plan is \emph{no hype and partial information}. Finally, when the designer values ex-post consumer surplus, manipulation and information once again deliver strictly more value together than each can deliver separately.

\textbf{Welfare and policy implications.} We employ our characterization of optimal advertising plans under ex-ante (\cref{thrm:exante_general}) and ex-post welfare (\cref{thrm:expost_general}) to shed light---both analytically and numerically---on how targeted advertising shapes the surplus split among producers and consumers (\cref{sec:comparison}). We find that for a wide range of demand functions, targeted advertising is substantially more powerful than uniform advertising (under the same cost function), and can be used to deliver substantially more producer surplus and/or consumer surplus. This suggests that the interaction between targeting and manipulation is a first-order economic force independent of privacy or other concerns. In the absence of regulation,\footnote{In practice, firms do exert some control over the plan (e.g., which ad auctions to bid on).} consumers can be doubly hurt---inframarginal consumers face higher prices, and marginal consumers are manipulated into buying against their own interests. What then, if anything, can regulators do to safeguard consumer welfare?

We show how our framework and results can be extended to analyze these questions. We model regulation as constraining the set of available advertising maps on offer with the constraint that firms can freely deviate downwards.\footnote{i.e., it can always choose not to advertise, or advertise less than any feasible plan prescribes.} This specifies (i) \emph{coarseness of targeting} which determines how finely firms can advertise to different consumers; and (ii) \emph{limits on manipulation} which determines how much firms can shift valuations. Our framework nests a number of benchmark policies such as an outright ban on targeting which recovers the uniform advertising case, and a Pigouvian tax on advertising which has been recently proposed (\cite{romer2021taxing,acemoglu2024online}). We solve for consumer-optimal regulatory policies and show that it can deliver substantial consumer welfare. 
By contrast, coarse instruments like outright bans or taxes might ameliorate consumer harms, but cannot improve consumer welfare over the no-advertising benchmark.

\textbf{Related literature.} Our model nests, as special cases, classic analyses of optimal persuasive advertising for a monopolist seller \citep{dorfman1954optimal,dixit1978advertising,nichols1985advertising}. 
These models reflect coarse advertising---e.g., via newspapers, radios, TV and so on---in which every consumer had, roughly, the same probability of viewing the ad. The shift towards online advertising has opened new possibilities for advertisements to be targeted \citep{goldfarb2014different}; our model speaks directly to this. Contemporaneous work by \cite*{acemoglu2024online} studies how platforms might raise revenue via informative advertising to misspecified consumers, thereby offering a simple microfoundation for persuasive advertising.\footnote{Our approach to modeling persuasive advertising is deliberately ``reduced-form''---we are agnostic about the channels valuations might be influenced, though we discuss potential mechanisms when we detail our model.} In all of these papers, the quantity of advertising is one-dimensional which parametrically shifts out the demand curve. By contrast, we are interested in how targeting \emph{per se} shapes prices and welfare which motivates our richer model of flexible demand manipulation. This distinction is depicted in second row of \cref{table:taxonomy}. Our framework allows us to refine and expand on old intuitions \citep{dorfman1954optimal,dixit1978advertising} that persuasive advertising can raise prices: by characterizing the form of optimal advertising maps, we make precise how the additional flexibility wielded by a designer seeking to maximize profits exacerbates losses to consumer welfare.\footnote{An additional benefit is that the flexibility allows for transparent analyses of how advertising shapes demand. For instance, the parametric approaches of \cite{dorfman1954optimal} show that incentives to increase advertisements increase when demand is more elastic while \cite{becker1993simple} call their analysis \emph{`highly misleading'}, showing that the opposite can obtain.} Moreover, our analysis shows that consumer harm is not inevitable: we also show how flexibility can be harnessed to improve consumer welfare.

\begin{table}[htbp]
    \centering
    \begin{tabular}{|cc|c|c|c|}
    \hline
       \parbox{3.5em}{\small }&\small \diagbox[width=12em]{Instrument}{Degree of \\ Flexibility} 
       & \small Parametric 
       & \small Flexible 
       \\
       \hline
       \multicolumn{2}{|c|}{\small Informative} & \small \makecell{ \cite{lewis1994supplying} \\[2pt] \cite{johnson2006simple}}  & \small \makecell{\cite{roesler2017buyer}} \\
       \hline
       \multicolumn{2}{|c|}{\small Persuasive} &  \small \makecell{
       \cite{dorfman1954optimal}
       \\[2pt] \cite{dixit1978advertising}
       } 
       & \small \makecell{ \textbf{This paper}} \\
       \hline
       \multicolumn{2}{|c|}{\small Info +  Persuasive} &  \small \makecell{ $\emptyset$
       } 
       & \small \makecell{ \textbf{This paper}} \\
       \hline
    \end{tabular}
    \caption{Taxonomy of how advertising shapes demand}
    \label{table:taxonomy}
\end{table}

Our analysis takes seriously the classic view of advertising as persuasive. This might function via behavioral manipulation, \emph{`playing upon the mind with studied skill'} \citep{robinson1933economics}, or as a complementary good \citep{becker1993simple} in which marginal utility of consumption increases in the level of advertisements.\footnote{These differing mechanisms have distinctive welfare implications so we are distantly related to the literature on behavioral welfare economics from which we sometimes borrow language.} This persuasive view of advertising stands in contrast to more recent literature modeling advertising as informational.\footnote{For instance, advertising might make consumers aware of a firm's existence and/or prices \citep[among many others]{ozga1960imperfect, stigler1961economics,nelson1974advertising, butters1978equilibrium,
fudenberg1984fat, roy2000strategic}, signal quality \citep{milgrom1986price,bagwell1988advertising}, or directly convey information. \cite{bagwell2007economic} offers an excellent survey. See also recent work by \cite{ichihashi2024buyer} for an analysis of recommendation algorithms that simultaneously informs a consumer about existence of a product as well as their value for it.}  Within this literature, we relate most closely to work studying how the demand curve is shaped by consumer information. When information is parametric, the influential analysis of \cite{johnson2006simple} study `demand rotations' and their corresponding impact on prices. When information is flexibly provided, \cite{roesler2017buyer} elegantly characterize how the demand curve can be shaped to drive down prices and maximize consumer surplus.\footnote{See also \cite*{bergemann2025screening} for an analysis of flexible information design to shape valuations, and the monopolist is screening with a menu.} In this regard, we view our analysis of flexible persuasive advertising as a natural complement to recent interest in flexible information design.\footnote{\cite{johnson2006simple} note: \emph{``traditional notions [persuasive advertising] correspond operationally to what we call pure hype, and change the position of a firm’s demand curve, whereas what we call real information rotates demand, thereby changing its shape.''} We offer a more nuanced view since flexible persuasive advertising changes \emph{both} position and shape.} At the same time, our analysis is distinct and delivers new economic insights. Moreover, to better understand how these instruments interact, we analyze how informative and persuasive advertising should be jointly and flexibly designed to maximize producer or consumer welfare (last row of \cref{table:taxonomy}).

On a technical level, we solve the problem of jointly choosing (i) the target demand curve which, via the monopolist's pricing problem, shapes market outcomes; and (ii) the least-cost advertising plan to implement it. In the former, the monopolist is not mean-constrained which reflects a distinctive feature of flexible manipulation vis-a-vis information design.\footnote{The flexibility our monopolist has to manipulate valuations is reminiscent of work in political science literature on vote buying \citep{groseclose1996buying}. The key differences are that (i) our environment features prices which is set endogenously by a profit-maximizing monopolist; (ii) our manipulation costs are nonlinear which might reflect variable responsiveness to advertising; and (iii) we focus on general economic objectives---increasing in consumer and producer surplus---with no counterpart in the literature on vote buying.} The latter means that the choice of distribution depends on manipulation costs and, since the objective function is non-convex, the optimal target distribution need not be an extreme point \citep{yang2023economics}.\footnote{Since transportation costs are convex (submodular), our objective function is non-convex and in general, results on the extreme points of first-order stochastic dominance do not seem to apply.} To handle these features, we combine (i) the familiar logic of monopoly pricing; with (ii) a functional first-order approach; and (iii) basic results from optimal transportation to pin down optimal advertising plans.\footnote{See \cite*{kolotilin2025persuasion} for information design via optimal transportation.}
Our functional first-order approach is loosely related to \cite*{georgiadis2022flexible} who study a principal-agent problem in which the agent can flexibly choose the distribution over outputs. Our setting differs in several substantive ways, the chief of which is that our designer's choice of the target distribution is chosen \emph{before} prices are set, and must take the monopolist's incentives into account. Moreover, the original demand curve is a key primitive of our environment and shapes the form of optimal advertising plans; in \cite*{georgiadis2022flexible}, the agent can choose any distribution over outputs.\footnote{In \cite*{georgiadis2022flexible}, the principal first chooses a contract then the agent flexibly chooses the output distribution. Hence, their agent solves an unconstrained optimization problem; in our setting, the designer (i) must accommodate the monopolist's pricing incentives; and (ii) faces a directional constraint on transportation plans. Furthermore, the cost of each advertising plan depends on the initial distribution (a primitive).} 

\textbf{Organization.} The rest of the paper is organized as follows. \cref{sec:model} details the model. \cref{sec:manipulation} characterizes optimal advertising plans under ex-ante welfare. \cref{sec:complements} characterizes optimal plans under ex-post welfare. \cref{sec:comparison} quantifies the value and welfare implications of flexible demand manipulation. \cref{sec:joint} analyzes how flexible manipulation and informative advertising should be jointly deployed. \cref{sec:extensions} outlines how our framework can be readily extended to analyze regulation, intermediation, uncertainty over welfare measures, heterogeneous susceptibility, and derive robust welfare bounds. \cref{sec:discussion} concludes.

\section{Model}\label{sec:model}
\textbf{Initial valuations.} There is a unit measure of consumers, whose initial valuations $x$ are drawn from an atomless probability measure $\mu \in \Delta(\Re_+)$, supported on a closed interval $X\subseteq \Re_+$ which may be bounded or unbounded. Let $F$ be the corresponding CDF and $f$ be the density. We assume that $f > 0$ in the interior of $X$ and that the expected valuation is finite. Let $p^M$ denote the lowest optimal price for the monopolist facing initial valuations. We assume $p^M > \min X$.

\textbf{Notation.} 
For the probability measure $\nu$, denote its CDF with $F_{\nu}$ and denote the survival function with $\oF_{\nu}(x) \coloneqq \nu([x,\infty))$, noting that the inclusion of $x$ will reflect our assumption that consumers break indifferences in favor of buying. Given a joint distribution $\pi \in \Delta(\Re_+^2)$, we denote the marginal over the first coordinate by $\marg_x \pi$ and the other marginal by $\marg_y \pi$. We say $\pi$ is induced by the map $T$ on $\Re_+$ if $\pi(A \times B) = \mu(A \cap T^{-1}(B))$ for measurable sets $A,B\subseteq \Re_+$.

\textbf{Advertising as transportation.} We consider the setting where a designer can perfectly target consumers with different initial valuations, increasing individual valuations through advertising at some cost. That is, the designer chooses an \emph{advertising plan}, which is  a joint distribution $\pi \in \Delta(\Re^2_+)$: the first marginal is the \emph{initial distribution} i.e., $\text{marg}_x \pi = \mu$; the second marginal is the \emph{target distribution} the monopolist faces. Consumers make decisions independently, so the cost of persuading a consumer only depends on their initial and final valuation.\footnote{This implicitly assumes that there are no economies of scale in advertising---that is, the total cost is simply the sum of individual costs. This corresponds to a world in which the marginal cost of personalizing advertisements based on user data is near-zero (see, for instance, the rise of AI copywriting tools), and the bulk of the cost of advertising is from reaching consumers.} 

Individual costs are denoted by the function $c:\Re^2_+ \to \Re_+$, and the cost of an advertising plan $\pi$ is the aggregate of individual costs $\int c(x,y) d\pi$. We maintain throughout the following assumptions on $c$.
\begin{assumption} \label{main assumption}
    The cost function $c$ satisfies:
    \begin{itemize}[nosep]
        \item [(i)] \textbf{Smoothness.} Twice continuously differentiable and $c(x,x) = 0$ for each $x \geq 0$.  
        \item [(ii)] \textbf{Inada condition.} 
        \begin{itemize}
            \item[a.] for every $x \geq 0$, $c_y(x,x) = 0$
            \item [b.] $c_{yy} >0$ on $\Re_+^2$; and 
            \item[c.] $\underset{M \to \infty}{\lim} \underset{x}{\inf \,}
            c_y(x,x+M) = +\infty$.
        \end{itemize}
        \item [(iii)]\textbf{Submodularity.} Strictly submodular. 
    \end{itemize}
\end{assumption}

Assumption (i) is a regularity condition; Assumption (ii) imposes an `Inada condition': the marginal cost from increasing valuations is initially zero, but grows unboundedly with the distance transported. Assumption (iii) is our most substantive assumption, stating that the marginal cost of increasing the final valuation is decreasing in the initial valuation. This is fulfilled, for instance, if the cost of scaling up a consumer's initial valuation by some fixed proportion is the same across consumer types; or consumers are equally susceptible to advertisements and the cost of moving type $x$ to $y > x$ is a convex function of $|y-x|$, depending only on distance moved. 

In the main text we focus on advertising plans which weakly increase valuations because we think this is the relevant case for many product markets. In Online Appendix \ref{appendix:directional} we also completely characterize optimal advertising plans without this directional constraint---their underlying logic is quite similar though they feature both `forward' and `backward' shifts.\footnote{See Theorems \ref{thrm:exante_U} and \ref{thrm:expost_U} in Online Appendix \ref{appendix:directional} that complete characterizes optimal plans without this constraint. They have an  additional \emph{twist} structure to reflect both `forward' and `backward' comonotone shifts. Relaxing the directional constraint cannot improve producer surplus; more subtly, it can sometimes improve consumer surplus because `backward shifts' can help fulfill the monopolist's upward deviation constraints. See also \cite{li2023mismatch} who studies a model where consumers are `mismatched' which can induce lower prices.} For a given distribution over initial valuations $\mu \in \Delta(X)$ and choice of target distribution $\nu \in \Delta(\Re_+)$, choosing the least-cost advertising plan is a directional optimal transportation problem \citep{nutz2022directional}. To reflect this directional constraint, we require that the advertising plan concentrates on forward shifts: $\pi(\mathbb{H}) = 1$ where $\mathbb{H} := \{(x,y) \in \Re_+^2: y \geq x\}$. Finally, we will require that the target distribution admits a finite mean. Putting these constraints together, the set of feasible advertising plans is 
\[
\Pi \coloneqq \left\{\pi \in \Delta(\Re_+^2): {\substack{\text{\normalsize (i) $\text{marg}_x \pi = \mu$;} \\
\text{\normalsize{(ii)$\pi(\mathbb{H}) = 1$;}} \\
\text{\normalsize (iii) $\int y \de \pi(x,y) < +\infty$ 
}}}\right\}.
\]

\textbf{Mechanisms for increasing valuations.} We are agnostic about the exact mechanism through which advertising shapes valuations. These will differ across products and contexts, and an extensive literature in psychology and marketing analyzes various channels; we will not attempt a comprehensive account. Nonetheless, let us briefly offer a few examples.\footnote{See \cite{armstrong2010persuasive} and references therein.}
For instance, with anchoring effects, \cite{adomavicius2013recommender} note that \emph{``viewers' preference ratings are malleable and can be significantly influenced by the recommendation received''.} Likewise, salience can be exploited e.g., by emphasizing favorable aspects of the product (see \cite{bordalo2022salience} for a survey). There is also extensive documentation of the tendency for consumers to prefer products they are familiar with via `mere exposure effects', availability biases, recognition effects.\footnote{See \cite{tversky1973availability} and  \cite{goldstein2002models}. The `mere exposure effect' has been extensively documented by \cite{fang2007examination}.} With new technologies, recent experimental evidence suggest that humans interacting with large-language models are susceptible to covert manipulation, even when \emph{``decisions are grounded in quantifiable risks and trade-offs (e.g., budget constraints, product quality)''} as noted by \cite{sabour2025human}\footnote{See also \cite{matz2024potential,schoenegger2025large}. \cite{meguellati2024good} and \cite{feizi2023online} discuss LLM-driven advertising.} who find experimental evidence that AIs tailor manipulation to \emph{``exploit context-specific vulnerabilities''} over the course of interacting with humans.  Of course, there is substantial heterogeneity in consumers' responsiveness to advertising. For instance, susceptible consumers tend to be agreeable \citep{hirsh2012personalized}, new and infrequent users \citep{blake2015consumer}, and the elderly \citep{lewis2014advertising}. Our framework and results extend quite readily to incorporate such heterogeneity; we do so in \cref{sec:extensions}.

\textbf{Welfare metric.} We associate producer surplus with profits. There is, however, some subtlety in how we define consumer welfare. We offer two welfare measures: 
\begin{itemize}
    \item[(i)] \textbf{Ex-ante consumer surplus} (denoted $CS^{A}$) is CS measured relative to valuations under the initial distribution.  That is, if the price is $p$ and the transport map is $\pi$, 
    \[
    CS^{A}(p,\pi) = \int_{\mathbb{R}_+^2}\underbrace{\mathbb{I}(y \geq p)}_{\substack{\text{Purchase}\\ \text{decision}}} \cdot \underbrace{(x - p)}_{\text{Surplus}} d\pi (x,y)
    \]
    noting that this depends on the plan $\pi$ since welfare depends on both buying decisions (final valuations) and surplus (initial valuations less price). 
    \item[(ii)]  \textbf{Ex-post consumer surplus} (denoted $CS^{P}$) is CS measured relative to valuations under the final distribution. That is, if the price is $p$ and the transport map is $\pi$,  
    \[
    CS^{P}(p,\pi) := \int_{\mathbb{R}_+^2}\mathbb{I}(y \geq p) (y - p) d\pi (x,y)
 = \int_{p}^{+\infty} (y - p) \de F_{\nu}(y) 
    \]
noting that this depends only on the target distribution $\nu$. 
\end{itemize}
The appropriate welfare measure will depend on how advertising is interpreted.\footnote{We emphasize that under either measure of consumer welfare, we can remain agnostic about what is constitutive of welfare; see \cite{hausman2011preference} for a summary of some philosophical positions.} 
Ex-ante consumer surplus reflects \emph{advertising as behavioral manipulation}: it creates a wedge between choices and preferences, and the modifier `behavioral' reflects this wedge.\footnote{The modifier `behavioral' is also in the spirit of behavioral welfare economics that accommodates the possibility that agents make mistakes. Within the taxonomy of \cite{bernheim2016good}, the ex-ante view violates the premise from `standard' welfare economics that `Each individual's preferences determine their choices: when they choose, they seek the greatest benefit'.}  Under this view, consumers transported from $x$ to $y$ choose to purchase at price $p$ with $x < p < y$, receiving negative surplus although they are better off not purchasing. On the other hand, the formulation of ex-post consumer surplus is consistent with the view of \emph{advertising as a complementary good} \citep{stigler1977gustibus,becker1993simple} since advertising increases consumers' `true' valuations. This is also in line with the revealed preference approach to measuring consumer welfare. 

\textbf{Timing.} We will suppose throughout that consumers break ties in favor of buying and the monopolist breaks ties in favor of the lowest price.\footnote{Implementation can, as usual, be made strict via an appropriate perturbation of the plan.} The timing of the game is illustrated in \cref{fig:timing} below: 

\begin{figure}[h]
\vspace{-0.8em}
\centering
    {\includegraphics[width=0.75\textwidth]{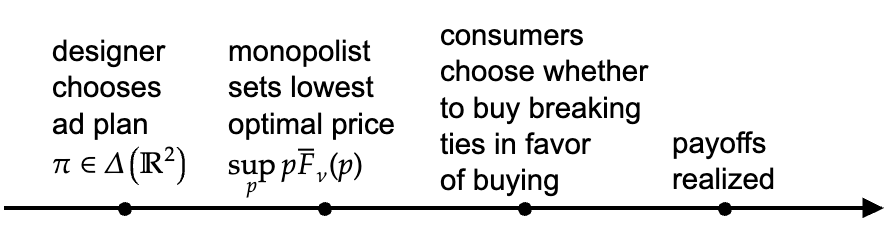}}
    \caption{Timing}\label{fig:timing}
\end{figure}

\textbf{Designer's problem.} For a given advertising plan $\pi$, the designer's objective function is $\phi\big(PS,CS,C\big)$, which is differentiable with bounded partial derivatives and weakly increasing in PS and CS: $\frac{\partial \phi}{\partial PS}, \frac{\partial \phi}{\partial CS} \geq 0$. Furthermore, $\phi$ is strictly decreasing in total cost $C(\pi) \coloneqq \int c(x,y) \de \pi$ with partial derivative bounded away from zero: $\frac{\partial \phi}{\partial C} < \ell < 0$ for some $\ell < 0$. The designer's problems under ex-ante and ex-post measures of welfare are: 
\begin{align*} \label{eq: designer_exante}
    \sup_{\pi \in \Pi} \phi \Big( PS(\pi), CS^A(\pi), C(\pi) \Big)
    \tag{D: ex-ante} 
\end{align*}
and 
\begin{align*} \label{eq: designer_expost}
    \sup_{\pi \in \Pi} \phi \Big( PS(\pi), CS^P(\pi), C(\pi) \Big).
    \tag{D: ex-post} 
\end{align*}

We will fully characterize the plans solving \eqref{eq: designer_exante} and \eqref{eq: designer_expost}. Our objective function $\phi$ is quite general and nests the standard case where the designer values a convex combination of $CS$ and $PS$: 
        \[
        \phi 
        = \underbrace{\alpha \cdot CS(\pi)  + (1-\alpha) \cdot PS(\pi)}_{\substack{\text{Weighted combination} \\  \text{of $CS$ and $PS$}}} - \underbrace{C(\pi)}_{\substack{\text{Total}\\ \text{cost}}} \quad \text{with weight $\alpha \in [0,1]$.}
        \]
We have chosen to work with general objective functions such as to speak to a wide array of economic environments. For instance, objectives which are multiplicative in consumer and producer surplus arise naturally in the context of platform-intermediated targeted advertising with endogenous market thickness (see \cref{sec:extensions}). Another example is when the designer faces non-Bayesian uncertainty about the appropriate welfare weight to place on consumer vs producer surplus.\footnote{For instance, maxmin preferences, or the multiplier preferences.}

\textbf{Discussion.} We briefly discuss our model. 
\begin{itemize}[leftmargin = 1em, nosep]
   \item \textbf{Interpreting the cost function.} We have taken a reduced form approach to modeling how advertising shapes valuations. The cost $c(x,y)$ might, for instance, capture the cost-per-unit-time of accessing a consumer scaled by the responsiveness of valuations to ads. As mentioned, we are agnostic about the exact channel through which valuations are shaped and our analysis extends quite readily to handle heterogeneity in susceptibility to manipulation. 
  \item  \textbf{Certain vs conditional costs.} In our model the cost of advertising is borne regardless of whether it leads to a purchase. One might, instead, suppose that such costs are paid only if a consumer purchases, or if consumers click through to the firm's website (as with `pay-per-click' advertising). It will turn out that flexibility ensures that consumers who are advertised to must---in equilibrium---purchase the good, so our formulation is consistent with such alternate interpretations.\footnote{We are grateful to Alessandro Bonatti for pointing this out.}
  \item \textbf{Atomless distribution of initial valuations.} We have assumed the distribution of initial valuations $\mu$ is atomless simply for expositional simplicity since advertising plans are then induced by deterministic (Monge) maps. With atoms, all our results hold essentially unchanged, though advertising plans must now be (more inconveniently) stated in terms of the joint distribution $\pi$ since mass may need to be split via randomization.
   \item \textbf{Uniform pricing.} Our designer is perfectly informed about consumer types. This raises the question of whether the monopolist is similarly informed and, if so, why has it not engaged in price discrimination? One possibility is that the monopolist is simply uninformed, as in the case where the designer is a large platform with detailed information about consumers' preferences while the monopolist is simply a seller. A second possibility is that the monopolist is unable to price discriminate because of reputational or other concerns.\footnote{\cite*{bergemann2025screening} offer a detailed discussion of this assumption.} Whatever the reason, personalized advertisements with uniform pricing remain the prevailing state of affairs in both on- and offline marketplaces.\footnote{In offline marketplaces, uniform pricing remains the norm. In online marketplaces, personalized pricing seems---at least for now---relatively scarce. \cite{cavallo2016billion} show a remarkable degree of similarity between online and offline prices. On the other hand, targeted online advertising is ubiquitous: US spending on online advertising reached $\$271$ Billion USD in 2022, roughly triple from $\ 91$ Billion USD in 2017. See also \cite{goldfarb2014different} for a discussion of how online advertising substantially reduces the cost of targeting ads.} 
   \item  \textbf{Interpreting the designer.} Our designer can be interpreted as an intermediary e.g., platform facilitating trade between sellers on one side, and consumers on the other. Platforms have incentive to balance the surplus on both sides of the market: on the one hand, they typically charge a commission on sales;\footnote{For instance, Amazon typically charges a differential commission based on product category, ranging from 1\% for grocery to up to 20\% for Amazon games.} on the other, they have incentive to deliver consumer surplus since consumers might have the option to buy off-platform; \cref{ext:intermediated} develops a simple model of intermediated targeted advertising and shows how our results apply there.\footnote{See also \cite{bergemann2024data} for a model of on-vs-off platform interaction. \cite*{elliott2022matching} characterizes all welfare outcomes a platform can implement via controlling access and information. Persuasive advertising is absent in both papers.} A more abstract interpretation is our results help understand the limits of how flexible demand manipulation can shape both total welfare as well the surplus split;\footnote{This is the same spirit as \cite*{bergemann2015limits}. More recently, in the context of digital advertising, \cite*{bergemann2022optimal} study both seller-optimal and bidder-optimal structures, where the seller is the platform and the bidders are firms wishing to access the consumer. In our setting there is a cost function so the set of $(PS,CS)$ pairs are technically unbounded. But for a \emph{fixed} expenditure on advertising, our results offer a simple recipe to trace out feasible combinations of producer- and consumer-surplus.} \cref{ext:robust_welfare} describes how our framework can be extended to derive robust welfare bounds from observed data in the presence of manipulation technologies. 
\end{itemize}

\section{Advertising as behavioral manipulation} \label{sec:manipulation}
We start by considering the view of advertising as behavioral manipulation. This reflects an old disquiet about advertising: \emph{``If the individual’s wants are to be urgent. ...they must be original with himself. They cannot be urgent if they must be contrived for him.''} \citep{galbraith1958affluent}. Under this view, advertising drives a {wedge between preference and choice} (see, e.g. \cite{bernheim2009beyond,bernheim2009behavioral,bernheim2016good}): a consumer's purchase decision depends on manipulated preferences, but her welfare depends on initial preferences.\footnote{We are agnostic about whether this wedge is driven by consumers' failure to purchase optimally (`technical ability') or understand the product (`categorization failure'); see \citep{bernheim2018behavioral} for a taxonomy.} This wedge might be driven by the proliferation of new technologies for manipulating consumer preferences, leading the analyst to suspect 
that purchase decisions are \emph{not welfare-revealing}.\footnote{Our notion is also related to ``glossiness'' introduced by \cite*{acemoglu2025big} in which a seller can lead consumers to misperceive the true quality of the product. However, we allow the designer to manipulate willingness-to-pay in a completely flexible manner.} We call this notion of welfare ex-ante. Our first result pins down the general form of optimal advertising plans. 

\begin{theorem} \label{thrm:exante_general}
The designer's problem under ex-ante welfare measure \eqref{eq: designer_exante} is solved by a deterministic advertising plan $\pi^* \in \Pi$ which is induced by the following map with an \emph{intermediate interval} structure: 
    \begin{align*}
        T(x) = \begin{cases}
        x \quad &\text{if $x \geq p^*$} \\
        p^*  \quad &\text{if $x \in [\underline p, p^*)$} \\ 
        x \quad &\text{if $x < \underline p$,}
        \end{cases}
    \end{align*}    
    for some $p^* \geq \underline p$, where $T(x)$ is the valuation type $x$ is transported to. Moreover, every designer-optimal advertising plan takes an intermediate interval structure.
\end{theorem}

\begin{figure}[h]
\centering
\caption{Form of optimal advertising plans (ex-ante)}
    \subfloat[CDF]{\includegraphics[width=0.45\textwidth]{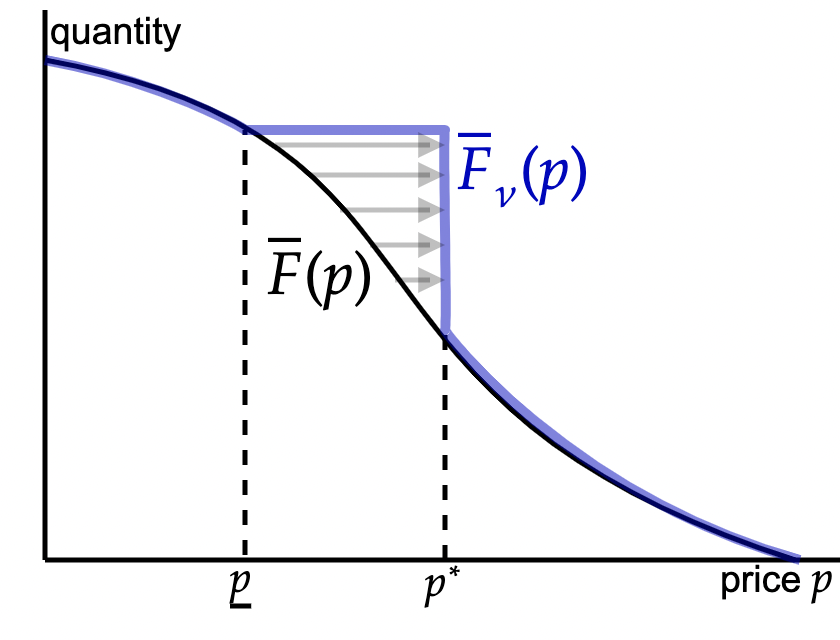}}
    \subfloat[Advertising map]{\includegraphics[width=0.45\textwidth]{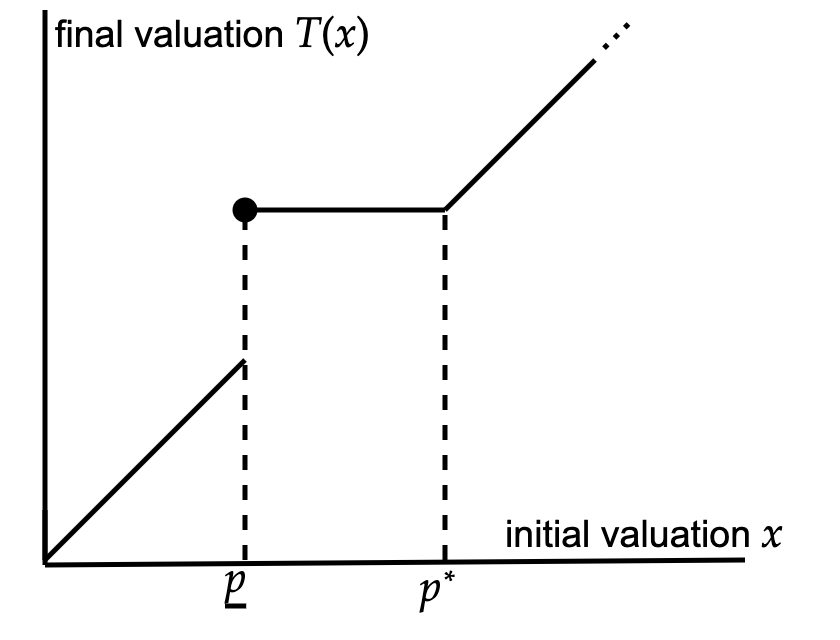}}
    \label{fig:exante}
\end{figure}

\cref{thrm:exante_general} is illustrated in \cref{fig:exante} which depicts the inverse demand function (panel (a)) and advertising map (panel (b)). Such plans consist of a \emph{target price} $p^*$ which the designer aims to implement, and a \emph{target quantity} $\oF(\underline p)$. The designer then (deterministically) transports all types in $[\underline p, p^*)$ to $p^*$ to shore up demand.\footnote{In the language of optimal transportation, the solution is induced by a deterministic `Monge map' without splitting mass.}

\begin{figure}[h] 
\centering
\caption{Intuition for intermediate interval maps}
    \subfloat[Forbidden maps]{\includegraphics[width=0.47\textwidth]{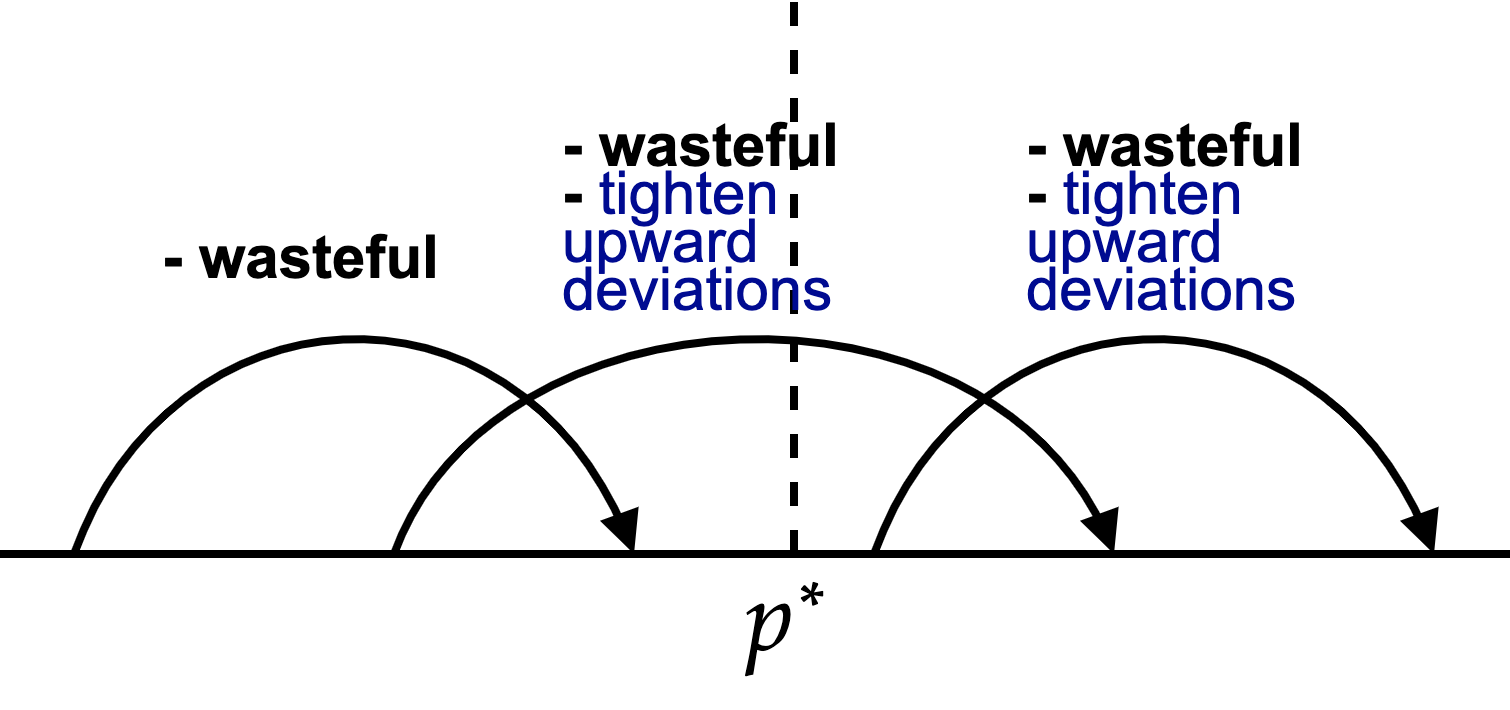}} \hfill 
    \subfloat[Greedy is better]{\includegraphics[width=0.47\textwidth]{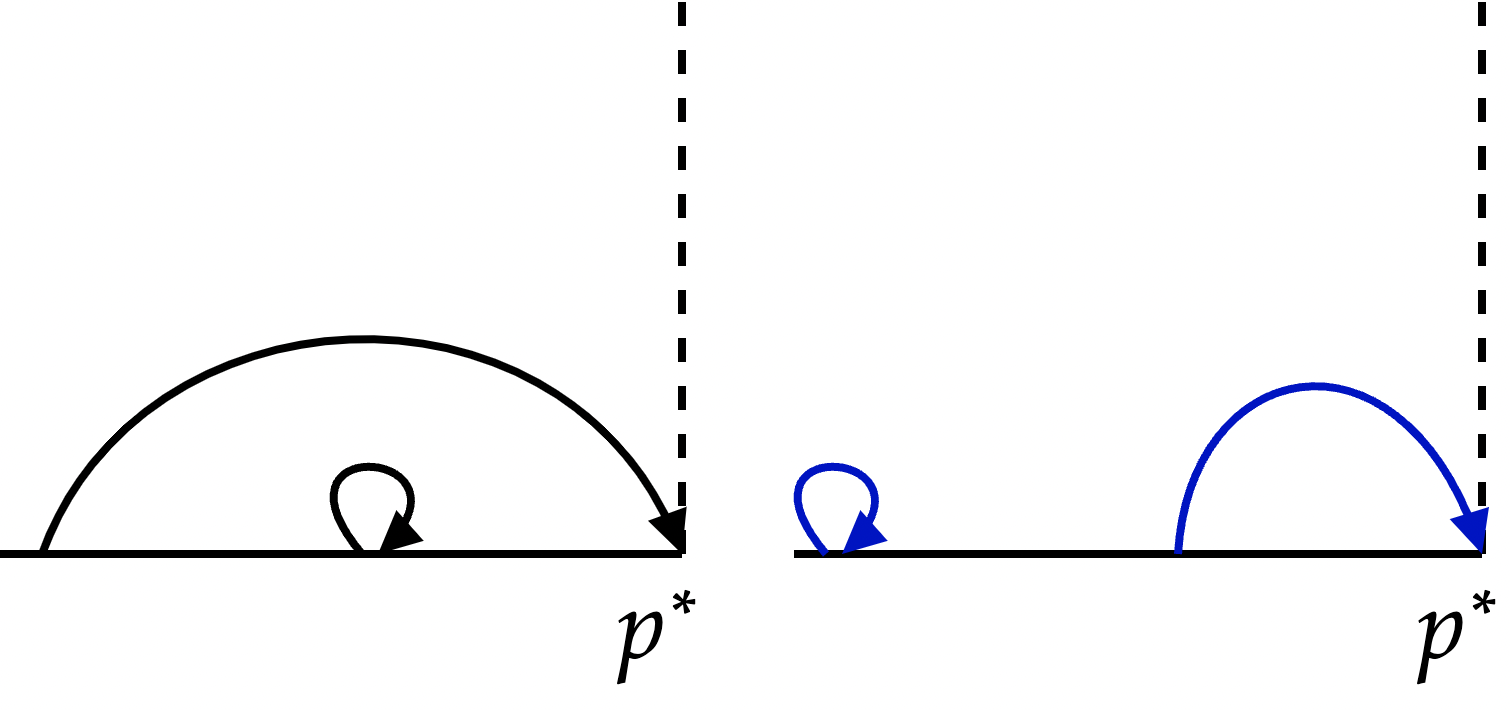}}
    \label{fig:forbidden_maps}
\end{figure}

\cref{fig:forbidden_maps} illustrates the simple reason why optimal maps must take this form. Panel (a) depicts maps from $x$ to $y$ which must be suboptimal i.e., if an advertising plan assigns positive probability to such maps, it can be strictly improved. Start by considering the arrow on the left which maps valuation $x < p^*$ to $y < p^*$. Since the consumer does not purchase, advertising is entirely wasteful and the designer is better off not advertising to those consumers. Next, consider the arrow in the middle which maps valuations $x < p^*$ to $y > p^*$, and the arrow on the right which maps valuations $x > p^*$ to $y > x > p^*$. In both cases, advertising is excessive for two distinct reasons: such maps are 
\begin{enumerate}[]
    \item \emph{Wasteful}: transporting valuations from $x$ to $p^*$ if $x < p^*$ and not advertising otherwise would (i) save on advertising costs; (ii) still suffice to induce purchase at price $p^*$ ($PS$ does not decrease); and (iii) under ex-ante consumer surplus, excess willingness-to-pay does not translate to excess welfare; and 
    \item \emph{Tempting}: consumers transported to $y > p^*$ would continue to purchase the product if the monopolist deviates upward
    to price $p^* < p' \leq y$. This tightens some of the monopolist's upward deviation constraints and makes the target price $p^*$ more challenging to implement.
\end{enumerate} 
Both of these forces act in the same direction and imply that only consumers with initial valuations $x < p^*$ can be targeted and, if they are, must be transported to exactly $p^*$; inframarginal consumers are never shown advertisements. Which among consumers with valuations $x < p^*$ should be targeted? Panel (b) illustrates why the ``greedy choice'' which starts with marginal consumers with valuations just shy of the price $p^*$ and working backwards is optimal.

\cref{thrm:exante_general} establishes the \emph{form} of optimal advertising plans for general objective functions $\phi$. Nonetheless, depending on whether the designer values producer or consumer surplus, the economic intuition underlying these plans differ---and these are reflected in different target prices and quantities.

\textbf{Producer-optimal plans.} Suppose that the designer is perfectly aligned with the monopolist\footnote{For instance, the monopolist might be in control of the advertising plan and accounts for the cost of advertising.} and solves 
\[
\sup_{\pi \in \Pi} PS(\pi) - \int c(x,y) \de\pi. 
\]
We call the solution the \emph{producer-optimal plan}. A first observation is that this problem is equivalent to the monopolist choosing both $p^*$ and the advertising plan $\pi$ directly. Observe that holding fixed the target distribution, the monopolist faces the standard tradeoff between her extensive margin (selling to fewer customers) with her intensive margin (selling to existing consumers at higher prices). However, by varying the advertising plan $\pi$, the monopolist is also able to artificially shore up demand. Thus, under flexible target advertising the monopolist's intensive-extensive tradeoff is no longer local. 

This implies that at the producer-optimal, the extensive margin is determined by the \emph{marginal cost of advertising}: for a given target price $p^*$, $\underline p$ is pinned down as the highest quantity so that the contribution to producer surplus from transporting $\underline p$ to $p^*$ is exactly offset by the cost of doing so. This is depicted in \cref{fig:CS_PS_optimal} (a) and is formalized by the following proposition. 

\begin{proposition}[Producer-optimal] \label{prop: produceroptimal}
    The producer-optimal plan is of the intermediate interval form as in \cref{thrm:exante_general}, where the target price $p^*$ solves
\begin{align*}
    \max_{p \geq 0} &\Bigg[ p \overline F\Big(\underline p(p)\Big) - \int_{\underline p(p)}^{p}c(x,p) \de F(x)\Bigg] \\   
&\text{with } \underline p(p) := \inf\Big\{x \in X: c(x,p) \leq p\Big\}.
\end{align*}
\end{proposition}

We will subsequently use \cref{prop: produceroptimal} to show that a monopolist equipped with the ability to flexibly manipulate the demand curve substantially drives up prices relative to uniform advertising. 

\begin{figure}[h]
\centering
\caption{Producer and consumer-optimal plans}
    \subfloat[PS-optimal]{\includegraphics[width=0.45\textwidth]{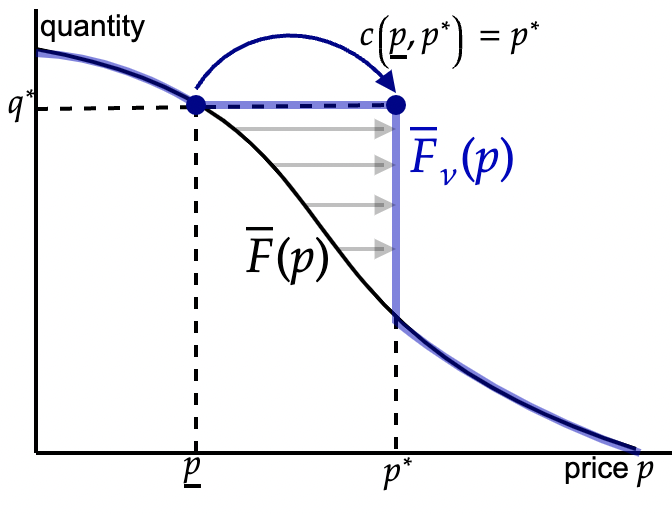}}
    \subfloat[CS-optimal (ex-ante)]{\includegraphics[width=0.45\textwidth]{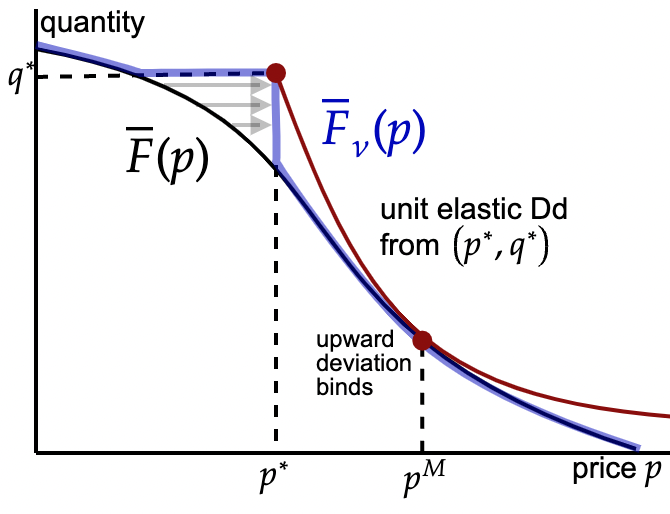}}
    \label{fig:CS_PS_optimal}
\end{figure}

\textbf{Consumer-optimal (ex-ante) plans.} Now suppose, instead, that the designer seeks to maximize consumer surplus less cost of advertising and so solves
\[
\sup_{\pi \in \Pi} CS^A(\pi) - \int c(x,y) \de\pi. 
\]
We call the solution to this problem the ex-ante consumer-optimal plan. This problem is a little more subtle: to deliver consumer surplus, the designer aims to implement a low target price $p^*$. But different from the producer-optimal case, the monopolist's incentives must be taken into account: for $p^*$ to be implementable under the new demand curve, the monopolist must not be tempted to deviate upwards to a higher price; we call these upward deviation constraints.\footnote{The monopolist must technically be deterred from deviating both up and downwards---but only upward deviations will be salient since we are interested in maximizing consumer surplus.}

The core tradeoff here is that to implement the lower target price $p^*$, the monopolist shores up demand which is associated with (i) the costs of advertising; and (ii) welfare losses from transported consumers who are manipulated into buying against their interest and so receive negative surplus. This suggests a one-dimensional characterization of the consumer-optimal plan: for a given target price $p^*$, $\underline p$ is pinned down as the highest threshold that can implement $p^*$. That is, different from the producer-optimal case in which the additional extensive margin was determined by the marginal costs of shaping preferences, here it is determined by the {monopolist's upward deviation constraints}. This is depicted in \cref{fig:CS_PS_optimal} (b) and is formalized as follows:

\begin{proposition}[Ex-ante consumer-optimal] \label{prop: exanteconsumeroptimal}
    The ex-ante consumer-optimal plan is solved by 
\begin{align*}
\max_{p \in [\min X, p^M]} &\Bigg[ \underbrace{ \int_{p}^{+\infty} (x-p) \de F(x)}_{\text{Positive part of $CS^A$}} - \underbrace{\int_{\underline p}^{p} \Big( c(x,p) + (p - x) \Big) \de F(x)}_{\text{Transportation and manipulation cost}} \Bigg]
\\   
&\text{where } \underbrace{\underline p(p) = \oF^{-1}\Big(\frac{p^M\oF(p^M)}{p}\Big)}_{\text{Upward deviation to $p^M$ binds}}
\end{align*}
with optimal price $p^* < p^M$ i.e., the monopolist's price strictly decreases.
\end{proposition}

Why does the upward deviation to the original optimal price $p^M$ bind? Observe that from \cref{thrm:exante_general} optimal advertising plans have intermediate interval structures. Recall that such plans leave the demand curve after the target price $p^*$ unchanged, so the monopolist's upward deviation constraint can be written as: 
\[
p\oF(\underline p) \geq \max_{p'\geq p} p' \cdot \oF(p')  = p^M \cdot \oF(p^M). 
\]
which is a sufficient condition to implement $p^*$. But it must also be an equality: if not, since $\oF$ is continuous the designer can increase $\underline p$ a little which saves on total costs (both transportation and manipulation) while continuing to implement $p^*$. In this way, \cref{prop: exanteconsumeroptimal} delivers a computationally simple procedure for solving for the consumer-optimal advertising plan. Moreover, it holds generically that only this upward deviation binds. This is in contrast to recent results on consumer-optimal information design which find that consumer-optimal information structures typically keep the monopolist indifferent between all upward deviations.\footnote{See, e.g., \cite*{bergemann2015limits,roesler2017buyer}.} This difference arises precisely because under the view of advertising as manipulation, if $p^*$ is already implementable, additional shifts of inframarginal consumers to further tighten the monopolist's upward deviation incentives cannot improve ex-ante consumer welfare. We revisit this distinction in \cref{sec:joint} where we study the joint provision of informative and persuasive advertising.

Finally, observe that under the consumer-optimal advertising plan, advertising improves ex-ante consumer welfare in the aggregate, but differentially across consumers. In particular, to induce a lower target price $p^*$ to benefit inframarginal consumers, intermediate consumers must be manipulated to purchase the good against their own interests.\footnote{See \cite*{barreto2022price} for an analysis of (informational) market segmentation for price discrimination under distributional concerns.} Our analysis highlights how consumer-optimal advertising plans are both \emph{powerful}, delivering substantially more surplus in the aggregate. but are also \emph{necessarily regressive} with intermediate consumers made strictly worse-off than under their outside option.  

\textbf{Efficiency-maximizing plans.} Now suppose that the designer is interested in maximizing a convex combination of producer and ex-ante consumer surplus: 
\[
\sup_{\pi \in \Pi} \bigg( \alpha CS(\pi) + (1-\alpha) PS(\pi) - \int c d\pi \bigg) \quad \text{for $\alpha \in [0,1]$.}
\]
$\alpha = 0$ corresponds to the producer optimal plan where the bottom of the intermediate interval $\underline p$ is pinned down by a first-order condition equating the marginal cost of advertising to the marginal increase in product surplus. $\alpha = 1$ corresponds to the consumer-optimal plan where $\underline p$ is pinned down by the minimal extra demand to deter the monopolist's upward deviations. These constraints are economically distinct and, when the designer has an interior welfare weight, can be combined via a simple complementary slackness condition illustrated by \cref{fig:comp_slackness}.

\begin{figure}[h!]
    \centering
      \caption{Complementary slackness} 
    \vspace{-0.5em}
    \subfloat[Slack]{\includegraphics[width=0.4\textwidth]{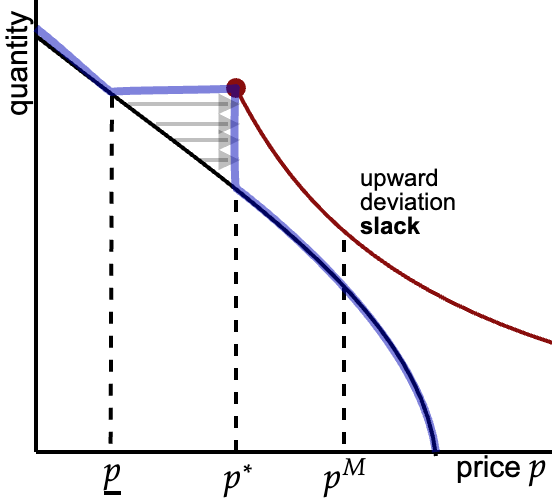}} 
    \subfloat[Binding]{\includegraphics[width=0.4\textwidth]{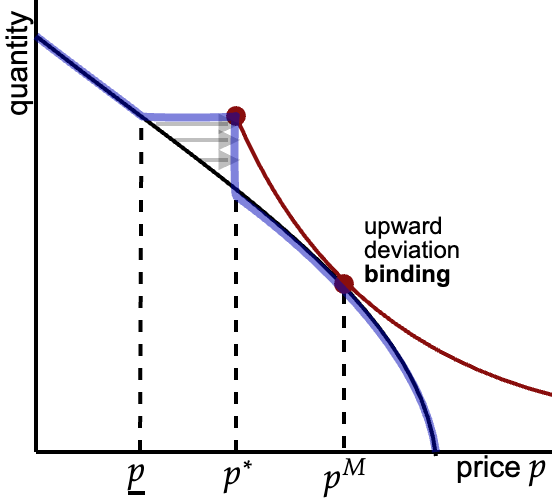}}     
    \label{fig:comp_slackness}
    \vspace{-1em}
\end{figure}

For a candidate price $p^*$, if all upward deviations are slack (panel (a)), then $\underline p$ is pinned down by the first-order condition as reflected by the \textcolor{blue}{blue} signs in the following equations:  
\begin{align*}
    \underbrace{\alpha (\underline p - p^*)}_{\substack{\text{Reduction in} \\ {CS^A}}} + \underbrace{(1-\alpha) \cdot p^*}_{\substack{\text{Addition to} \\ \text{PS}}}  -
    \underbrace{c(\underline p,p^*)}_{\substack{\text{Cost of} \\ \text{$\underline p \to p^*$}}} \underset{(\textcolor{red}{\leq})}{\textcolor{blue}{=}} 0 
    \quad \text{and} \quad 
    \underbrace{p^M \oF(p^M) - p^* \oF(\underline p)}_{\text{Best upward deviation}}  \underset{(\textcolor{red}{=})}{\textcolor{blue}{<}}  0
\end{align*}
On the other hand, if upward deviations bind (panel (b)) then $\underline p$ must satisfy the above equations with \textcolor{red}{red} signs. It is immediate to verify that in the producer-optimal case $(\alpha = 0)$ the first equation binds, and in the consumer-optimal case $(\alpha = 1)$ the second binds.

\section{Advertising as complements} \label{sec:complements}
We now consider an alternative view that advertising is simply a complementary good that delivers additional welfare to consumers \citep{becker1993simple}. Under this view, a consumer of type $x$ transported to $y \geq p$ enjoys a surplus of $y-p \geq 0$ from purchasing the monopolist's product. 
This view is consistent with the revealed preference approach so empirical estimates of the effect of advertising on consumer welfare would measure $CS^P$.

We start with a simple but important definition which tracks marginal contributions to the designer's objective $\phi$ as evaluated at a fixed advertising plan $\pi$: 

\begin{definition}[Locally-greedy maps]
    Say $\Lambda_{\pi}$ is the \emph{locally-greedy} transportation map at $\pi$ if for each type $x \in \Re_+$, 
    \[\Lambda_{\pi}(x) \coloneqq \arg\max_{y\geq x} \, \underbrace{\mathbb{I}\big[y\geq p\big] \cdot \Bigg( p \cdot \cfrac{\partial \phi}{\partial PS}\bigg|_{\pi} + (y-p) \cdot \cfrac{\partial \phi}{\partial CS^P}\bigg|_{\pi} \Bigg) + c(x,y)\cdot \cfrac{\partial \phi}{\partial C}\bigg|_{\pi}}_{\text{Local contribution from moving $x \to y$}}
    \]
    where $p$ is the seller's lowest optimal price under the advertising plan $\pi$.\footnote{Note that $\Lambda_{\pi}$ is technically a correspondence, but will be singleton-valued almost everywhere.} Note $\Lambda_{\pi}$ is increasing since $c$ is submodular.
\end{definition}

Locally-greedy maps are interpreted as follows. The expression within the argmax gives the contribution to the designer's objective from transporting consumer type $x$ to $y$, evaluated at the fixed plan $\pi$---this is the sense in which the contribution is \emph{local}.  $\Lambda_{\pi}(x)$ prescribes that each type $x$ should be transported to maximize its local contribution---this is the sense in which the map is \emph{greedy}.\footnote{The designer's objective $\phi$ is a function of producer surplus, consumer surplus, and advertising costs which are, themselves, functions of the advertising plan $\pi$. Thus, viewing $\phi$ as a functional on transport plans, this map arises in a functional first-order condition that follows from the chain rule for Fr\'echet derivatives. See \cref{Lem: expost linearization} in Appendix \ref{appendix:proofs_main} for details.} 

When the designer's objective is a weighted combination of producer and consumer surplus, the contribution to the designer's objective depends only on the price $p$ and not the plan $\pi$ so the locally-greedy map then simplifies to
\[ \Lambda_{\pi}(x) = \arg\max_{y\geq x} \mathbb{I}[y\geq p] \bigg( \underbrace{(1 - \alpha) 
 \cdot p}_{\substack{\text{Contribution}\\ \text{to $PS$}}} + \underbrace{\alpha \cdot (y-p)}_{\substack{\text{Contribution}\\ \text{to $CS^P$}}} \bigg) - \underbrace{c(x,y).}_{\substack{\text{Contribution}\\ \text{to cost}}}
\]
Importantly, this ignores the possibility that such improvements might alter the monopolist's pricing incentives. To this end, we now incorporate local-greediness into a class of advertising plans which do take the monopolist's pricing incentives into account: 

\begin{definition} \label{def: constrained-greedy}
    Say $\pi^* \in \Pi$ is \emph{constrained-greedy} for the \emph{target price and quantity pair} $(p^*,q^*)$, where (i) $p^*q^* \geq p^M \oF(p^M)$; and (ii) $p^*$ is the lowest optimal price under $ \pi^*$, 
    if it is induced by the following map:\footnotemark 
    \[
        T(x) = \begin{cases}
        x \quad &\text{if $x < \oF^{-1}(q^*)$,}
        \\ 
        p^*  \quad &\text{if $x \in \Big[\oF^{-1}(q^*), [\Lambda_{\pi^*}
        ]^{-1}(p^*) \Big)$} 
        \\ 
        \Lambda_{\pi^*}(x) \wedge \frac{p^*q^*}{\oF(x)}\quad &\text{if $x \geq \oF^{-1}(q^*) \vee [\Lambda_{\pi^*}]^{-1}(p^*)$.} 
        \end{cases}
    \]
\footnotetext{Here, $[\Lambda_{\pi^*}]^{-1}$ denotes the pseudo-inverse: $[\Lambda_{\pi^*}]^{-1}(x) = \inf\{t\geq 0: \Lambda_{\pi^*}(t) \geq x\}$. Note that the middle region is empty if $[\Lambda_{\pi^*}]^{-1}(p^*) < \oF^{-1}(q^*)$. Which follows from the convention that if $x < x'$ then $[x',x) = [x',+\infty) \cap (-\infty,x) = \emptyset$.}
\end{definition}
The conditions on the target price and quantity $(p^*,q^*)$ are intuitive. Part (i) requires that the monopolist's revenue is weakly larger than under no advertising---if it is violated, $(p^*,q^*)$ can never be implemented. Part (ii) requires that the target price $p^*$ is not merely optimal (guaranteed by the construction of $\pi^*$), but also the lowest optimal price; this ensures that the construction is consistent with the locally-greedy map $\Lambda_{\pi^*}$. We also note that for any admissible price and quantity pair $(p^*,q^*)$, the constrained-greedy plan is pinned down explicitly whenever the designer's objective function is a weighted combination of producer and consumer surplus (since $\Lambda_{\pi}$ does not depend on $\pi$). When the designer's objective function is more complex, the constrained-greedy plan solves a simple two-dimensional fixed point problem that can be efficiently computed.\footnote{For a pair of constants $(\partial_{CS}, \partial_C) \in \Re_+^2$, we can compute the map prescribed in \cref{def: constrained-greedy}, but substituting the constant $\partial_{CS}$ for the partial w.r.t. $CS^P$ and $\partial_C$ for the partial w.r.t. total costs $C$. This also pins down the partial w.r.t. $PS$ since $(p^*,q^*)$ is fixed. Then, the plan delivers a pair $(CS^P,C)$ of consumer surplus and total costs. Plug that into our objective $\phi$ and evaluate the partials---if it coincides with the constants $(\partial_{CS}, \partial_C)$ we are done; otherwise iterate.} 

Constrained-greedy plans split consumers by quantiles as shown in \cref{fig:constrained_greedy}:

\begin{figure}[H]
\begin{minipage}[t]{0.48\linewidth}
\begin{itemize}[leftmargin = 2em,nosep]
    \item[(i)] \textbf{Excluded.} Consumers with valuations $x < \oF^{-1}(q^*)$ are excluded. 
    \item[(ii)] \textbf{Transported to atom at $p^*$.} Consumers with intermediate valuations are transported to $p^*$ to shore up demand to $q^*$.
    \item[(iii)] \textbf{Transported to lower envelope.} Consumers with sufficiently high valuations---such that the locally-greedy map prescribes positive shifts---are transported to the lower-envelope of the locally-greedy map and the unit-elastic demand curve. 
\end{itemize}
\end{minipage}%
\hfill%
\begin{minipage}[t]{0.5\textwidth}\vspace{0pt}
\vspace{-1em}
\centering
{\includegraphics[width=\textwidth]{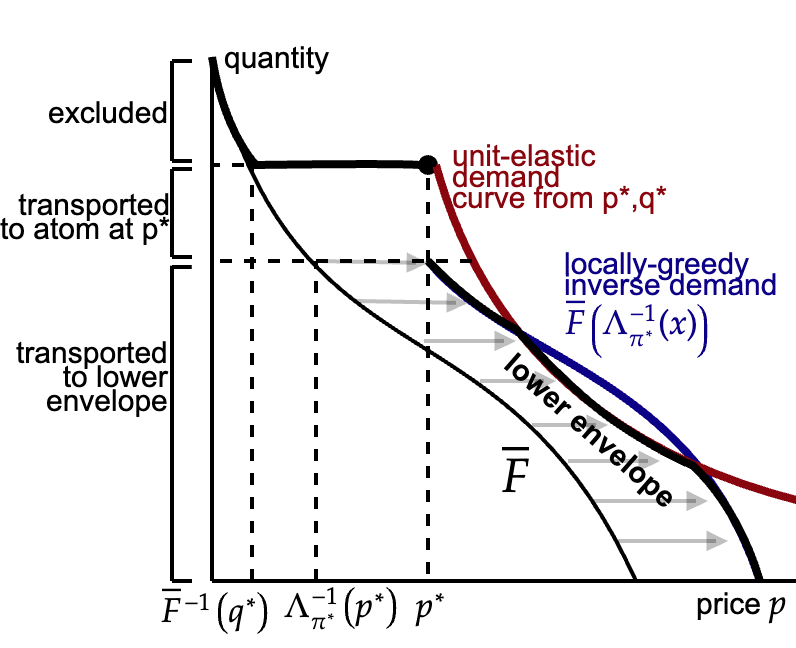}}
    \caption{Constrained-greedy plans}\label{fig:constrained_greedy}
\end{minipage}
\end{figure}

\begin{theorem}\label{thrm:expost_general}There exists a constrained-greedy solution to \eqref{eq: designer_expost}. Moreover, every solution to \eqref{eq: designer_expost} is constrained-greedy. 
\end{theorem}

\cref{thrm:expost_general} pins down the form of optimal advertising plans under ex-post measures of welfare. 
The optimal advertising plan performs a delicate balancing act to satisfy two economically distinct tradeoffs: 
\begin{enumerate}
    \item \emph{Technological.} The advertising technology governs the marginal cost of transporting consumers (raising their valuations) vis-a-vis the marginal benefit of doing so. The {locally-greedy} map $\Lambda_{\pi}$ offers a simple first-order approach to navigate this tradeoff: for each consumer $x$, we simply transport her to the pointwise optimal $y$. But this cannot always be done because this might tempt the monopolist to raise its price as reflected by the next constraint. 
    \item \emph{Incentive compatibility.} The monopolist's incentives are succinctly captured by the unit-elastic demand curve from the target price and quantity ($p^*,q^*$) which specifies an upper-bound on how many consumers can be transported beyond the valuation $p^*$ before the monopolist is tempted to raise its price. 
\end{enumerate}

\cref{thrm:expost_general} makes precise how these tradeoffs are navigated. Relative to the locally-greedy map, the need to satisfy the monopolist's incentives distorts the advertising plan \emph{upwards} for some consumer types, and \emph{downwards} for others. The upward distortion occurs for consumers in the middle quantile transported to the atom at $p^*$ though it is pointwise suboptimal---while the locally-greedy map prescribes that they should not be advertised to at all, this serves to shore up demand which slackens the monopolist's upward deviation constraints. The downward distortion occurs for consumers in the highest quantile transported to the lower-envelope of the locally-greedy map and unit-elastic demand curve whenever the latter is below the former---the locally-greedy map prescribes that valuations should be further increased but this would violate the monopolist's upward deviation constraints. 

The proof of \cref{thrm:expost_general} builds on the simple observation that, for a fixed target distribution, optimal transport plans are comonotone (positively assortative). This allows us to convert the global constraint of deterring \emph{all} upward deviations into a sequence of local constraints on the pointwise optimal advertising quantity for each consumer in the highest quantile. Then, we combine the logic of monopoly pricing (unit-elastic demand curve) with a first-order approach (locally-greedy maps) to establish that constrained-greedy plans---and only those---are optimal. But which constrained greedy plan in particular? The optimal choice of target price and quantity $(p^*,q^*)$ is driven by the following tradeoff.

\begin{figure}[H]
\begin{minipage}[t]{0.5\linewidth}
\textbf{Transportation costs vs. packing gain.} 
\cref{fig:constrained_greedy_tradeoffs} illustrates the core tradeoff. Increasing $q^*$ requires more consumers in the intermediate region to be transported to $p^*$. These consumers would, under the locally-greedy map, not be transported---and represent a loss. But, by inducing a higher quantity, this shifts up the unit-elastic demand curve which allows the designer to `pack' more inframarginal consumers from the highest quantile into the lower envelope as prescribed by \cref{thrm:expost_general}. 

\end{minipage}%
\hfill%
\begin{minipage}[t]{0.48\textwidth}\vspace{-1pt}
\vspace{-1em}
\centering
{\includegraphics[width=0.85\textwidth]{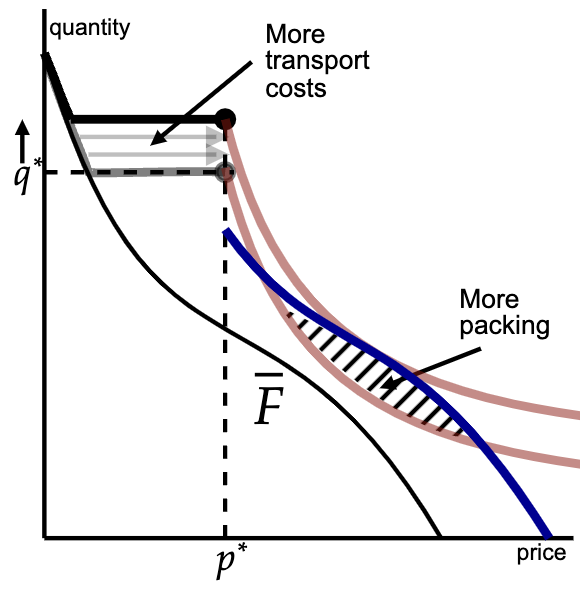}}
    \caption{Increase in `packing volume'}\label{fig:constrained_greedy_tradeoffs}
\end{minipage}
\end{figure}

Finally, we use \cref{thrm:expost_general} to derive the consumer-optimal plans under ex-post welfare from \cref{eg:uniform_eg} in the introduction. 

\begin{example}[Consumer-optimal plans under ex-post welfare] \label{eg:expost_uniform_eg}
    Consumer valuations are initially distributed uniformly on $[0,1]$ and the transport cost is quadratic in distance: $c(x,y) = \frac{a}{2}(x-y)^2$ for some $a > 0$. It is straightforward to compute the locally greedy map given target price $p^*$:\footnote{For any $\pi^*$ with lowest optimal price $p^*$, $\Lambda_{\pi^*}(x) = \text{argmax}_{y \geq x} \mathbb{I}(y \geq p^*)\cdot (y - p^*) - \frac{a}{2}(y-x)^2$ which depends only on $p^*$.} 
\[
    \Lambda_{\pi^*}(x) = 
    \begin{cases}
        x + \frac{1}{a} \quad &\text{if $x \geq p^* - \frac{1}{2a}$} \\
        x &\text{otherwise} \\
    \end{cases}
    \]

and the associated demand curve induced by this greedy map is depicted in Panel (a) of \cref{fig:expost_example_full} where we set $a = 4$ so this coincides with \cref{eg:uniform_eg} in the introduction. For each target price $p^*$, \cref{thrm:expost_general} pins down the value of the optimal advertising plan which implements $p^*$. Then, optimizing over $p^*$, the optimal target price and quantity are  $p^* = 1/4$, $q^* = 1$. 

\begin{figure}[h]
\centering
\caption{Construction $CS^P$-optimal plans in \cref{eg:expost_uniform_eg} $(a = 4)$}
    \subfloat[Locally-greedy (\cref{eg:expost_uniform_eg})]{\includegraphics[width=0.49\textwidth]{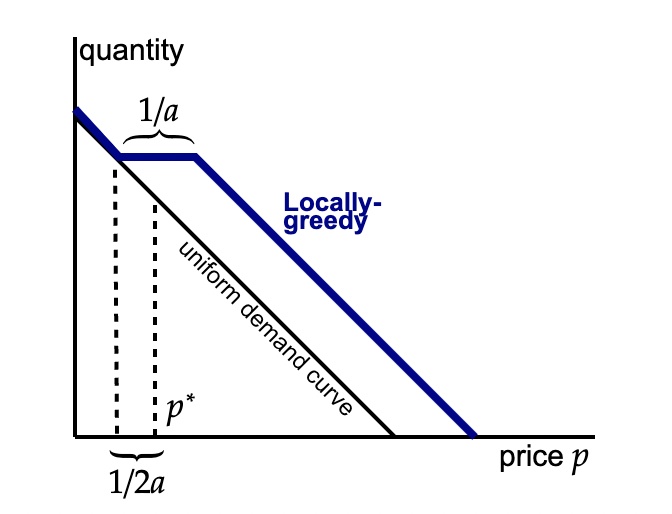}}
    \subfloat[Constrained-greedy (\cref{eg:expost_uniform_eg})]{\includegraphics[width=0.49\textwidth]{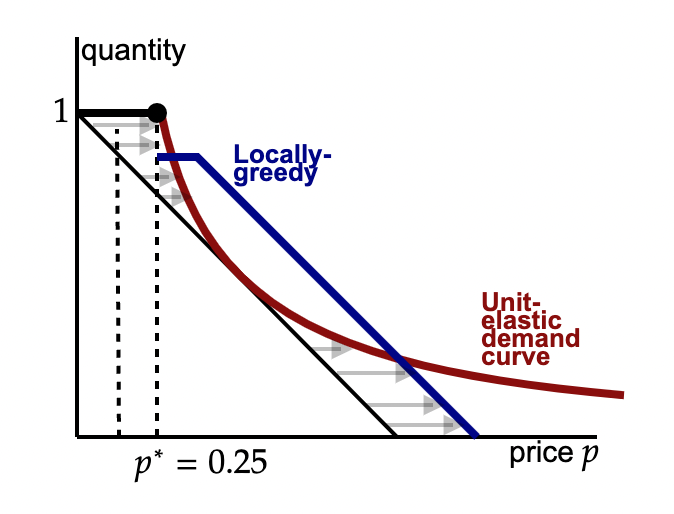}}
    \label{fig:expost_example_full}
\end{figure}

Observe that under the locally-greedy map with price $p^* = 1/4$, it is pointwise suboptimal to advertise to consumers with valuations $x \in [0, 1/8)$. Nonetheless, the constrained greedy plan transports these consumers to an atom at $p^*$. This is driven by two distinct forces. Similar to the ex-ante case, this allows the monopolist to \emph{sustain a lower price} which benefits inframarginal consumers. Different from the ex-ante case, this also \emph{shifts out the lower envelope} which allows more inframarginal consumers to be transported---and made better off. $\hfill \diamondsuit$
\end{example}

\section{Comparing and quantifying prices and welfare.}\label{sec:comparison}
Theorems \ref{thrm:exante_general} and \ref{thrm:expost_general} pin down the general form of optimal advertising plans. We now use this to quantify the power of flexible demand manipulation. 

\textbf{Benchmark: uniform advertising.} A helpful benchmark to isolate the impact of flexibility is to compare the designer's choice of $\Pi$ (all advertising plans) to a more restrictive class---the set of \emph{uniform} advertising plans: 
\begin{align*}
    \pi^d &:= \Big\{\pi \in \Pi: \underbrace{\int d\pi(x,x+d) = 1}_{\substack{\text{Each consumer's valuation} \\ \text{increases by $d$}}} \Big\} \quad \text{and} \quad \Pi^{U} := \bigcup_{d \geq 0} \pi^d.  
\end{align*}

Uniform advertising plans are in the same spirit as the persuasive advertising studied by \cite{dixit1978advertising} who consider a setting where valuations are uniformly increased by the same proportion though it is qualitatively similar.\footnote{See Online Appendix \ref{appendix:dixitnorma_comparison} where we develop a comparison between uniform additive and uniform proportional shifts.} We will also denote standard monopoly pricing case without advertising as the null advertising plan $\pi^0$. Finally, for this section we will assume that the cost function is a strictly convex function of the transportation distance: $c(x,y) = c_d(y-x)$, where $c_d:\Re_+\to\Re_+$ is strictly increasing, convex, $c_d(0) = c_d'(0) = 0$. We also assume, just for this section, that the initial distribution $F$ has an increasing hazard rate.

\textbf{Quantifying the impact of targeted vs uniform advertisements.} Figures \ref{fig:comparison_beta} and \ref{fig:comparison_exp} illustrate prices (panel (a)), producer surplus (b), and ex-ante consumer surplus (c) under different advertising plans. 

\begin{figure}[h!]
    \centering
      \caption{Comparison under beta distribution} 
    \begin{quote}
    \vspace{-1em}
    \centering 
    \footnotesize Note: $c_d(y-x) = (y-x)^2$, $F$ is the beta$(\alpha,\beta)$ distribution with $\beta = 2$ and varying $\alpha$
    \end{quote} 
    \vspace{-1em}
    \subfloat[Prices]{\includegraphics[width=0.33\textwidth]{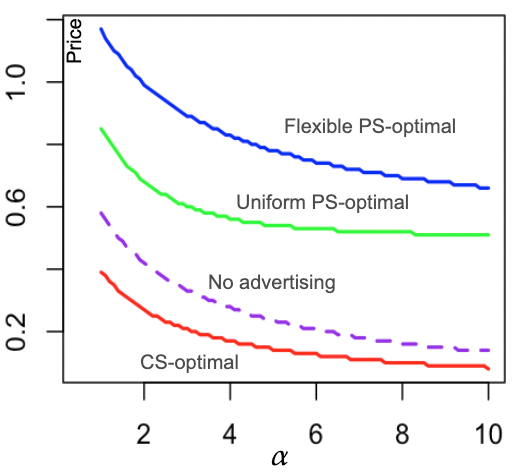}} 
    \subfloat[Producer Surplus]{\includegraphics[width=0.33\textwidth]{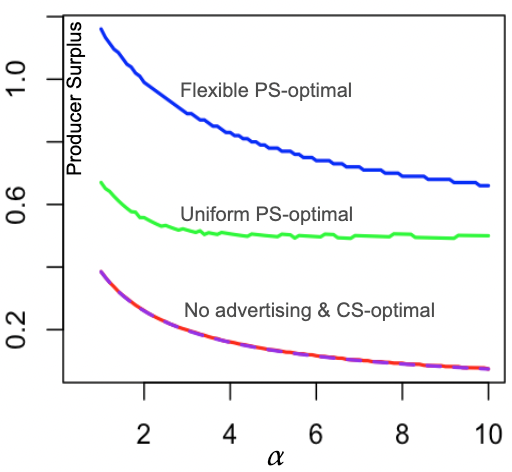}}     
    \subfloat[Consumer Surplus]{\includegraphics[width=0.33\textwidth]{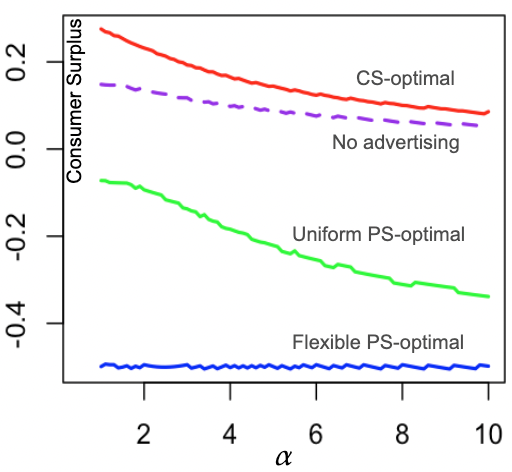}} 
    \label{fig:comparison_beta}
\end{figure}

The \textcolor{blue!70!black}{blue line} corresponds to the producer-optimal advertising plan when the designer has the ability to manipulate demand flexibly---this is pinned down by \cref{prop: produceroptimal}. The \textcolor{red!70!black}{red line} corresponds to the consumer-optimal advertising plan when flexible demand manipulation is possible---this is pinned down by \cref{prop: exanteconsumeroptimal}. The \textcolor{green!85!black}{green line}  corresponds to the producer-optimal plan when the designer is restricted to the set of uniform advertising plans. This is similar to the analysis of \cite{dixit1978advertising}. The dotted \textcolor{violet!95!black}{purple line}  corresponds to no advertising which recovers the standard monopoly pricing environment.

\begin{figure}[h!]
    \centering
      \caption{Comparison under exponential distribution} 
    \begin{quote}
    \vspace{-1em}
    \centering 
    \footnotesize Note: $c_d(y-x) = (y-x)^2$, $F$ is the exponential distribution with varying rate $\lambda$
    \end{quote} 
    \vspace{-1em}
    \subfloat[Prices]{\includegraphics[width=0.33\textwidth]{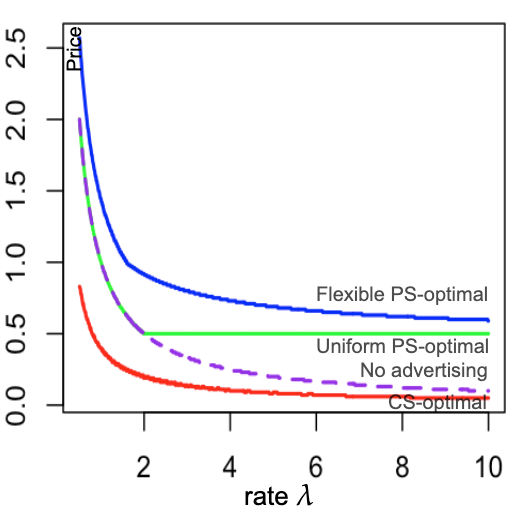}} 
    \subfloat[Producer Surplus]{\includegraphics[width=0.33\textwidth]{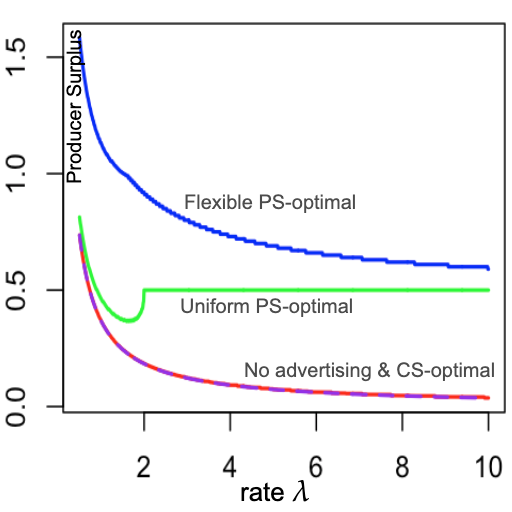}}     
    \subfloat[Consumer Surplus]{\includegraphics[width=0.33\textwidth]{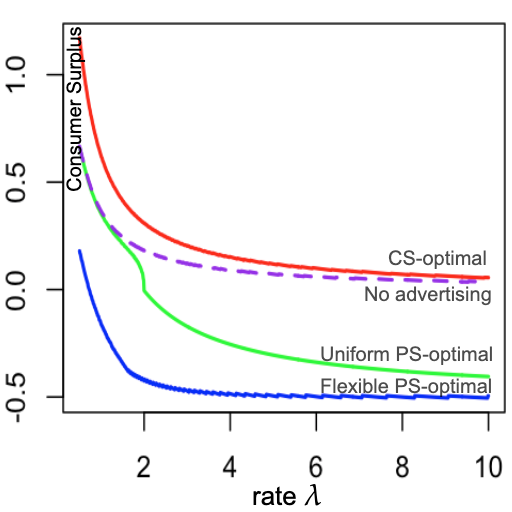}} 
    \label{fig:comparison_exp}
\end{figure}

A durable takeaway is that (i) flexible demand manipulation is substantially more powerful than uniform advertising, and can be used to skew the surplus split. When the designer's objective is to maximize producer surplus, the ability to flexibly shape valuations roughly doubles the increase in profits from advertising---to the detriment of consumers. This refines the economic message of \cite{dixit1978advertising} and more recent work by \cite{acemoglu2024online} who raise concerns that uniform advertising can hurt consumers by raising prices. Our analysis shows that targeted advertising can substantially magnify this force, as well as makes precise how this is done. But our results also offer a richer and more nuanced picture---the flexibility afforded by targeting can always be used to \emph{push down} prices by tempting the monopolist to capture the extensive margin. We now supplement these numerical results with a series of comparative statics. 

\begin{proposition}[Comparison under producer-optimal plans] \label{prop:produceroptimal_vs_uniform}
If the density $f$ is decreasing, 
    \begin{itemize}[nosep]
        \item[(i)] \textbf{Prices increase:} 
        \[p^M \leq \underbrace{p^U}_{\substack{\text{PS-optimal} \\ \text{uniform} \\ \text{advertising}}} < \underbrace{p^*}_{\substack{\text{PS-optimal} \\ \text{flexible} \\ \text{advertising}}}\]
        \item[(ii)] \textbf{Consumer welfare decrease:} 
        \begin{alignat*}{2}
            &\text{Ex-ante: if $q^* \geq q^U$,}\quad  &&CS^A(\pi^*_{PS}) < CS^A(\pi^U_{PS}) < CS^A(\pi^0)
            \\
            &\text{Ex-post:} && CS^P(\pi^*_{PS}) <CS^P(\pi^U_{PS}) \quad \text{and} \quad CS^P(\pi^*_{PS}) < CS^P(\pi^0)
        \end{alignat*}
    \end{itemize}
    where ${\pi^*_{PS}}$ is a producer-optimal plan under full flexibility and $q^*$ is the corresponding optimal quantity; ${\pi^U_{PS}}$ is a producer-optimal plan under uniform advertising and $q^U$ is the corresponding quantity.
\end{proposition}

\cref{prop:produceroptimal_vs_uniform} states that under additional assumptions on the distribution of initial valuations, the additional flexibility from targeted advertising tends to raise prices and hurt consumer welfare relative to uniform advertising.\footnote{Note that although \cref{prop:produceroptimal_vs_uniform} is stated for distributions with decreasing densities, we have found the ordering of prices to hold for canonical distributions---for instance, \cref{fig:comparison_beta} depicts prices and welfare under the beta distribution which has a non-monotone density.} Indeed, as the cost of advertising grows, the \emph{relative} increase in prices under targeted advertising vis-a-vis uniform advertising also grows unboundedly: if manipulation costs are proportional to the quadratic distance i.e., $c_d(y-x) = a\cdot (y-x)^2$ then 
    \vspace{-1em}
    \[
    \vspace{-0.5em}
    \lim_{a\to +\infty} \cfrac{p^* - p^M}{p^U - p^M} = +\infty.
    \]
This highlights how targeting is, on the margin, `infinitely' more cost-efficient than uniform advertising, and this drives the difference in relative price increases.

\begin{proposition}[Comparison under ex-ante consumer-optimal plans] 
\label{prop:consumeroptimal_vs_uniform} \leavevmode
    \begin{itemize}[nosep]
        \item[(i)] \textbf{Prices decrease:} 
        \[\underbrace{p^U}_{\substack{\text{$CS^A$-optimal} \\ \text{uniform} \\ \text{advertising}}} = p^M > \underbrace{p^*}_{\substack{\text{$CS^A$-optimal} \\ \text{flexible} \\ \text{advertising}}}\]
        \item[(ii)] \textbf{Consumer welfare increase:} $CS^A(\pi^*_{CS^A}) > CS^A(\pi^U_{CS^A}) =  CS^A(\pi^0)$,
    \end{itemize}
    where ${\pi^*_{CS^A}}$ is an ex-ante CS-optimal plan under flexible advertising, and $\pi^U_{CS^A}$ is an ex-ante CS-optimal plan under uniform advertising.
\end{proposition}

\cref{prop:consumeroptimal_vs_uniform} states that consumer surplus can generally be strictly improved via flexible demand manipulation. This is in contrast to uniform advertising which is powerless to do so: a designer aiming to maximize consumer surplus subject to the constraint that the advertising map is uniform \emph{cannot} improve consumer surplus over the level under no advertising. This poses a conundrum for regulators: outright bans on targeted advertising i.e., restricting advertising to uniform shifts will generally not safeguard consumer welfare.

\section{Joint Design of Manipulation and Information} \label{sec:joint}
We now augment our model by considering a designer that employs both persuasive and informative advertising. This connects our analysis of flexible persuasive advertising to the extant literature on informative advertising, as well as delivers insights into how these instruments interact. As in our main model, the distribution of consumer valuations under full information is given by $\mu$. Differently, we now allow the designer to \emph{jointly} and \emph{flexibly} shape both preferences and information: 

\begin{figure}[h!]
    \centering
    \vspace{-0.5em}
    {\includegraphics[width=0.85\textwidth]{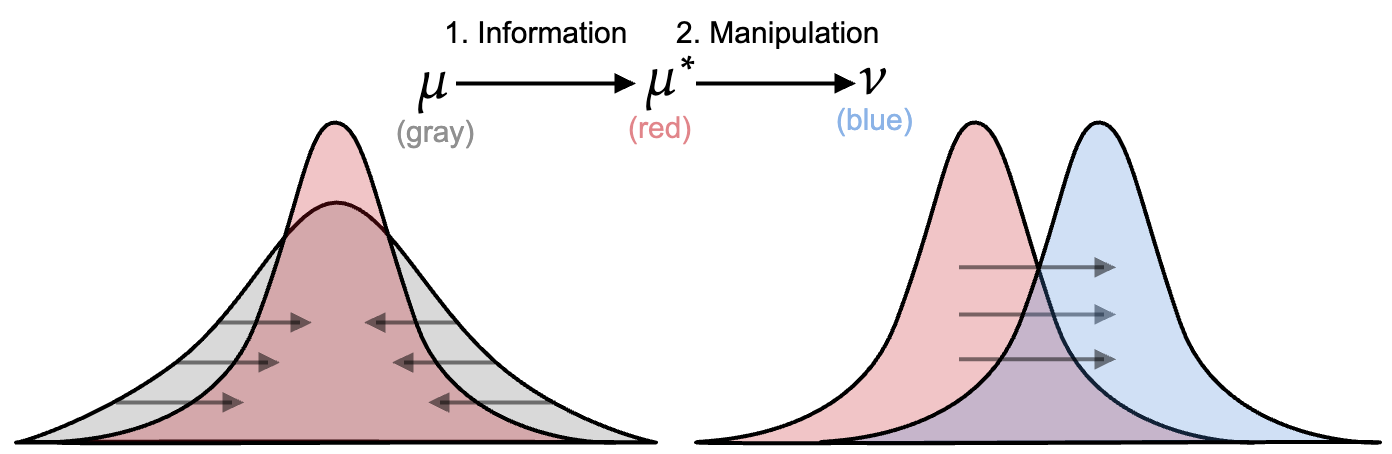}}     
    \vspace{-0.5em}
      \caption{Illustration of a joint plan} 
    \label{fig:joint_plan_illust}
\end{figure}

\begin{enumerate}
    \item Consumers receive private information about their valuation which induces the distribution $\mu^* \preceq_{\text{cx}} \mu$, where $\preceq_{\text{cx}}$ denotes the convex order.\footnote{That is, $\int h(x) \de \mu^*(x) \leq \int h(x) \de \mu(x)$ for all convex functions $h$ on $\Re_+$.} This mean-preserving contraction is illustrated on the left of \cref{fig:joint_plan_illust}. 
    Informative advertising preserves the mean valuation due to Bayesian updating, while the ability to implement any mean-preserving contraction reflects the designer’s flexibility in choosing the information structure.
    \item Consumers are then shown targeted persuasive advertising which induces the final distribution $\nu \succeq_{\text{st}} \mu^*$, where $\succeq_{\text{st}}$ denotes the first-order stochastic dominance order. This shift is illustrated on the right of \cref{fig:joint_plan_illust}. The total cost is  $\int c(x,y) d\pi$ and the individual cost is $c(x,y) = c_d(|y - x|)$, an increasing convex function of distance valuations that are shifted as in \cref{sec:comparison}. Fixing the first-stage choice of $\mu^*$, the second stage is identical to our main model in \cref{sec:model}.\footnote{With the caveat that $\mu^*$ might not be atomless---this does not make any substantive difference. The atomless assumption was made in \cref{sec:model} for expositional clarity since optimal advertising plans are induced by deterministic (Monge) maps. With atoms, all our results hold essentially unchanged, though advertising plans must now be stated directly in terms of $\pi$ where mass may need to be split.}
\end{enumerate}

Our modeling choices reflect economically distinct roles of informative and persuasive advertising. Informative advertising helps consumers learn the \emph{suitability of the product} i.e., where along the distribution their true valuation falls, while persuasive advertising makes consumers \emph{desire the product more}. We can have the former without the latter (e.g., personalized recommendations based on objective characteristics), the latter without the former (e.g., uninformative advertising that induces `hype' or other forms of behavioral manipulation), or mixtures of both \citep{kaldor1950economic}.\footnote{\cite{kaldor1950economic} offers a classic discussion of the distinction, noting the informative-persuasive split is ultimately `one of degree'. The term `hype' is borrowed from \cite{johnson2006simple}.} Importantly, both dimensions of advertising are \emph{targeted and personalized}, reflecting distinctive features of new technologies. 
We are interested in how information and preferences should be jointly delivered to maximize either producer or consumer surplus and, in so doing, offer a unified analysis of different channels for shaping the demand curve that has thus far been analyzed in isolation \citep{johnson2006simple}. 

For the objective $f \in \{PS, CS^A,CS^P\}$, the designer's problem is: 
\begin{align*}
    &\sup_{\substack{\mu^* \preceq_{\text{cx}} \mu 
    \\
    \pi \in \Pi[\mu^*]
    }} 
    f(\nu) - \int c(x,y) d\pi 
\end{align*}
where $\Pi[\mu^*]$ is the set of feasible advertising plans where the first marginal is equal to $\mu^*$.\footnote{We impose identical assumptions to our baseline model: it must continue to obey (i) the directional constraint; (ii) the target distribution must have finite mean. The only difference is that the first marginal is now $\mu^*$ rather than the true distribution of initial valuations $\mu$.} We call the designer's choice of $\mu^*$ (information) and $\pi \in \Pi[\mu^*]$ (manipulation) a \emph{joint plan}. Our main result for this section analyzes the form and value of joint plans: 

\begin{proposition}\label{prop:joint}
\phantom{} 
    \begin{itemize}[nosep]
        \item[(i)] \textbf{Producer-optimal.} The optimal joint plan is to offer no information i.e., $\mu^* = \delta_{\int x \de F}$, and increase all consumers' valuations by $[c_d']^{-1}(1)$. The additional value of manipulation over information is always strictly positive. 
        \item[(ii)] \textbf{Consumer-optimal (ex-ante).} The optimal joint plan is to induce the consumer-optimal information design of \cite{roesler2017buyer} and not manipulate valuations (i.e., the advertising plan induced by the identity map). The additional value of manipulation over information is always zero. 
        \item[(iii)] \textbf{Consumer-optimal (ex-post).} The additional value of manipulation over information is always strictly positive.  
    \end{itemize}
\end{proposition}

Part (i) of \cref{prop:joint} pins down producer-optimal joint plans. This consists of \emph{all hype and no information}---the designer withholds all information about idiosyncratic product fit such as to maximally contract the distribution $\mu^*$ in the first stage. In the second stage, all consumers' valuations are increased up to the point at which the marginal cost ($c_d'$) is equal to the marginal contribution to producer surplus $(1)$. The underlying logic combines the basic observation that (i) manipulation costs are convex so contractions in the first stage make implementing any candidate optimal target distribution (given by \cref{thrm:exante_general}) in the second-stage cheaper; and (ii) the usual reasoning that a seller who cannot engage in price discrimination will strictly prefer to offer no information to consumers \citep{lewis1994supplying}.

Part (ii) of \cref{prop:joint} states that consumer-optimal joint plans is \emph{no hype and partial information} where the information in question is exactly the consumer-optimal information design of \cite{roesler2017buyer}.\footnote{This generates a unit-elastic demand curve with an atom at top.} Thus, although we have seen  from \cref{prop: exanteconsumeroptimal} that flexible demand manipulation can generically strictly improve ex-ante consumer surplus, it has no residual value over flexible information design.\footnote{\cref{prop: exanteconsumeroptimal} that under the consumer-optimal design prices strictly decrease does not apply here because under the consumer-optimal information design of \cite{roesler2017buyer}, the monopolist's optimal price is at the bottom of the support.} This arises for two reasons. First, mean-preserving contractions via information allow the designer to drive down prices without offering a profit guarantee. By contrast, because we assumed that advertising weakly improves valuations, upward deviation constraints are more severe---they must ensuring that the monopolist's profits are at least as high as under the initial distribution. Second, manipulation is both (i) costly and (ii) increasing the extensive margin requires  manipulating consumers into buying against their interests. By contrast, we have followed the literature by assuming information provision is costless. 

A few remarks are in order. First, as a matter of practical implementation, fully flexible informative advertising is typically infeasible and/or costly. For instance, consumers might learn about the suitability of products via third-party reviews which constrains how finely the designer can shape demand through information provision.\footnote{A subsequent literature has accordingly worked to incorporate more realistic features by studying information design with noise, exogenous information, or costly information generation. With exogenous information, this corresponds to an outside signal generating distribution $\underline{\mu}$ of valuations. Then, the distribution of expected valuations $\mu^*$ in the first-stage must now satisfy the more stringent constraint $\underline \mu \preceq_{\text{cx}} \mu^* \preceq_{\text{cx}} \mu$ so the consumer-optimal design might no longer be feasible.} In these environments, the consumer-optimal information design might no longer be feasible which creates room for flexible demand manipulation---precisely because it is an economically distinct instrument---to strictly improve ex-ante consumer surplus.

Second, \cref{prop:joint} (iii) states that if the designer is interested in maximizing ex-post consumer surplus, flexible persuasive advertising always has strictly positive value. The quickest way to see this is to start from the consumer-optimal information design of \cite{roesler2017buyer} denoted by $\mu^{\mathsf{RS}}$ depicted by the black curve in \cref{fig:info_manipulation_strictimprovement}.    

\begin{figure}[H]
\begin{minipage}[t]{0.5\linewidth}
By construction there is an atom at the top of $\mu^{\mathsf{RS}}$. We can then strictly improve the joint plan by `spreading' this atom over the lower-envelope of the unit-elastic demand curve and the locally-greedy map (as in \cref{thrm:expost_general}) and this delivers a strict improvement in ex-post consumer surplus net of advertising costs.
\end{minipage}%
\hfill%
\begin{minipage}[t]{0.5\textwidth}\vspace{0pt}
\vspace{-1em}
\centering
{\includegraphics[width=0.98\textwidth]{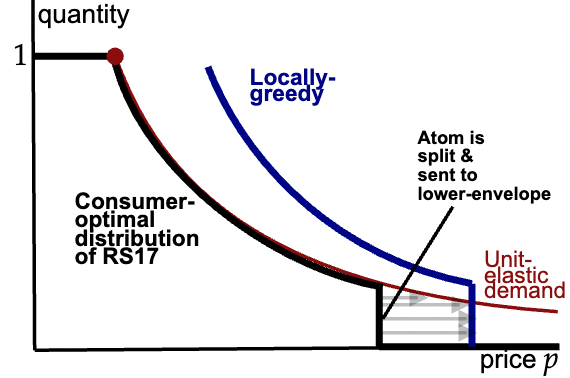}}
    \vspace{-1em}
    \caption{}\label{fig:info_manipulation_strictimprovement}
\end{minipage}
\end{figure}

Third, we have restricted our attention to advertising plans that improve valuations. While we think this is the relevant case for many product markets, we can interpret negative advertising as making some outside option more attractive so valuations for the monopolist's product decreases. Without such directional constraints, flexible demand manipulation can, on its own, deliver strictly more consumer surplus than is possible under flexible information design. Online Appendix \ref{appendix:directional} develops a complete analysis of such plans (Theorems \ref{thrm:exante_U} and \ref{thrm:expost_U}): they resemble those we have seen, but have an extra \emph{twist} structure to reflect both forward and backward comonotone shifts.

\section{Extensions and Generalizations}  \label{sec:extensions}
We now discuss how our framework of flexible demand manipulation can be extended to speak to regulation (\ref{ext:intermediated}), intermediation (\ref{ext:intermediated}), welfare uncertainty (\ref{ext:welfare_uncertainty}), and hetrogeneous suscepbitility to manipulation (\ref{ext:multidim}), and to derive manipulation-robust welfare bounds (\ref{ext:robust_welfare}). Details are collected in Online Appendices \ref{appendix:regulation}-\ref{appendix:het}.

\subsection{Regulation}\label{ext:regulation}
The additional flexibility afforded by targeted advertising can substantially shape the consumer-producer surplus split. This raises the natural question of how advertising plans might be regulated to safeguard consumer surplus. We have seen that \emph{coarse instruments} like outright bans on targeted advertising can ameliorate the extent to which consumer surplus is hurt when firms are in control of the advertising plan (\cref{prop:produceroptimal_vs_uniform}) but are powerless to improve consumer surplus (\cref{prop:consumeroptimal_vs_uniform}).\footnote{At least for the ex-ante case; for the ex-post case, it increases consumer surplus strictly less than flexible advertising.} Can we do better? To shed light on this question, we develop a simple but flexible framework for thinking about regulation:

\begin{definition}
    A regulation of advertising plans $R = \big(\mathcal{G}, \mathcal{L} \big)$ consists of: 
    \begin{itemize}[nosep]
        \item[(i)] \textbf{Limits on targeting.} A partition $\mathcal{G}$ of $X$. Consumers within the same partition must be transported by the same amount.
        \item[(ii)] \textbf{Limits on manipulation.} A function $\mathcal{L}: \mathcal{G} \to \Re_+\cup \{+\infty\}$. Consumers in partition element $g$ can be transported by at most $\mathcal{L}(g)$. 
    \end{itemize}
\end{definition}

$R$ induces a set of feasible transport maps $\Pi[R]$ in the natural way.\footnote{A formal definition of $\Pi[R]$ is in Online Appendix \ref{appendix:regulation}.} Our definition implicitly encodes \emph{free disposal}: the monopolist faces a limit on how much each group can be manipulated, but is free to spend less on advertising. A profit-maximizing monopolist then chooses from the set $\Pi[R]$. Let $\mathcal{R}$ be the set of feasible regulations, our regulator's problem of maximizing (ex-ante) consumer surplus is then 
\begin{align*}
    &\sup_{R \in \mathcal{R}, \pi^* \in \Pi[R]} CS^A(\pi^*) \tag{REG} \label{eqn:regulation}
    \\ 
    &\text{s.t. } \pi^* \text{ solves } \sup_{\pi \in \Pi[R]} PS(\pi) - \int c(x,y) \de\pi
\end{align*}

Our regulatory framework is quite rich, and nests as special cases (i) laissez-faire (no regulation); (ii) bans on targeting (so advertising must be uniform in the spirit of \cite{dixit1978advertising}); and (iii) linear taxes on advertising (the `digital ad tax' proposed by \cite{romer2021taxing,acemoglu2024online}).\footnote{That is, every outcome which can be implemented in equilibrium via linear taxes can also be implemented within our framework of flexible regulation.} \emph{Complex} regulatory instruments that employ a mixture of limits on targeting for some consumers and limits on manipulation on others not only do better, but can achieve the first-best outcome, delivering the maximum amount of consumer surplus while holding the monopolist down to its outside option of selling without advertising.

The regulation that achieves this is formalized by \cref{thrm:regulation} in Online Appendix \ref{appendix:regulation}. It takes an \emph{intermediate interval} form (resembling that of \cref{thrm:exante_general}), prescribing that consumers in the interval $[\underline p, p^*)$ can be precisely targeted, and that their valuations can be manipulated up until $p^*$ but no further. All other consumers cannot be targeted.\footnote{$[\underline p, p^*)$ is chosen such that the firm is exactly indifferent between shifting valuations of intermediate consumers to $p^*$, and forgoing advertising altogether and simply pricing as a monopolist.} Although the firm can freely advertise less, it never chooses to do so---thus, the careful design of a limit on manipulation effectively serves as a quota. 

This model of regulation is, of course, stylized and in practice will be constrained by information or political concerns. Nonetheless, we view this as a benchmark for the {first best} regulatory policy. Moreover, there is growing precedent and/or proposals for employing more complex regulatory instruments---for instance, prohibiting targeted ads conditioned on location or protected class (the proposed `Banning Surveillance Advertising Act'), as well as at children (the `Kids Online Safety and Privacy Act').\footnote{See \url{https://www.congress.gov/bill/117th-congress/house-bill/6416} and \url{https://www.congress.gov/bill/118th-congress/senate-bill/2073/}.} These proposals have been largely motivated by privacy concerns---we highlight the wider implications on prices and consumer surplus, as well as offer a unified framework amenable to more careful analysis.

\subsection{A simple model of intermediated advertising} \label{ext:intermediated}
We now embed flexible demand manipulation into a richer model of the market for digital advertising with endogenous entry. The timing of the game is as follows: 
\begin{itemize}[nosep,leftmargin=*]
    \item \textbf{Consumers enter at $t=1$.} There is a unit measure of consumers with heterogeneous outside options distributed $\eta$ with CDF $H$. At time $t = 1$ they decide on whether to use the platform by comparing their idiosyncratic outside option with their expected surplus on-platform.\footnote{See recent work by \cite{ichihashi2023addictive,acemoglu2024online} among many others for related models incorporating endogenous consumer decisions on whether to use platforms.} 
    \item \textbf{Intermediary offers plan at $t=2$.} Conditional on participation, each consumer's valuation $x \in X$ is drawn independently from the distribution $\mu$.\footnote{We have in mind the intermediary as a gateway to many products, and dispersion in valuations for each product is driven by horizontal differentiation so that aggregating across products, expected individual consumer welfare is independent from the outside option.\label{footnote:intermediary_manyproducts}} These valuations are observed by the intermediary which raises revenue by selling such plans to the monopolist.\footnote{This might correspond to ``supply side platforms'', or a large e-commerce platform like Amazon. \cite{bergemann2024data} develop a model of how platforms might raise revenue by selling access to consumers; the crucial difference in that here, our platform is selling targeted advertising plans to shape valuations.} That is, the intermediary chooses a take-or-leave offer comprising an advertising plan and transfer pair $(\pi,t)$.
    \item \textbf{Firm chooses plan at $t=3$.} The firm decides whether or not to accept the offer $(\pi,t)$. If so, it pays the intermediary the transfer $t$ and the intermediary advertises on the firm's behalf. The firm then sets an optimal price against the transformed demand curve and payoffs are realized.
\end{itemize}
The intermediary's problem is then 
\[
\sup_{\pi\in \Pi} \, \underbrace{H\Big(CS(\pi)\Big)}_{\substack{\text{Measure of} \\ \text{consumers}}}
\cdot 
\underbrace{\Big(PS(\pi) - \int c(x,y) \de\pi\Big)}_{\substack{\text{Average platform} \\ \text{revenue per consumer}}}
\]
and appropriately charging the firm to extract the residual surplus.\footnote{That is, letting $\pi^*$ be the solution, setting $t^* :=  H(CS(\pi^*)) \cdot PS(\pi^*) - H(CS(\pi^0))\cdot PS(\pi^0)$.} 
This multiplicative functional form is nested within our general objective $\phi$ so \cref{thrm:exante_general} (for $CS^A$) and \cref{thrm:expost_general} (for $CS^P$) apply.\footnote{It is without loss to take the maximum of the objective function and $0$ since the designer can always guarantee herself weakly positive surplus by not advertising (so $C = 0$).}

Likewise, a slightly different extensive-form modeling endogenous entry by the monopolist with outside option $\eta \sim G$\footnote{At time $t = 1$, the platform chooses and implements an advertising plan $\pi$. At time $t = 2$, the monopolist decides whether to enter, breaking ties in favor of entering i.e., if and only if $\eta \leq PS(\pi)$. At time $t = 3$, if the monopolist entered it sets optimal prices and payoffs are realized; if the monopolist did not enter, $CS = PS = 0$ but the advertising costs are sunk.}
implies that a platform wishing to maximize consumer surplus now solves
\[
 \sup_{\pi\in \Pi} \, \underbrace{G\Big(PS(\pi)\Big)}_{\substack{\text{Probability} \\ \text{that firm enters}}} \cdot \underbrace{\Big(CS(\pi)\Big)}_{\substack{\text{Consumer welfare} \\ \text{upon entry}}} - \underbrace{\int c(x,y) \de\pi}_{\substack{\text{Sunk advertising} \\ \text{costs}}}
\]

Indeed, a durable insight from the literature on two-sided markets is that platforms often have incentives to deliver surplus to both sides of the market \citep{rochet2003platform}. \cref{thrm:exante_general,thrm:expost_general} speak directly to such settings.

\subsection{Welfare Uncertainty} \label{ext:welfare_uncertainty} We have discussed both advertising as behavioral manipulation (\cref{sec:manipulation}) and as complements (\cref{sec:complements}). What is the right measure welfare, and how should uncertainty be reflected in advertising plans? 

Welfare uncertainty can be either \emph{positive} e.g., the designer is unsure whether advertising works by deceiving consumers or by making them genuinely happier, or \emph{normative} e.g., the designer knows that advertising changes preferences, but is unsure whether this constitutes welfare.\footnote{Non-normative welfare uncertainty is in the spirit of behavioral welfare economics \citep{bernheim2016good}. Normative welfare uncertainty lies at the intersection of economics and philosophy; see \cite{broome1991weighing}.  \cite*{macaskill2020moral} have recently defended taking expectations to aggregate normative uncertainty e.g., about the right moral theory.} In the absence of such uncertainty---e.g., if it is known that advertising works by inducing mistakes---the designer has enough information to compute `true' welfare.\footnote{This is because we \emph{assume} that the designer knows the distribution of initial preferences $\mu$ and can pick the final distribution $\nu$ at some cost. In this regard, our exercise is distinct from the literature on how to identify welfare from choices. This exercise is loosely related to the notion of `idealized welfare analysis' proposed in \cite*{ambuehl2022evaluating}.} But many cases are not so clear cut; we now analyze consumer-optimal advertising plans under welfare uncertainty.
\begin{enumerate}[leftmargin =*]
    \item \textbf{Expectations.} Suppose that our designer assigns probability $\beta \in [0,1]$ on the correct measure being ex-ante. Taking expectations, she solves 
    \[
    \sup_{\pi \in \Pi} \bigg( \beta\, CS^A(\pi) + (1-\beta)\, CS^P(\pi) - \int c \de\pi \bigg).
    \]
    The solution is a \emph{$\beta$-distorted constrained greedy map} that qualitatively resembles the solution in \cref{thrm:expost_general} ($\beta = 0$) with two main differences. First, the gain to inframarginal consumers in the highest quantile is scaled down by $(1-\beta)$ so the technological frontier is shifted in. Second, consumers in the middle quintile shifted to the target price $p^*$ might now be purchasing against their own interests and this deduction to consumer welfare is accordingly scaled by $\beta$.\footnote{The first difference changes our definition of locally-greedy plans since the technological constraint is scaled down. The second difference is only reflected in the designer's optimal target price and quantity $(p^*,q^*)$.} 
    \item \textbf{Max-min.} Suppose that our designer has non-Bayesian uncertainty and wishes to maximize her consumer welfare guarantee. Then, she solves 
    \[
    \sup_{\pi \in \Pi}\inf_{\beta \in [0,1]} \bigg( \beta\, CS^A(\pi) + (1-\beta)\, CS^P(\pi) - \int c \de\pi \bigg).
    \]
    The solution is simply the consumer-optimal solution under ex-ante welfare analyzed in \cref{prop: exanteconsumeroptimal}, precisely because our ex-ante notion of welfare is inherently conservative. 
\end{enumerate}

\subsection{Heterogeneity in susceptibility} \label{ext:multidim} In our model, consumers only differed along initial valuations and this determined their manipulation costs. In practice, consumers might differ in how susceptible they are to persuasive advertising. This sort of heterogeneity can be quite straightforwardly embedded within our framework by augmenting the type space to encode consumers' susceptibility---more susceptible consumers are cheaper to transport (e.g., they need to be shown fewer advertisements) while less susceptible consumers are more expensive---or even impossible---to transport. Online Appendix \ref{appendix:het} formalizes this and shows that optimal advertising plans take qualitatively similar forms. 

\subsection{Manipulation-Robust Welfare Bounds} \label{ext:robust_welfare}  We have analyzed how a designer facing initial distribution $\mu$ might flexibly manipulate demand. But for an analyst `outside the model', only the final distribution $\nu$ is observed since they determine purchase decisions. Our analyst might nonetheless suspect that consumers are behaviorally manipulated, and thus seeks to `invert' the (unknown) advertising plan to compute ex-ante welfare. In ongoing work, we draw on the current framework to derive \emph{robust welfare bounds} in the presence of manipulation technologies: given (i) (partially) observed demand; (ii) observed prices; (iii) total expenditure on advertising; and (iv) bounds on the elasticity of manipulation,\footnote{(i) can be straightforwardly estimated with market data; (ii) and (iii) are readily available; (iv) can be estimated experimentally and in the field.} we extend techniques from this paper to derive lower-bounds on ex-ante welfare when nature chooses the advertising plan adversarially. 

\section{Discussion}\label{sec:discussion}
We have offered a simple but novel framework to study the impact of targeted persuasive advertising on prices and welfare. Our main results (\cref{thrm:exante_general,thrm:expost_general}) fully characterize the form of optimal advertising plans under general objective functions. These results yield qualitative insights into the workings of flexible persuasive advertising, as well as efficient numerical procedures for computation. We used them to shed light on the power of targeted advertising to shape both efficiency and the surplus split---targeting \emph{per se} is substantially more powerful than uniform advertising to either harm or benefit consumers. Moreover, to understand how flexible persuasive and informative advertising interact, we characterize the form and value of using both instruments together vs separately. Finally, we demonstrated how our framework and results can be readily extended to answer a richer set of economic questions related to regulation, intermediation, heterogeneous susceptibility, and deriving manipulation-robust welfare bounds.

Our model is deliberately stylized along several dimensions. We focused on a designer with the ability to both target and persuade precisely. Advertising is, of course, noisy in practice: not every consumer type is precisely known, and not every consumer's valuation can be precisely shifted. Nonetheless, the rise of internet advertising over the past two decades has seen the proliferation of precise targeting via digital advertising \citep{goldfarb2014different}. In the coming years, it is conceivable that new technologies like artificial intelligence might also facilitate precise manipulation. Indeed, there is already substantial experimental evidence that large language models are able to engage in `superhuman persuasion' \citep{matz2024potential,salvi2024conversational,sabour2025human,schoenegger2025large} and systematically influence users' choices. In this regard, our analysis can be viewed as the limiting case in which advertising becomes increasingly precise. 

If preferences can, in fact, be systematically and finely manipulated by these new technologies, this has profound consequences for the prevailing view within economics that \emph{de gustibus non est disputandum} (see an influential account by \cite{stigler1977gustibus}) which, by definition, pushes toward an \emph{ex-post} measure of welfare. Indeed, we think the increasing deployment of such technologies are reasons to---in a growing number of contexts---suspect a rift between choices and welfare: consumer behavior need not be welfare-revealing \citep{bernheim2009beyond}, and an \emph{ex-ante} view can be warranted. Nonetheless, as economists \emph{qua} economists, we have chosen not to take a normative stance on which measure is  constitutive of welfare. Instead, we hope that our results have highlighted the interplay between preference manipulation and market power---the latter of which economists understand well---as a rich and important topic for future work.

\titleformat{\section}
		{\bfseries\center\scshape}     
         {Appendix \thesection:}
        {0.5em}
        {}
        []

\clearpage 
\appendix 
\setstretch{1.2}
\begin{center}
    \textbf{APPENDIX TO FLEXIBLE DEMAND MANIPULATION}
\end{center}

Appendix \ref{appendix:proofs_main} collects proofs of results from \cref{sec:manipulation,sec:complements} characterizing the optimal advertising plan under both ex-ante and ex-post measures of consumer welfare. Appendix \ref{appendix:proofs_compstat} collects proofs of results from \cref{sec:comparison,sec:joint}, which developed comparative statics, as well as analyzed the joint design of manipulation and information.

\section{Proofs of results in Sections \ref{sec:manipulation} \& \ref{sec:complements}} \label{appendix:proofs_main}

\subsection{Outline and preliminaries} 

For clarity, we first offer an overview of the proof of \cref{thrm:exante_general,thrm:expost_general} and how various lemmas fit together (\cref{fig:proof_outline}). 

\textbf{Sketch of \cref{thrm:exante_general}.} \cref{lem: existence} establishes existence of solutions to problems \eqref{eq: designer_exante} and \eqref{eq: designer_expost}. \cref{lemma:domination_lemma} establishes that all solutions to \eqref{eq: designer_exante} must have an intermediate interval structure which, combined with existence, establishes \cref{thrm:exante_general}.

\begin{figure}[h!]
    \centering
    \caption{Proof outline for \cref{thrm:exante_general,thrm:expost_general}} 
    \vspace{-0.5em}
    {\includegraphics[width=0.7\textwidth]{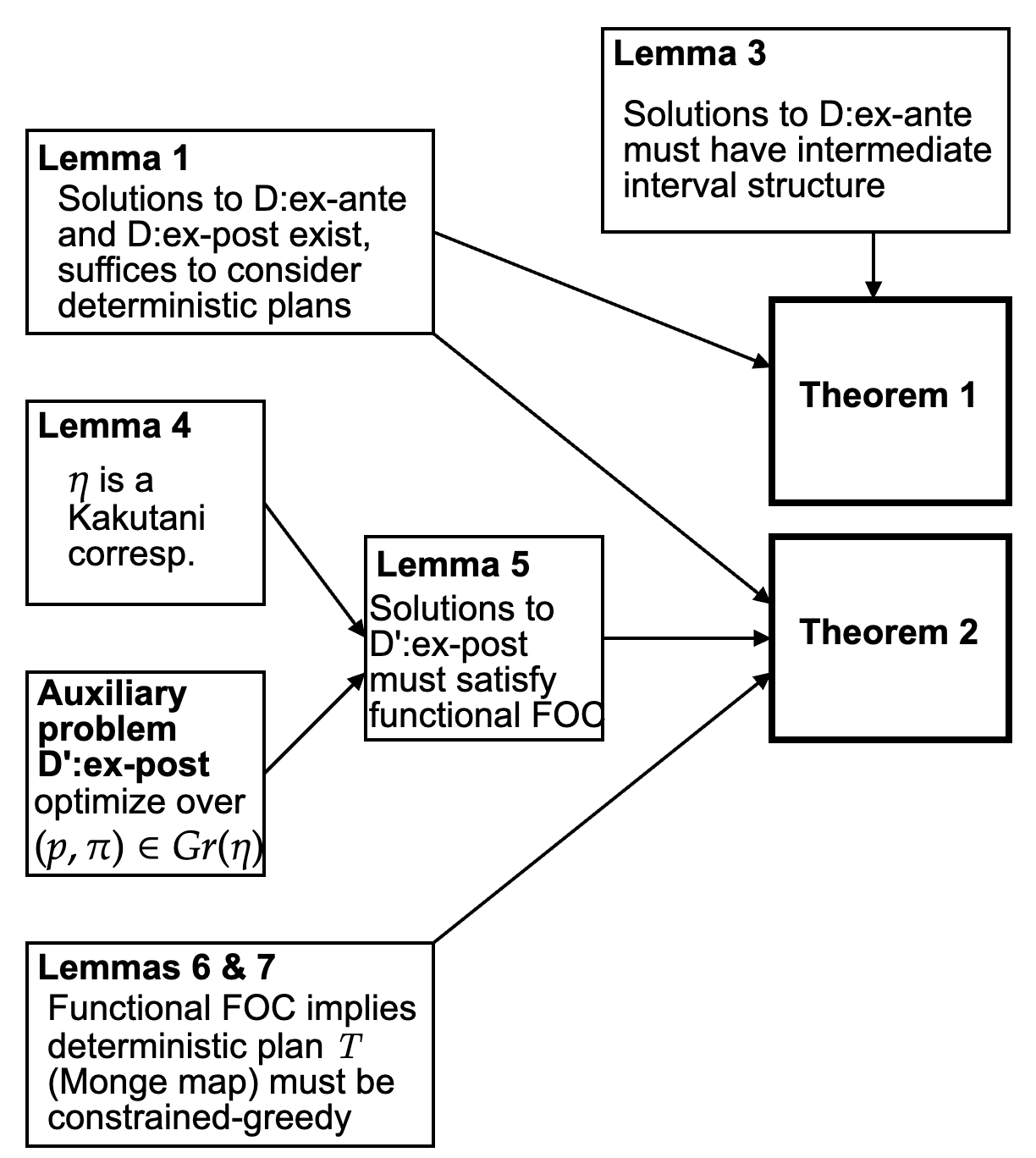}}     
    \vspace{-0.5em}
    \label{fig:proof_outline}
\end{figure}

\textbf{Sketch of \cref{thrm:expost_general}.} For optimal plans under ex-post welfare, we first set up an auxiliary problem \eqref{eq: D' expost} which consists of jointly picking a target price $p$, and an advertising plan $\pi$ that can implement $p$. \cref{Lem: eta property} establishes that the problem is well-behaved. Then, \cref{Lem: expost linearization} establishes that solutions to this auxiliary problem must satisfy a functional first-order condition. Finally, \cref{Lem: expost deterministic,Lem: local soln form} bridges the auxiliary problem with that of optimizing over deterministic plans (Monge maps) so show that all solutions to the original problem \eqref{eq: designer_expost} must be constrained-greedy which establishes \cref{thrm:expost_general}. 

\underline{Definitions.} We now introduce definitions we will use throughout the proofs. For every $\pi \in \Pi$, write $\oF_\pi := \oF_{\marg_y \pi}$ so we can keep track of the second marginal as the transport plan $\pi$ changes. If $\pi$ is induced by the Monge map $T$, we use also $\oF_T$ to denote $\oF_{\pi}$. Let \[I\coloneqq \Big\{T:\Re_+\to \Re_+: T(x) \geq x, \text{$T$ non-decreasing}\Big\}\]
to denote the set of increasing Monge maps that induce $\pi \in \Pi$. For every $\pi \in \Pi$, let $\underline B(\pi) \coloneqq \min \arg\max_{p \geq 0} p\oF_{\pi}(p)$ be the monopolist's lowest optimal price given $\pi$.\footnote{The minimum of the argmax is attained since $\oF_{\pi}(p)$ is upper semi-continuous and so is $p\oF_{\pi}(p)$. The argmax exists by \cref{lem: tail}.} Let $L > 0$ be a global upper bound on the partial derivatives of $\phi$. Finally, it will be convenient to define $r^M := p^M\oF(p^M)$ which is the monopoly's profit in the absence of any advertising. 

A map $f: U \to V$ between normed linear spaces is Fr\'echet differentiable at $u\in U$ if there is a bounded linear map $\mathcal{L}: U \to V$ such that 
\[
\lim_{u' \to 0} \frac{\|f(u+u') - f(u) - \mathcal{L}(u')\|}{\|u'\|} = 0, 
\]
and this map $\mathcal{L}$ is called the Fr\'echet derivative of $f$ at $u$, denoted by $Df(u)$. This linear map $\mathcal{L}$ is bounded if and only if there exists $\kappa > 0$ such that $\|\mathcal{L}(u)\|_V \leq \kappa \|u\|_U$ for all $u \in U$.

The following lemma establishes the existence of solutions and sufficiency of deterministic plans; it is shown in Online Appendix \ref{app: existence}.
\begin{lemma} \label{lem: existence}
    Both \eqref{eq: designer_exante} and \eqref{eq: designer_expost} admit a solution. Moreover, any solution is induced by some $T\in I$ with a finite total cost of advertising.
\end{lemma}

\begin{lemma} \label{lem: tail}
    If $h$ is a non-increasing function on $[0,+\infty)$ and $\int_0^{+\infty} h(x) \de x < +\infty$,  then $p h(p) \to 0$ as $p\to +\infty$.
\end{lemma}
\begin{proof}
    Since $\int_0^{+\infty} h(x) \de x $ converges, for any $\e > 0$, there exists $N > 0$ such that $\int_{x}^\infty h(t) \de t < \e/2$ for every $x > N$. It follows that for every $p > 2N$, $p h(p) \leq 2\int_{p/2}^p h(t) \de t < 2\tfrac{\e}{2} = \e$, where the first inequality follows from the fact that $h$ is non-increasing. 
\end{proof}
\cref{lem: tail} implies that for any $\pi \in \Pi$, the target demand function $\oF_\pi$ is a decreasing function with $\int_0^{+\infty} \oF_\pi(y) \de y = \int_{\Re_+^2} y \de \pi(x,y) < +\infty$, so $p\oF_\pi(p) \to 0$ as $p\to +\infty$ and $\underline B(\pi)$ is well-defined.

\subsection{Proof of \cref{thrm:exante_general}}
The following lemma shows that any optimal solution of \eqref{eq: designer_exante} has an intermediate interval structure. \cref{thrm:exante_general} then follows from Lemma \ref{lem: existence} and Lemma \ref{lemma:domination_lemma}.

\begin{lemma}\label{lemma:domination_lemma}
    For every $\pi \in \Pi$ induced by $T\in I$, $\pi' \in \Pi$ induced by 
    \begin{align*}
        T'(x) = \begin{cases}
        x \quad &\text{if $x \geq p^*$ or $x < \oF^{-1}(q^*)$} \\
        p^*  \quad &\text{if $x \in [\oF^{-1}(q^*), p^*)$} 
        \end{cases}
    \end{align*}  
    weakly dominates $\pi$ in \eqref{eq: designer_exante}, where $p^* = \underline B(\pi)$ and $q^* = \oF_{\pi}(p^*)$. Moreover, the dominance is strict if $\pi \neq \pi'$.
\end{lemma}
\begin{proof}
    Fix a $\pi \in \Pi$, let $\pi'$ as constructed in the lemma statement. We first show that $p^* = \underline B(\pi')$. By construction, $\oF_{\pi'}(p) = q^*$ for $p \in [\oF^{-1}(q^*), p^*]$ and coincides with $\oF(p)$ otherwise. Hence, in the former case, $p \oF_{\pi'}(p) = p q^* \leq p^* q^*$ with equality holds only at $p^*$; in the latter case, $p \oF_{\pi'}(p) = p \oF(p) \leq p\oF_{\pi}(p) \leq p^* q^*$, where the first inequality holds because $\marg_y \pi \succeq_{\text{st}} \mu$. It remains to show there's no $p < \oF^{-1}(q^*) \leq p^*$ that is an optimal price for the monopolist. Suppose there is such a $p$, then $p \oF(p) = p^* q^*$ by our previous argument, it follows that $p\oF_{\pi}(p) = p^* q^*$, violates the fact that $p^* = \underline B(\pi)$. 

    Since $p^* = \underline B(\pi')$, $PS(\pi') = p^*q^* = PS(\pi)$ and $CS(\pi') = \int_{\oF^{-1}(q^*)}^{\infty}(x-p^*) \de F(x)$. Next, we show $CS(\pi') \geq CS(\pi)$. Note that $T(\oF^{-1}(q^*)) \geq p^*$ as otherwise $\oF_{\pi}(p^*) < q^*$. Therefore,  
    \[
    CS(\pi') = \int_{\oF^{-1}(q^*)}^{\infty}(x-p^*) \de F(x) \geq \int_{T(x) \geq p^*} (x-p^*) \de F(x) =  CS(\pi),
    \]
    where the first inequality holds because $T$ is increasing and $\oF^{-1}(q^*) \leq p^*$.

    As $T$ is increasing, $T(x) \geq p^*$ for every $x \in [\oF^{-1}(q^*), p^*]$, hence by monotonicity of $c$, $\int c \de \pi' = \int_{\oF^{-1}(q^*)}^{p^*} c(x, p^*) \de F(x) \leq \int c(x,T(x)) \de F(x) = \int c \de \pi$. Moreover, this inequality is strict if $\pi \neq \pi'$. Therefore, $\pi'$ leads to weakly higher PS and CS while (strictly) reducing the cost, which leads to a higher objective value in \eqref{eq: designer_exante}.
\end{proof} 

\subsection{Proof of \cref{thrm:expost_general}} \label{appendix: proof of expost} Before we show \cref{thrm:expost_general} formally, we briefly outline the key ideas.

The existence of solution follows from Lemma \ref{lem: existence}. We show that any solution to \eqref{eq: designer_expost} is constrained-greedy. By \cref{lem: existence}, it suffices to focus on advertising plans with finite total cost: $\Pi^* = \{\pi \in \Pi: \int c \de \pi < +\infty\}$.

Define $\varphi: \Re_+ \times \Pi^* \to \Re^3$ as
\[
\varphi(p,\pi) = \left(\int \mathbb{I}[y\geq p] p \de \pi(x,y), \int \mathbb{I}[y\geq p](y-p)\de \pi(x,y), \int c(x,y) \de \pi(x,y) \right),
\]
which maps a pair of price and advertising plan to the corresponding vector of PS, ex-post CS, and total cost. Let $\eta: [r^M,\infty) \rightrightarrows \Pi^*$ be a correspondence that maps a price $p$ to the set of advertising plans such that $p$ is an optimal price after advertising: $\eta(p)  \coloneqq \{\pi \in \Pi^*: p \oF_{\pi}(p) \geq p'\oF_{\pi}(p')\text{ for every } p'\geq 0\}$.\footnote{For $p < r^M$ the set is empty.} 
\begin{lemma} \label{Lem: eta property}
    $\eta$ is upper hemi-continuous and $\eta(p)$ is non-empty, convex and compact.
\end{lemma}
We prove \cref{Lem: eta property} in Online Appendix \ref{app: existence}.

We consider an auxiliary problem which has the same solution as \eqref{eq: designer_expost}\footnote{In the sense that if $(p,\pi)$  solves the auxiliary problem \eqref{eq: D' expost}, then $p = \underline B(\pi)$ and hence $\pi\in \Pi^*$ solves the original problem \eqref{eq: designer_expost}. Suppose not, then $p^*(\pi) < p$ and thereby $(p^*(\pi),\pi)$ induces a strictly higher CS while the PS and cost are the same, violating the optimality of $(p,\pi)$.}
\begin{align}\label{eq: D' expost}
    \sup_{(p,\pi) \in \Gr(\eta)} \phi \circ \varphi(p,\pi) \tag{D': ex-post}
\end{align}
where $\Gr(\eta)$ is the graph of $\eta$.

The following lemma establishes a first-order necessary condition for a pair $(p,\pi) \in \Gr(\eta)$ to be optimal in \eqref{eq: D' expost}.  Let $M(\Re_+^2)$ be the set of finite signed Borel measures on $\Re_+^2$ with finite weighted total variation norm $\|\pi\|_{WTV} = \int 1 + y + c(x,y) \de |\pi|(x,y)$. By definition, $\Pi^* \subseteq M(\Re_+^2)$. We use $\phi_i$ to denote the partial derivative of $\phi$ with respect to the $i$-th entry.

\begin{lemma} \label{Lem: expost linearization}
    If $(p,\pi) \in \Gr(\eta)$ solves \eqref{eq: D' expost}, then $\pi$ solves 
    \begin{equation} \label{eq: expost local} 
        \max_{\pi'\in \eta(p)} \int \mathbb{I}[y\geq p] \big( \phi_1(\varphi(p,\pi)) p + \phi_2(\varphi(p,\pi))(y-p) \big) + \phi_3(\varphi(p,\pi)) c(x,y) \de \pi'(x,y) \tag{$(p,\pi)$-local}
    \end{equation}
\end{lemma}
\begin{proof}
    If $(p, \pi)$ solves \eqref{eq: D' expost}, then $\phi\circ\varphi(p,\pi) \geq \phi\circ\varphi(p,\pi')$ for every $\pi' \in \eta(p)$. For any $p \geq r^M$, define $\varphi^p:M(\Re_+^2) \to \Re^3$ by $\varphi^p(\pi) = \varphi(p,\pi)$ while extending $\varphi$ to $\Re_+ \times M(\Re_+^2)$. $\varphi^p$ is a bounded linear map and hence it is Fr\'echet differentiable at $\pi$, with $D\varphi^p(\pi) = \varphi^p$ by linearity. Since $\phi: \Re^3\to \Re$ is differentiable, the chain rule of Fr\'echet derivative applies, so we have
\begin{equation}\label{eq: chain rule}
    D(\phi\circ \varphi^p)(\pi) = D\phi(\varphi^p(\pi)) \circ D \varphi^p(\pi) = D\phi(\varphi^p(\pi)) \circ \varphi^p,
\end{equation}
where the first equality follows from the chain rule and the second follows from the linearity of $\varphi^p$. 

By optimality of $(p,\pi)$, we have $D(\phi\circ \varphi^p)(\pi)(\pi' - \pi) \leq 0$ for every $\pi' \in \eta(p)$. By \cref{eq: chain rule}, 
\[
D(\phi\circ \varphi^p)(\pi)(\pi' - \pi) = (D\phi(\varphi^p(\pi)) \circ \varphi^p)(\pi' - \pi) = \langle \nabla \phi(\varphi^p(\pi)), \varphi^p(\pi'-\pi) \rangle,
\]
where $\nabla \phi$ is the gradient of $\phi$. It follows that $\langle\nabla\phi(\varphi^p(\pi)),  \varphi^p(\pi) \rangle \geq \langle\nabla\phi(\varphi^p(\pi)),  \varphi^p(\pi') \rangle$ for every $\pi'\in \eta(p)$.
\end{proof}

For any $(p,\pi) \in \Gr(\eta)$, let 
\[\lambda_{p,\pi}(x,y) \coloneqq \mathbb{I}[y\geq p] \Big( \phi_1(\varphi(p,\pi)) p + \phi_2(\varphi(p,\pi))(y-p) \big) + \phi_3(\varphi(p,\pi)) c(x,y)\] and 
\[\Lambda_{p,\pi}(x) \coloneqq \max \Big(\arg\max_{y\geq x} \lambda_{p,\pi}(x,y)\Big).\] Note that $\Lambda_{p,\pi}(x) \in \{x\} \cup [p,\infty)$ by  the fact that $\phi_3<0$ and $c(x,y) > 0$ for $y > x$. Since $c_{xy} < 0$, $\lambda_{p,\pi}$ is strictly supermodular and hence $\Lambda_{p,\pi}(x)$ is increasing. Recall that by \cref{lem: existence}, it suffices to consider $\pi \in \Pi^*$ induced by some $T \in I$ in \eqref{eq: designer_expost}, hence so is \eqref{eq: D' expost}.

\begin{lemma} \label{Lem: expost deterministic}
    For every $(p,\pi) \in \Gr(\eta)$ such that $\pi$ is induced by some $T \in I$. Let $q = \oF_{T}(p)$, we have
    \begin{enumerate}[nosep]
        \item[(a)] $T(x) \geq p$ for every $x > \oF^{-1}(q)$; 
        \item[(b)] $T(x) \leq \tfrac{pq}{\oF(x)}$ for every $x \geq 0$;
        \item[(c)] If $q < \oF([\Lambda_{p,\pi}]^{-1}(p))$, then $\pi$ is suboptimal in \eqref{eq: expost local}.
    \end{enumerate}
\end{lemma}

\begin{proof} We show each part in turn. 

\underline{Part (a):} Direct consequence of the monotonicity of $T$. 

    \underline{Part (b):}  Suppose toward a contradiction (b) does not hold, then there exists $x\geq 0$ such that $T(x) > \frac{pq}{\oF(x)}$, then by monotonicity of $T$, we have $\oF_{T}(T(x)) \geq \oF(x)$, which implies that $T(x) \cdot\oF_{T}(T(x)) > \frac{pq}{\oF(x)}\oF(x) = pq$, contradicts with the fact that $\pi \in \eta(p)$.

    \underline{Part (c):}  Suppose $q < \oF([\Lambda_{p,\pi}]^{-1}(p))$. Then, the set $S \coloneqq \{x \in ([\Lambda_{p,\pi}]^{-1}(p), p): T(x) < p\}$ has measure $\mu(S) > 0$. Consider an alternative map $T'$ such that $T'(x) = p$ for every $x\in S$ and $T' =T$ outside of $S$. By construction, $T' \in I$ and $\pi'$ induced by $T'$ is in $\eta(p)$. It follows that $\pi'$ strictly dominates $\pi$ in \eqref{eq: expost local} since $ \lambda_{p,\pi}(x,p) > 0 \geq \lambda_{p,\pi}(x,T(x))$ for every $x \in S$. 
    \end{proof}

    Let $\Gamma_{p,\pi}(x)$ be the unique maximizer of $\max_{y \geq x} \phi_2(\varphi(p,\pi)) y + \phi_3(\varphi(p,\pi)) c(x,y)$. For any $x \geq 0$, $\Lambda_{p,\pi}(x) = x$ if $\lambda_{p,\pi}(x,p) < 0$; $\Lambda_{p,\pi}(x) = p$ if $\lambda_{p,\pi}(x,p) \geq 0$ and $\Gamma_{p,\pi}(x) \leq p$; and $\Lambda_{p,\pi}(x) = \Gamma_{p,\pi}(x)$ otherwise.
\begin{lemma}  \label{Lem: local soln form}
    For every $(p,\pi) \in \Gr(\eta)$ such that $\pi$ is induced by some $T \in I$ and $q = \oF_{\pi}(p) \geq \oF([\Lambda_{p,\pi}]^{-1}(p))$, $\pi'\in \eta(p)$ induced by 
    \begin{align*}
        T'(x) = \begin{cases}
        \left(\Lambda_{p,\pi}(x) \wedge \tfrac{pq}{\oF(x)}\right) \vee p \quad &\text{if $x \geq \oF^{-1}(q)$} \\
        x \quad &\text{if $x < \oF^{-1}(q)$} 
        \end{cases}
    \end{align*}  
    weakly dominates $\pi$ in \eqref{eq: expost local}. Moreover, the dominance is strict if $\pi \neq \pi'$.
\end{lemma}
\begin{proof}
    We first show that $\pi' \in \eta(p)$. By construction, $\oF_{\pi'}(x) = \oF(x)$ for $x < \oF^{-1}(q)$ and $\oF_{\pi'}(x) = q$ for $x \in [\oF^{-1}(q), p]$, so $x \oF_{\pi'}(x) \leq pq$ for $x \leq p$. For $x > p$, note that $\oF_{\pi'}(x) = \oF(\inf_t\{T'(t)\geq x\}) \leq pq/x$, where the last inequality holds because $T'(t) \geq x$ implies $\oF(t) \leq pq/x$ by construction. Therefore, $x \oF_{\pi'}(x) \leq pq$ for any $x \geq 0$. 
    
    It remains to show that $\marg_y \pi'$ has a finite mean and the total cost of $\pi'$ is finite. For simplicity, we abuse the notation and let $\phi_i$ denote $\phi_i(\varphi(p,\pi))$. By \cref{main assumption}, there exists $M > 0$ such that $y - x > M$ implies that $c_y(x,y) > L/|\ell|$. It follows that $\Gamma_{p,\pi}(x) \leq x + M$ as otherwise $c_y(x, \Gamma_{p,\pi}(x)) > L/|\ell| \geq \phi_2/|\phi_3|$. Hence, $T'(x) \leq \Lambda_{p,\pi}(x) \vee p \leq  \max \{p, x+M\}$ and $\int T'(x) \de F(x) < +\infty$ since $F$ has a finite mean. By convexity of $c$ in $y$, $c(x, \Gamma_{p,\pi}(x)) \leq (\Gamma_{p,\pi}(x) - x) \phi_2/|\phi_3| \leq M \phi_2/|\phi_3|$. Therefore, $\int c(x,T'(x)) \de F(x) \leq \int_0^{p} c(x,p+M) \de F(x) + \int_{p}^\infty c(x, \Gamma_{p,\pi}(x)) \de F(x) \leq \int_0^{p} c(x,p+M) \de F(x) + M\phi_2/|\phi_3| \oF(p) < +\infty$.

    By Lemma \ref{Lem: expost deterministic}, $T(x) \in [p, \tfrac{pq}{\oF(x)}]$ for every $x > \oF^{-1}(q)$. By definition of $\Lambda_{p,\pi}(x)$ and the assumption that $c_{yy} > 0$, $ T'(x)$ maximizes $\lambda_{p,\pi}(x,y)$ over $y \in [p, \tfrac{pq}{\oF(x)}]$ for every $x > \oF^{-1}(q)$, and thus $\lambda_{p,\pi}(x, T'(x)) \geq \lambda_{p,\pi}(x,T(x))$. For every $x \leq \oF^{-1}(q)$, $\Lambda_{p,\pi}(x) = x$ since $q \geq \oF([\Lambda_{p,\pi}]^{-1}(p))$, so $\lambda_{p,\pi}(x,x) \geq \lambda_{p,\pi}(x,T(x))$. It follows that $\int \lambda_{p,\pi} \de \pi' = \int \lambda_{p,\pi}(x,T'(x)) \de F(x) \geq \int \lambda_{p,\pi}(x,T(x)) \de F(x) = \int \lambda_{p,\pi} \de \pi$. If $\pi \neq \pi'$, then $T' \neq T$ on a non-zero measure set (under $\mu$), so the inequality becomes strict.
\end{proof}

For any solution $\pi \in \Pi$ that solves \eqref{eq: designer_expost}, $(p, \pi)\in \Gr(\eta)$ solves \eqref{eq: D' expost} where $p = \underline B(\pi)$. Lemma \ref{Lem: expost linearization} then implies that $\pi$ solves \eqref{eq: expost local}. By \cref{lem: existence}, it suffices to consider $\pi$ induced by some $T \in I$, and \cref{Lem: expost deterministic,Lem: local soln form} imply that $\pi$ is constrained-greedy for the pair $(p, \oF_\pi(p))$.

\section{Proofs of results in \cref{sec:comparison,sec:joint}} \label{appendix:proofs_compstat}
We first consider the producer-optimal comparison.
By \cref{prop: produceroptimal}, the designer's problem when maximizing producer surplus can be rewritten as finding an optimal pair $(p', p)$, where $p$ is the price posted by the monopolist and the advertisement increases the valuations of all types $x \in [p', p)$ to $p$. That is, the designer's problem is $\max_{p \geq p' \geq \min X} r(p', p)$, where $r(p', p) = p \cdot \oF(p')- \int^{p}_{p'} c_d(p-x) \de F(x)$.

The first-order necessary conditions for $(\underline p, p^*)$ to be optimal are
\begin{align}
    \frac{\de}{\de p}  r (\underline{p}, p^*) &=1-F(\underline p)- \int^{p^*}_{\underline p} c_d'(p^*-x) \de F(x) =0 \label{eq: FOC-u}
    \tag{FOC-u}\\
    \frac{\de}{\de p'} r (\underline{p}, p^*)&= f(\underline p) [c_d(p^*-\underline p) - p^*] =(\leq) 0 \label{eq: FOC-l} \tag{FOC-l},
\end{align}

If the advertising plan is restricted to $\Pi^U$, then the designer's problem can be rewritten as 
\begin{equation} \label{eq: uniform ad}
    \sup_{\min X \leq p' \leq p} p \overline F(p') - c_d(p - p'). \tag{M-U}
\end{equation}
The corresponding first-order conditions for $(\underline p^U, p^U)$ to be optimal are 
\begin{align}
    \overline F(\underline p^U) - c_d'(p^U - \underline p^U) &=0 \label{eq: U-FOC-u} \tag{U-FOC-u} \\
    -p^U f(\underline p^U) + c_d'(p^U - \underline p^U) &= (\leq) 0. \label{eq: U-FOC-l} \tag{U-FOC-l}
\end{align}

\begin{proof}[Proof of \cref{prop:produceroptimal_vs_uniform}]
    Throughout the proof, we fix a solution $(p^*, \underline p)$ of the monopolist's problem under targeted advertising and $(p^U, \underline p^U)$ under uniform advertising. We denote the induced final value distribution under optimal targeted (uniform) advertising by $\nu, \nu'$, respectively. 

    \textbf{Proof of (i).} Observe that since $c_d'(0) = 0$, under either targeted or uniform advertising the monopolist finds it optimal to shift valuations out by some positive value: $ p^* > \underline p $ and $ p^U > \underline p^U $. We first show that $p^* - \underline p > p^U - \underline p^U$. By \eqref{eq: FOC-u}, we have $\oF(\underline p) = \int^{p^*}_{\underline p} c_d'(p^*-x) \de F(x) < c_d'(p^*- \underline p) \oF(\underline p)$, where the inequality holds as $c_d'$ is strictly increasing. It follows that $c_d'(p^*- \underline p) > 1$. By \eqref{eq: U-FOC-u}, we have $F(\underline p^U) = c_d'(p^U - \underline p^U)$. Therefore, $c_d'(p^U - \underline p^U) \leq 1 < c_d'(p^*- \underline p)$, and thereby $p^* - \underline p > p^U - \underline p^U$.
    
    It follows that if $\underline p \geq \underline p^U$, then $p^* > p^U$. If $\underline p < \underline p^U$, then by the first order conditions, 
        \[
        p^U = \frac{\oF(\underline p^U)}{f(\underline p^U)} \leq \frac{\oF(\underline p)}{f(\underline p)} 
        \]
    where the equality follows from combining \eqref{eq: U-FOC-u} and \eqref{eq: U-FOC-l}, and the inequality follows from the fact that $F$ has an increasing hazard rate.\footnote{Note that $\underline p^U > \underline p \geq \min X$ so the first-order conditions for uniform advertising hold with equality. $f(\underline p) > 0$ since $f$ is decreasing and $f > 0$  in the interior of $X$.} By \eqref{eq: FOC-u}, we have
    \begin{align} \label{eq: PS comp}
        \frac{\oF(\underline p)}{f(\underline p)} = \int^{p^*}_{\underline p} c_d'(p^*-x) \frac{f(x)}{f(\underline p)} \de x \leq  \int^{p^*}_{\underline p} c_d'(p^*-x) \de x = c_d(p^* - \underline p) \leq p^*,
    \end{align}
    where the first inequality holds because $f$ is decreasing and the second inequality follows from \eqref{eq: FOC-l}. Note that if $f$ is constant on $(\underline p, \underline p^U)$, then $\frac{\oF(\underline p^U)}{f(\underline p^U)} < \tfrac{\oF(\underline p)}{f(\underline p)}$; otherwise, the first inequality in \eqref{eq: PS comp} is strict, so in either case $p^* > p^U$.

    Next, we show $p^U \geq p^M$. Adding \eqref{eq: U-FOC-u} and \eqref{eq: U-FOC-l}, we have 
    \[
    p^U \geq \cfrac{\oF (\underline p^U)}{f(\underline p^U)} \geq \cfrac{\oF (p^U)}{f(p^U)} 
    \]
    while on the other hand, the solution to the monopolist problem without advertising, $p^M$, fulfills $p^M = \oF (p^M)/f(p^M)$. Since $F$ has an increasing hazard rate, $\bar F/f$ is decreasing which allows us to conclude $p^U \geq p^M$. 

    \textbf{Proof of (ii).} We first show the result under the ex-ante measure. Since $p^U > \underline p^U$, we have 
    \begin{align*}
        CS^A(\pi^U_{PS})  
        = \int^{p^U}_{\underline p^U} (x-p^U) \de F(x) +  \int^{+\infty}_{p^U} (x-p^U) \de F(x) < CS^A(\pi^0).   
    \end{align*}
    The last inequality holds since the first term is strictly negative, and the second term is bounded above by $CS^A(\pi^0)$ since $p^U \geq p^M$. 

    We now show $CS^A(\pi^*_{PS}) < CS^A(\pi^U_{PS})$. 
    Note that
    \begin{align*}
        CS^A(\pi^*_{PS}) = \int_{\underline p}^{+\infty} (x - p^*) \de F(x) = \int_{\underline p}^{+\infty} x \de F(x) - p^*\oF(\underline p)
    \end{align*}
    and 
    \begin{align*}
        CS^A(\pi^U_{PS}) = \int_{\underline p^U}^{+\infty} (x - p^U) \de F(x) =  \int_{\underline p^U}^{+\infty} x \de F(x) - p^U \oF(\underline p^U)
    \end{align*}
    By assumption, $\oF(\underline p) = q^* \geq q^U = \oF(\underline p^U)$, so $\underline p \leq \underline p^U$ and we have 
    \begin{align*}
        \int^{\underline p^U}_{\underline p} x \de F(x)  \leq \underline p^U (\oF(\underline p) - \oF(\underline p^U)) < p^* \oF(\underline p) - p^U \oF (\underline p^U),
    \end{align*}
    where the second inequality holds by the fact that $p^* > p^U$ and $\underline p \leq \underline p^U$. It follows that $CS^A(\pi^*_{PS}) < CS^A(\pi^U_{PS})$.

    For the comparison under ex-post measure, observe that 
    \begin{align*}
        CS^P(\pi^*_{PS}) = \int^{+\infty}_{p^*} (x-p^*) \de F_{\nu}(x) =& \int^{+\infty}_{p^*} (x-p^*) \de F(x)\\
        <& \int^{+\infty}_{p^M} (x-p^M) \de F(x) = CS^P(\pi^0),
    \end{align*}
    where the second equality follows from the fact that $F_{\nu}(x) = F(x)$ for $x > p^*$, as suggested by \cref{thrm:exante_general}, and the inequality holds since $p^* > p^M$.
    Let $d = p^U - \underline p^U > 0$, we have
    \begin{align*}
     CS^P(\pi^*_{PS}) = \int^{+\infty}_{p^*} (x-p^*) \de F(x) &< \int^{+\infty}_{p^*} (x + d -p^*) \de F(x) \\
     &= \int^{+\infty}_{p^* + d} (x -p^*) \de F_{\nu'}(x) \\
     &< \int^{+\infty}_{p^U} (x -p^U) \de F_{\nu'}(x) = CS^P(\pi^U_{PS}).
    \end{align*}
    The first equality holds because $F_{\nu}(x) = F(x)$ for $x > p^*$, the second equality follows from $F_{\nu'} (x + d) = F(x)$ for $x \in X$, and the last inequality follows from $p^* > p^U$. 
\end{proof}

Next, we consider the Consumer-optimal comparison under the ex-ante metric. 
\begin{proof}[Proof of \cref{prop:consumeroptimal_vs_uniform}] \leavevmode

    \textbf{Proof of (i).} 
    By \cref{prop: exanteconsumeroptimal}, the price $p^*$ under optimal targeted advertising is strictly lower than $p^M$.
    Suppose the advertising plan is restricted to $\Pi^U$. For any $\pi^d \in \Pi^U$ with $d>0$, the monopolist's optimal price given $\pi^d$ is $\arg\max_{p \geq d} p \oF(p-d)$. The derivative of the objective is $\oF(p-d) - p f(p-d)$, and $\oF(p^M-d) - p^M f(p^M-d) > 0$ since $p^M = \oF(p^M) / f(p^M)$ and $F$ has increasing hazard rate. Therefore, the optimal price given $\pi^d$ is strictly higher than $p^M$, and hence the ex-ante CS-optimal advertising plan in $\Pi^U$ is $\pi^0$.

    \textbf{Proof of (ii).} By \cref{prop: exanteconsumeroptimal}, the ex-ante CS optimal plan is of an intermediate interval form with cutoffs $\underline p$ and $p^*$ such that $p^* < p^M$ and $p^*\oF(\underline p) = p^M\oF(p^M)$. Therefore,
    \begin{align*}
        CS^A(\pi^*_{CS}) &= \int^{+\infty}_{\underline p} (x-p^*) \de F(x) = \int^{+\infty}_{\underline p} x \de F(x) - p^*\oF(\underline p)\\
        &>  \int^{+\infty}_{p^M} x \de F(x) - p^M\oF(p^M) = CS^A(\pi^0),
    \end{align*}
    where the strict inequality holds because $\underline p < p^* < p^M$. By our proof of (i), $\pi^U_{CS^A} = \pi^0$ so the induced CS are the same.
\end{proof}

\begin{proof}[Proof of \cref{prop:joint}]\leavevmode

\textbf{Proof of (i).} Given the cost function $c(x,y) = c_d(|y-x|)$, the pointwise optimization problem $\max_{y\geq x} y - c_d(y-x)$ is solved by $ x + d$ where $d = (c_d')^{-1}(1)$. Hence, the total (ex-post) surplus after advertising net of cost is bounded above by $\int_{\Re_+^2} y - c(x,y) \de \pi \leq \int_{\Re_+} x+d - c_d(d) \de F(x) = \int_{\Re_+} x\de F(x) +d - c_d(d)$. This is also an upper bound for the producer surplus net of cost, and it is attained by first contracting to $\delta_{\int x \de F}$ and transporting all consumers by $d$.

\textbf{Proof of (ii).}
    Note that \cref{prop: exanteconsumeroptimal} on the ex-ante consumer-optimal advertising plan can be extended for any $\mu^* \preceq_{\text{cx}} \mu$ even if $\mu^*$ has atoms. The optimal advertising plan in the second stage is still of an intermediate interval form that increases valuations in $[\underline p, p^*)$ to $p^*$, but potentially randomizes among consumers with valuation $\underline p$ if there's an atom, so the producer is exactly indifferent between setting price $p^*$ and $\min \arg\max_{p \geq 0} p \oF_{\mu^*}(p)$. Let $\pi_{CS^A}^*$ be the $CS^A$-optimal plan in the advertising second stage, we have
    \begin{align*}
        CS^A(\pi_{CS^A}^*) = \int \mathbb{I}[y\geq p^*] (x - p^*) \de \pi_{CS^A}^*(x,y) \leq \mathbb{E}_{F^*}(x) - PS(\pi_{CS^A}^*) \leq \mathbb{E}_{F}(x) - p^{RS},
    \end{align*}
    where $p^{RS}$ denotes the seller's lowest profit in the information design first stage as in \cite{roesler2017buyer}, and the last inequality follows from definition of $p^{RS}$ and the fact that $PS(\pi_{CS^A}^*) = \max_{p \geq 0} p \oF_{\mu^*}(p)$, the seller's profit in the first stage. \cite{roesler2017buyer} show this upper bound $\mathbb{E}_{F}(x) - p^{RS}$ can be attained without the second stage advertising, so the result follows.
\end{proof}

\setstretch{0.8}
\small
\setlength{\bibsep}{0pt}
\bibliography{targeting}
\normalsize
\setstretch{1.2}

\clearpage

\appendix 
\titleformat{\section}
		{\normalsize\bfseries\center\scshape}     
         {Online Appendix \thesection:}
        {0.5em}
        {}
        []
\renewcommand{\thesection}{\Roman{section}}

\renewcommand{\theHsection}{\Roman{section}}
\setcounter{section}{0}

\setcounter{page}{1}
{\renewcommand{\thefootnote}{\fnsymbol{footnote}}
\begin{center}
    \large{\textbf{\MakeUppercase{Online Appendix to \\
    `Flexible Demand Manipulation'}}} \\
    \normalsize{{\textbf{YIFAN DAI\footnote[1]{Department of Economics, MIT, \href{mailto:yfdai@mit.edu}{yfdai@mit.edu}} \quad  ANDREW KOH\footnote[2]{Department of Economics, MIT, \href{mailto:ajkoh@mit.edu}{ajkoh@mit.edu}}}\\
    \small{FOR ONLINE PUBLICATION ONLY}}}
\end{center}}
Online Appendix \ref{app: existence} provides some omitted technical details in the proofs. 
Online Appendix \ref{appendix:regulation} details a model of regulating advertising maps and characterizes the consumer-optimal regulatory policy. Online Appendix \ref{appendix:het} details an extension of \cref{thrm:exante_general} to consumers who are heterogeneous across both valuations and susceptibility to manipulation. 
Online Appendix \ref{appendix:directional} characterizes optimal advertising plans without the directional constraint imposed in the main text. Online Appendix \ref{appendix:dixitnorma_comparison} connects our analysis to that of \cite{dixit1978advertising}.

\section{Technical Details} \label{app: existence} 

\begin{proof}[Proof of Lemma \ref{Lem: eta property}] We show each property in turn: 
\begin{enumerate}
    \item  \underline{Non-emptiness:} If $p\geq r^M$, the deterministic plan induced by $T(x) = p$ for every $x < p$ is in $\eta(p)$. 
    \item 
    \underline{Convexity}: For any $\pi,\pi' \in \eta(p)$. If there exists $p' > 0$ such that 
    \[p' \oF_{\alpha \pi + (1-\alpha)\pi'}(p') > p \oF_{\alpha \pi + (1-\alpha)\pi'}(p),\] 
    then either $p' \oF_{\pi }(p') > p \oF_{\pi }(p)$ or $p' \oF_{\pi' }(p') > p \oF_{\pi' }(p)$, contradicting the fact that $p$ is an optimal price under $\pi$ and $\pi'$.
    \item  \underline{Compactness and UHC:} Finally, we show that for any sequence $\{(p_n, \pi_n)\}$ in $\Gr(\eta)$ satisfies $p_n \to p$, then $\{\pi_n\}$ has a limit point in $\eta(p)$, which is then equivalent to $\eta$ being upper-hemicontinuous and compact-valued by Theorem 17.20 of \cite{aliprantis06}.

    Since $p_n \to p$, it is without loss to assume $p_n \leq \hat p$ for some $\hat{p} > p$. Let $\nu$ be a measure such that $\oF_{\nu}(x) = 1$ for $x < \hat{p}$ and $\oF_{\nu}(x) = \hat{p}/x$ otherwise. Since $p_n \leq \hat{p}$ is an optimal price given $\pi_n$, we have $\marg_y \pi_n \preceq_{\text{st}} \nu$ for every $n$. Therefore, the sequence $\{\pi_n\}$ is tight because for any $\e > 0$, there exists $a,b > 0$ such that $\mu(a,\infty) < \e/2$ and $\nu(b,\infty) < \e/2$ by Ulam's Theorem, hence $\pi_n(\Re_+^2 \setminus [0,a]\times[0,b]) \leq \mu(a,\infty) + \nu(b,\infty) < \e$ for any $n$. Since $\Re_+^2$ is Polish, Prokhorov's Theorem implies that $\{\pi_n\}$ has a convergent subsequence. 

    Taking subsequence if necessary, we may assume that $p_n \to p$ and $\pi_n\to \pi$, and we want to show $\pi \in \eta(p)$. By the Portmanteau Theorem, $\oF_{\pi}(x) \geq \limsup \oF_{\pi_n}(x)$ for every $x \geq 0$ and $\oF_{\pi_n}(x) \to \oF_{\pi}(x)$ for any $x$ at which $\oF_{\pi}$ is continuous. Assume to the contrary that there exists $p' > 0$ such that $p' \oF_{\pi}(p') > p\oF_{\pi}(p)$. Since there are at most countably many points where $\oF_{\pi}$ is discontinuous, there exists $\tilde{p}$ in a neighbourhood of $p'$ such that $\tilde{p}\oF_{\pi}(\tilde{p}) > p\oF_{\pi}(p)$ and $\oF_\pi$ is continuous at $\tilde p$. It follows that $\tilde{p}\oF_{\pi}(\tilde{p}) = \lim_n \tilde{p}\oF_{\pi_n}(\tilde{p}) \leq \limsup_n p_n\oF_{\pi_n}(p_n) \leq p \oF_{\pi}(p)$, which is a contradiction.\footnote{The former inequality holds by optimality of $p_n$. The latter holds because the product of non-negative upper semi-continuous function is upper-semi-continuous.}\qedhere
    \end{enumerate}
\end{proof}

\begin{proof}[Proof of \cref{lem: existence}]
We first show that any solution of \eqref{eq: designer_exante} or \eqref{eq: designer_expost} has finite total cost and is induced by some $T\in I$. For any $\pi \in \Pi$, since $\marg_y \pi$ has a finite mean, \cref{lem: tail} implies the corresponding $PS(\pi)$ and $CS(\pi)$ are finite (regardless of welfare measure). If $C(\pi)$ is not finite, then $\phi(PS(\pi),CS(\pi),C(\pi)) = -\infty$ since $\phi_3 < \ell <0$, so $\pi$ is suboptimal. Now consider any advertising plan with finite cost $\pi \in \Pi^*$, it is weakly dominated by $\pi'$ induced by $T = F_{\marg_y \pi}^{-1} \circ F \in I$. This is because $\marg_y \pi = \marg_y \pi'$ so $PS(\pi) = PS(\pi')$ and the price under $\pi$ and $\pi'$ are the same. Since $c$ is strictly submodular and $\mu$ is atomless, Theorems 4.3 and 4.7 of \citep{galichon2018optimal} imply that $\pi'$ weakly increases CS and reduces the cost, and the cost reduction is strict if $\pi \neq \pi'$. Therefore, any solution is induced by some $T \in I$.

Next, we show the existence of solutions for the ex-post problem, the ex-ante case can be established with similar arguments. We show that there exists a $(p,\pi)\in\Gr(\eta)$ that solves the auxiliary problem \eqref{eq: D' expost}. Dropping the incentive constraints in \eqref{eq: expost local}, we consider the following relaxed locally linearized problem at $(p,\pi) \in \Gr(\eta)$:
\begin{equation} \label{eq: local relaxed} 
        \max_{\pi'\in \Pi^*} \int \lambda_{p,\pi}(x,y) \de \pi'(x,y) \tag{relaxed $(p,\pi)$-local}
\end{equation}

\begin{lemma} \label{Lem: existence upper bound}
    There exists $M > 0$ such that for every $p > 4M$ and $\pi \in \Pi^*$ induced by some $T \in I$, if in the interior of $X$, $T(x) > \max\{p, x + M\}$ for some $x \geq p/2$, or $T(x) > x$ for some $x < p/2$, then $\pi$ is suboptimal in \eqref{eq: local relaxed}.
\end{lemma}
\begin{proof}
    By \cref{main assumption}, there exists $M > 0$ such that $y-x > M$ implies that $c_y(x,y) > 4L/|\ell|$. Take a pair $(p, \pi) \in \Gr(\eta)$ such that $p > 4M$ and $\pi$ is induced by some $T \in I$. For simplicity, denote 
    $\phi_i = \phi_i(\varphi(p,\pi))$ for $i = 1,2,3$. Recall that $\Gamma_{p,\pi}(x)$ is the unique maximizer of $\max_{y \geq x} \phi_2 y + \phi_3 c(x,y)$ and $\Gamma_\pi(x) \leq x + M$.

    For any $x \geq 0$, $\Lambda_{p,\pi}(x) \leq \max\{x, p, \Gamma_{p,\pi}(x)\} \leq  \max\{p, x+M\}$. For any $x < p/2$, $c(x,p) = \int_{x}^{p} c_y(x,t) \de t \geq \int_{3p/4}^{p} c_y(x,t) \de t > \tfrac{p}{4}\tfrac{4L}{|\ell|}> \tfrac{L}{|\ell|} p \geq \tfrac{\phi_2}{|\phi_3|}p$. It follows that $\lambda_{p,\pi}(x,p) = \phi_2 p + \phi_3 c(x,p) < 0 $ for $x < p/2$, and hence $\Lambda_{p,\pi}(x) = x$. 

    Without the incentive constraints, \eqref{eq: local relaxed} is solved by the distribution induced by the pointwise optimal $\Lambda_{p,\pi}$. Suppose there exists some $x$ in the interior of $X$ such that $T(x) > \Lambda_{p,\pi}(x)$, then $\pi$ is suboptimal in \eqref{eq: local relaxed}.
\end{proof}

\begin{lemma}
    There exists $\bar p > 0$ such that for any $p > \bar p$ and $\pi \in \eta(p)$, $\phi\circ\varphi(p,\pi) < \phi\circ\varphi(p^M,\pi^0)$.
\end{lemma}
\begin{proof}
    Similar arguments via Fr\'echet derivative as in Lemma \ref{Lem: expost linearization} show that a necessary condition for $\pi \in \Pi^*$ to solve the relaxed problem without incentive constraint $\sup_{\pi \in \Pi^*} \phi\circ\varphi(p,\pi)$ is that $\pi$ solves \eqref{eq: local relaxed}. It suffices to focus on $\pi \in \Pi^*$ induced by some $T\in I$ since $\lambda_{p,\pi}$ is strictly supermodular.
    
    For any $p > 4M$, 
    \[
    \sup_{\pi \in \Pi^*} \phi\circ \varphi(p,\pi) \leq \phi\left(p\oF(p/2), \int_{p-M}^{\infty}(x+M - p) \de F(x), 0\right),
    \]
    where the first and second entries of $\phi$ are bounded above by $p\oF(p/2)$ and $\int_{p-M}^{\infty}(x+M - p) \de F(x)$ by Lemma \ref{Lem: existence upper bound}, respectively, and the third entry is bounded below by $0$. 

    Note that $\oF$ is a decreasing function and $\int_0^{+\infty} \oF(x) \de x = \int_0^{\infty} x \de F(x) < +\infty$. \cref{lem: tail} implies that as $p\to +\infty$, $p \oF(p/2) \to 0$ and 
    \[
    \int_{p-M}^{\infty}(x+M - p) \de F(x) = \int_{p-M}^{\infty} \oF(x) \de x \to 0,
    \]
    where the equality follows from changing the order of integration. Since $\phi$ is monotone in the first two entries, there exists $\bar p > 4M$ such that $\phi\circ\varphi(p,\pi) < \phi\circ\varphi(p^M,\pi^0)$ for every $p > \bar p$ and $\pi \in \eta(p) \subseteq \Pi^*$.
\end{proof}

By the preceding lemma, it is without loss to restrict to $p \in [r^M,\bar p]$ in \eqref{eq: D' expost}:
\[
\sup_{(p,\pi)\in \Gr(\eta)} \phi\circ\varphi(p,\pi) = \sup_{p\in[r^M, \bar p]} \sup_{\pi\in \eta(p)} \phi\circ\varphi(p,\pi).
\]
Note that $\phi\circ\varphi$ is upper semi-continuous in $(p,\pi)$ and $\eta$ is upper hemi-continuous with non-empty compact values, Lemma 17.30 of \cite{aliprantis06} implies the value of the inner problem $\sup_{\pi\in \eta(p)} \phi\circ\varphi(p,\pi)$ is upper semi-continuous in $p$. Therefore, there exists $(p,\pi)\in\Gr(\eta)$ that solves \eqref{eq: D' expost}, and $\pi$ is also a solution of \eqref{eq: designer_expost}.
\end{proof}

\section{Regulation of Advertising Plans} \label{appendix:regulation}

\textbf{A model of regulation.} Recall our definition of a regulation in the main text as $R = \big(\mathcal{G}, \mathcal{L} \big)$ where $\mathcal{G}$ is a partition of $X$ and $\mathcal{L}: \mathcal{G} \to \Re_+ \cup \{+\infty\}$. $R$ induces a set of feasible transport maps $\Pi[R]$ in the natural way: 
    \[
    \Pi[R] := \bigg\{\pi \in \Pi: \substack{\text{\normalsize (i) $x,x' \in g \implies F_{\pi^x}(y+x) = F_{\pi^{x'}}(y+x')$ for every $y\geq 0$}\\
\text{\normalsize (ii) $\pi^x([x, x+\mathcal{L}(g)]) = 1$ for every $x \in g$}}
\bigg\}. 
    \]
where we use $\pi^x \in \Delta(\Re_+)$ to denote a version of the conditional distribution of $\pi$ given $x$. The following are several simple examples: 
\begin{itemize}[nosep,leftmargin = 2em]
    \item \textbf{Laissez-faire.} There are no restrictions on targeting or manipulation: 
    \[R:= \Big(\underbrace{\big\{\{x\}:x \in X\big\}}_{\substack{\text{No restriction}\\\text{on targeting}}}, \underbrace{\mathcal{L} \equiv +\infty}_{\substack{\text{No restriction}\\\text{on manipulation}}}\Big).\]
    and $\Pi[R] = \Pi$ i.e., all advertising plans are feasible. 
    \item \textbf{Bans on targeting.} Each consumer's valuation must be transported by the same amount: 
    \[R := \Big(\underbrace{\{X\}}_{\substack{\text{No }\\\text{targeting}}}, \underbrace{\mathcal{L} \equiv +\infty}_{\substack{\text{No restriction}\\\text{on manipulation}}}\Big)
    \]
    and $\Pi[R] = \Pi^U$.
    \item \textbf{Limits on manipulation.} Consumers can be flexibly targeted, but there is a limit on total manipulation: 
    \[R = \Big(\underbrace{\big\{\{x\}:x \in X\big\}}_{\substack{\text{No restriction}\\\text{on targeting}}}, \underbrace{\mathcal{L}\equiv D}_{\substack{\text{Cap on}\\\text{manipulation}}}\Big) \quad \text{where $D > 0$ is a constant}
    \] 
    and this induces the set of advertising plans 
    \[
    \Pi[R] = \Big\{\pi \in \Pi: \int_{y \leq x + D} \de\pi(x,y) = 1\Big\}.
    \]
    \item \textbf{Pigouvian taxes.} Suppose that the regulator, instead, imposes a Pigouvian tax on advertising, as has been recently proposed by \cite{romer2021taxing} and \cite{acemoglu2024online} among others. Translated into our framework, for any expenditure $C$ on advertising, the designer bears a cost of $\lambda \cdot C$ with $\lambda > 1$. Then, this is equivalent to working with the cost function $\lambda \cdot c(x,y)$, noting that this inherits the assumptions we imposed on $c$. Now, for a given choice of $\lambda$, let $\pi^*_{\lambda}$ denote the producer-optimal plan. From \cref{thrm:exante_general} and \cref{prop: produceroptimal} in the main text, we know this takes an intermediate interval structure. 
    \begin{claim}
        For any tax $\lambda > 1$, the producer-optimal advertising plan $\pi^*_{\lambda}$ can be implemented by choosing $R$ appropriately i.e., there exists $R$ such that 
        \[
        \pi^*_{\lambda} \in \text{argmax}_{\pi \in \Pi[R]} PS(\pi) - \int c(x,y) d\pi.
        \]
    \end{claim}
    The logic is quite similar to the proof of \cref{thrm:regulation} below so we omit the proof. 
\end{itemize}

\begin{proposition}[Optimal regulation]\label{thrm:regulation}  There exists a deterministic plan $\pi^*$ with an intermediate interval structure characterized by $(\underline p,p^*)$ such that $(R^*,\pi^*)$ solves \eqref{eqn:regulation} where 
\begin{align*}
    R^* := \Big(\mathcal{G}^*, \mathcal{L}^*\Big)  \quad \text{with} \quad 
     &\mathcal{G}^* := \bigg\{\underbrace{\Big\{X \setminus [\underline p,p^*)\Big\}}_{\text{No targeting}}, 
        \underbrace{\Big\{ \{x\}: x \in [\underline p,p^*)\Big\}}_{\substack{\text{Precise targeting of} \\ \text{intermediate types}}} 
        \bigg\}\\ 
        &\mathcal{L}^*(x) := \underbrace{(p^* - x) \mathbb{I}(x\in [\underline p, p^*))}_{\text{Limit equal to distance $p^*-x$}}
    \end{align*}

and this strictly improves consumer welfare: $\ref{eqn:regulation} > CS(\pi^0)$. 
\end{proposition}

\begin{proof}[Proof of \cref{thrm:regulation}] Define the following relaxed problem: 
\begin{align*}
    \sup_{\pi \in \Pi} & \,CS^A(\pi) \\
    &\text{s.t. } \Big(PS(\pi) - PS(\pi^0)\Big) \geq \int c(x,y) d\pi.
    \tag{REL} \label{eqn:relaxed}    
\end{align*}
    Observe $\ref{eqn:relaxed} \geq \ref{eqn:regulation}$ since  for any $(R, \pi \in \Pi[R])$ feasible in \eqref{eqn:regulation}, $PS(\pi) - \int c \de \pi \geq PS(\pi^0) - \int c \de \pi^0 = PS(\pi^0)$ and hence $\pi$ is feasible in \eqref{eqn:relaxed}.
    
    We start with the following helpful result: 
    \begin{lemma}\label{lemma:regulation_relaxed}
        The solution to \eqref{eqn:relaxed} is deterministic and monotone with an intermediate interval structure characterized by $(\underline p,p^*)$.
    \end{lemma} 
    \begin{proof}
       First, observe that it is without loss to restrict to comonotone deterministic maps. This holds because fixing the marginals and changing to assortative matching does not affect the objective value and the IC constraint (within the definition of PS), so the monopolist's IR constraint is slackened (the explicit constraint written).

    Those maps are induced by increasing functions $T$ with $T(x) \geq x$. For any feasible $T$ with associated monopolist optimal price $p$, consider $T'$ such that $T'(x) = x$ if $x \geq p$ or $x < p$ with $T(x) < p$, and $T'(x) = p$ if $x < p$ with $T(x) \geq p$. The induced demand $\oF_{T'}(p) = \oF_{T}(p)$ while $\oF_{T'}(p') \leq \oF_{T}(p')$ for every $p' \geq 0$ so the IC is slackened. The IR is also slacked since the total cost decreases. The CS weakly increases since the cost of manipulation goes down. Therefore, it is without loss to consider $T$ such that $T(x) = x$ if $x \notin [\underline p, p)$ and $T(x) = p$ otherwise.

    Similar arguments as in Online Appendix \ref{app: existence} then show that it is without loss to restrict to a compact set of $p$, so there exists a solution of \eqref{eqn:relaxed}. Note that this construction proof also shows $PS(\pi^*) - \int c \de \pi^* = PS(\pi^0)$ at the optimal solution $\pi^*$.
    \end{proof}

    From \cref{lemma:regulation_relaxed}, there is an intermediate interval plan solving \eqref{eqn:relaxed}. Let this be $\pi^*$ and $[\underline p, p^*)$ is the corresponding interval. Now consider $R^*$ constructed in the statement of \cref{thrm:regulation}, with the same $\underline p$ and $p^*$. It remains to show that $(\pi^*, p^*)$ is a solution to the problem  
    \[
    \sup_{\pi \in \Pi[R^*], p \geq 0} p \oF_{\pi}(p) - \int c(x,y) d\pi ,
    \]
    which would imply $\eqref{eqn:relaxed}$ and $ \eqref{eqn:regulation}$ have the same value. Assume to the contrary that it is not and that it is strictly dominated by $(\pi', p') \neq (\pi^*, p^*)$.
    
    If $p' \leq p^*$, then the corresponding optimal map $\pi' \in \Pi[R^*]$ transports some types in $[\underline p', p')$ exactly to $p'$ where $\underline p' \geq \underline p$. But then ex-ante CS must strictly improve under $(\pi', p')$ since either (i) fewer types are being manipulated; or (ii) $p' < p^*$. This contradicts the optimality of $\pi^*$ in \eqref{eqn:relaxed}. 

    If $p' > p^*$, then no $\pi \in \Pi[R^*]$ can improve demand given $p'$, so the corresponding optimal advertising plan must be $\pi' = \pi^0$. Since $(\pi^*, p^*)$ is strictly dominated by $(\pi',p')$, we have 
        \[
        p' \oF(p') > p^* \oF(p^*) - \int c\de  \pi^* = \sup_p p \oF(p),
        \]
    which is a contradiction. In either case there is a contradiction which implies $(\pi^*, p^*)$ is a solution.
\end{proof}

\section{Heterogeneous susceptibility}\label{appendix:het}
In the main text, we assumed for simplicity that consumer were heterogeneous only along their valuations. Implicit in this assumption is that consumers with the same valuations have identical manipulation costs. We now introduce heterogeneity along this dimension, which captures how susceptible to manipulation different consumers are. 

\textbf{Augmenting the main model.} Let $S := \mathbb{R}$ and the augmented type space is $\widehat X := X \times S$. Consumers are distributed according to $\widehat \mu \in \Delta(\widehat X)$ and, consistent with notation in the main text, we let $\mu \coloneqq \text{marg}_{X} \widehat \mu$ denote the marginal over valuations. The cost of transporting type $(x,s)$ to $(y,s)$ is $c_s(x,y)$ where we impose \cref{main assumption} on $c_s$ for each $s \in S$.  We further impose the assumption that higher $s$ corresponds to a lower cost: 
\[
s \geq s' \implies c_s(x,y) \leq c_{s'}(x,y) \quad \text{for all $x,y \in X$.}
\]

An advertising plan $\pi \in \Delta( \widehat X \times \widehat X)$ is admissible if (i) it concentrates on the set 
\[
\widehat{\mathbb{H}} := \Big\{\Big((x,s), (y,s)\Big) \in \widehat X^2: y \geq x \Big\};
\]
(ii) the first marginal of $\pi$ coincides with the initial distribution $\widehat\mu$; and (iii) the distribution over valuations on the second marginal has finite mean. Let $\widehat \Pi$ be the set of all admissible transport plans. 

The designer's problem in the ex-ante case with heterogeneous susceptibility to manipulation is 
\[
\sup_{\widehat \pi \in \widehat \Pi } \alpha CS^A(\widehat\pi) + (1-\alpha) PS(\widehat\pi) - \int c_s(x,y) \de\widehat \pi. \tag{H} \label{eqn:het}
\]
 
Say that the transport map $T$ has \emph{monotone intermediate interval} structure if there exists some decreasing function $\gamma: S \to \mathbb{R}$ and target price $p$ such that: 
\[
T_{\gamma,p}(x,s) := 
\begin{cases}
     (x,s) \quad &\text{if $x \leq \gamma(s)$}  \\ 
     (p,s) \quad &\text{if $x \in [\gamma(s),p)$}\\
     (x,s) \quad &\text{if $x \geq p$}
\end{cases}
\]

\begin{manualtheorem}{1B} \label{thrm:het}
    \eqref{eqn:het} is solved by an advertising plan $\pi \in \widehat{\Pi}$ with a monotone intermediate interval structure.
\end{manualtheorem}

\begin{figure}[H]
\begin{minipage}[t]{0.5\linewidth}
\cref{thrm:het} is illustrated in \cref{fig:multidim}: for the target price $p^*$, the curve $\gamma(s)$ denotes the boundary. The proof of \cref{thrm:het} proceeds quite similarly to that of \cref{thrm:exante_general}; we omit it. The decreasing cutoff $\gamma(s)$ follows from the monotonicity in cost: If a consumer with type $(x, s')$ is moved to $(p, s')$, then any type $(x, s)$ should also be moved to $(p, s)$ since it leads to the same changes in CS and PS, but the cost is lower.
\end{minipage}%
\hfill%
\begin{minipage}[t]{0.5\textwidth}\vspace{0pt}
\vspace{-1.5em}
\centering
{\includegraphics[width=0.8\textwidth]{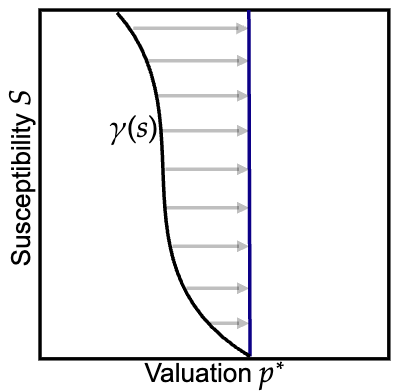}}
    \caption{}\label{fig:multidim}
\end{minipage}
\end{figure}

\section{Relaxing the directional constraint}\label{appendix:directional}
In the main text we imposed a directional constraint on advertising plans by assuming that $\pi(\mathbb{H}) = 1$ so the final distribution first-order stochastically dominates the initial distribution. This was motivated by the observation that advertising is, in practice, typically positive (perhaps with the exception of political advertising). 

This directional constraint can be relaxed completely. To formalize this, let
\[
\widetilde \Pi := \Big\{\pi \in \Delta(\Re_+^2): \marg_x \pi = \mu,\, \marg_y \pi \text{ has a finite mean}\Big\};
\]
as before we will use the notation $\nu$ to denote the final distribution, though now it is not necessarily ordered with respect to $\mu$. 

We remain interested in solutions to the problem of maximizing the general objective function $\phi$, with the difference that the designer, no longer restricted to directional advertising plans, chooses from the set $\widetilde \Pi \supseteq \Pi$: 
\begin{align*} \label{eq: U_designer_exante}
    \sup_{\pi \in \widetilde \Pi} \phi \Big( PS(\pi), CS^A(\pi), \int c d\pi \Big)
    \tag{U: ex-ante} 
\end{align*}
and 
\begin{align*} \label{eq: U_designer_expost}
    \sup_{\pi \in  \widetilde \Pi} \phi \Big( PS(\pi), CS^P(\pi), \int c d\pi \Big).
    \tag{U: ex-post} 
\end{align*}

\subsection{Solving the unconstrained problem under ex-ante welfare} For the target price $p^*$ and quantity $q^*$, say that the plan $\pi^*_{(p^*,q^*)}$ has a \emph{interval twist} structure if it is induced by the following map:
\begin{align*}
        T(x) = \begin{cases}
        x \quad &\text{if $x < \oF^{-1}(q^*)$} \\
        p^* \vee \left(x \wedge \tfrac{p^*q^*}{\oF(x)}\right) \quad &\text{if $x \geq \oF^{-1}(q^*)$.}
        \end{cases}
    \end{align*}  

Target distributions associated with \emph{interval twist} structures are depicted in \cref{fig:twist_interval}. Note that the key difference from the intermediate interval structure we saw in the main text is the presence of \emph{backward} comonotone shifts which can be employed to implement the price $p^*$ at a lower transportation cost, as well as lower manipulation cost.  Thus, relaxing the directional constraint can---for some demand curves---strictly improve the designer's value under the ex-ante consumer-optimal advertising plan.

\begin{manualtheorem}{1U} \label{thrm:exante_U}
    \eqref{eq: U_designer_exante} is solved by an advertising plan with a {interval twist} structure. Conversely, every solution to \eqref{eq: U_designer_exante} has an interval twist structure. 
\end{manualtheorem}

The form of optimal solutions can be shown via arguments similar to \cref{lemma:domination_lemma} in Appendix \ref{appendix:proofs_main}.

\begin{figure}[H]
\centering
\caption{Interval twist advertising plans
}
    {\includegraphics[width=0.6\textwidth]{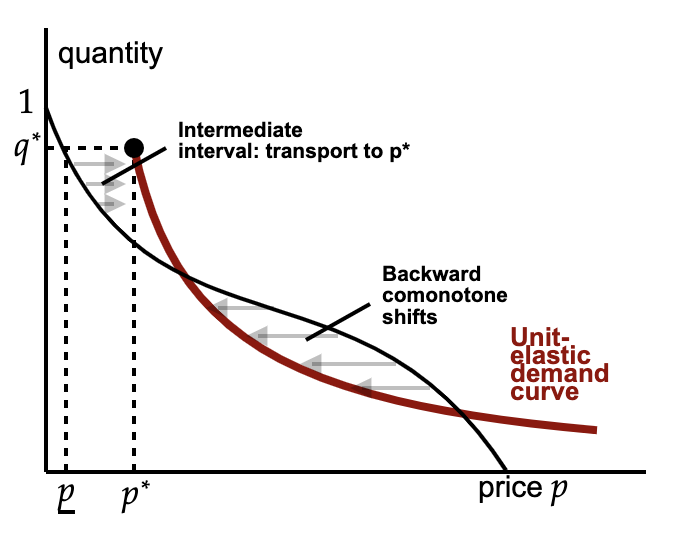}}
    \label{fig:twist_interval}
\end{figure}

\subsection{Solving the unconstrained problem under ex-post welfare}
We say a plan $\pi^* \in \Delta(\Re_+^2)$ has a \emph{constrained-greedy twist} structure for the target price and quantity pair $(p^*,q^*)$ if $p^* = \underline B(\pi^*)$ and it is induced by the following map:
\begin{align*}
        T(x) = \begin{cases}
        x \quad &\text{if $x < \oF^{-1}(q^*)$} \\
        p^* \vee \left(\Lambda_{\pi^*}(x) \wedge \tfrac{p^*q^*}{\oF(x)}  \right)\quad &\text{if $x > \oF^{-1}(q^*)$,}
        \end{cases}
    \end{align*}  
where $\Lambda_{\pi^*}$ is the locally-greedy map at $\pi^*$ as defined in \cref{sec:complements} of the main text.

Target distributions associated with \emph{interval twist} structures are depicted in \cref{fig:twist_constrainedgreedy}. Note that the key difference from the constrained-greedy plan we saw in the main text is once again the presence of \emph{backward} comonotone shifts which can be employed to implement the price $p^*$ by transporting consumers \emph{backward} to the lower-envelope (defined identically to the main text). This is reflected in the definition of constrained-greedy twist structure as we no longer require $p^*q^* \geq r^M$, so the induced unit-elastic frontier might not dominate the original demand. Once again, relaxing the directional constraint can---for some demand curves---strictly improve the designer's value under the ex-ante consumer-optimal advertising plan.

\begin{manualtheorem}{2U} \label{thrm:expost_U}
    \eqref{eq: U_designer_expost} is solved by an advertising plan with a {interval twist} structure. Conversely, every solution to \eqref{eq: U_designer_expost} has a constrained-greedy twist structure. 
\end{manualtheorem}

\begin{figure}[H]
\centering
\caption{Constrained-greedy twist advertising plans}
    {\includegraphics[width=0.6\textwidth]{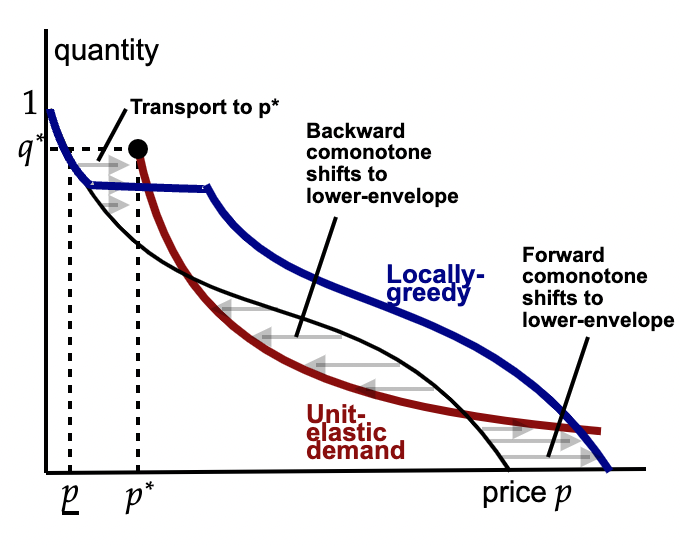}}
    \label{fig:twist_constrainedgreedy}
\end{figure}

\cref{thrm:expost_U} can be shown by adjusting the arguments in Appendix \ref{appendix: proof of expost}. It is without loss to consider $\pi$ induced by some increasing function $T$. Fixing a pair $(p, \pi)$ where $p$ is an optimal price under $\pi$, we may again translate the design problem to a sequence of pointwise problems $\max_{T(x)\geq 0} \lambda_{p,\pi}(x,T(x))$, but without directional constraint in the final valuation $T(x)\geq x$ as in the main text. Let $q = \oF_T(p)$, the final valuation $T(x) \in [p, pq/\oF(x)]$ for any type $x$ that has been shifted. With the directional constraint, $x \leq pq/\oF(x)$ and hence $\Lambda_{p,\pi}(x) \wedge \tfrac{pq}{\oF(x)}$ is the solution of the pointwise problem subject to monopolist's IC. Now, without the directional constraint, if $x > pq/\oF(x)$, then the optimal $T(x) = pq/\oF(x)$ because $\Lambda_{p,\pi}(x) \geq x$ and $\lambda_{p,\pi}(x,y)$ is increasing in $y$ when $y \leq x$. Therefore, $\Lambda_{p,\pi}(x) \wedge \tfrac{pq}{\oF(x)}$ still describes the solution of the pointwise problem subject to monopolist's IC.

\section{Connection to \cite{dixit1978advertising}} \label{appendix:dixitnorma_comparison}
In \cref{sec:comparison}, we compared targeted advertising and uniform advertising, where consumers' valuations are transported in an additive way. There are, of course, alternative ways of quantifying the extent to which valuations increase. For instance, in the seminal paper of \cite{dixit1978advertising}, they consider multiplicative increases in valuations. It will turn out that this is simply a rescaling. In particular, suppose that the cost $c$ is of the form $c(x,y) = \psi(y/x)$ for some twice continuously differentiable and strictly convex function $\psi$ with $\psi(1) = \psi'(1) = 0$. This fulfills the submodularity assumption in the main text hence, as long as optimal advertising plans exist, Theorems \ref{thrm:exante_general} and \ref{thrm:expost_general} continue to obtain.

\paragraph{Uniform advertising under multiplicative shifts.} 
It will also be instructive to explicitly compare these kinds of uniform multiplicative shifts in \cite{dixit1978advertising} with the additive uniform shifts studied in the main text.

Consider a model of uniform advertising where the designer chooses an advertising level $z \geq 0$, and consumers with initial valuation $x$ have final valuation $A(z) x$, where $A$ is a strictly increasing function on $\Re$ with $A(0) = 1$. That is, the demand under advertising level $z$ is $\oF_z(p) \coloneqq \oF(p/A(z))$ for any $z,p\geq 0$. Given this shifted demand, the monopolist's optimal posted price solves $\max_{p\geq 0} p \oF_z(p)$, which is $p^D = A(z) p^M$.

The designer's problem under PS-optimal objective is
\begin{align*}
    \sup_{p\geq 0, z\geq 0} p\oF_z(p) - \int c(x, A(z)x) \de F(x).
\end{align*}
When the original demand curve $\oF_{z = 0}$ is parametrized by an isoelastic demand curve, this is in effect the model of \cite{dixit1978advertising}. Substituting the optimal price $p^D = A(z) p^M$, the designer's problem of choosing the optimal advertising level can be rewritten as 
\begin{align*}
    \sup_{z \geq 0} A(z) p^M \oF(p^M) - \psi(A(z)).
\end{align*}
The first-order condition leads to $A'(z) (r^M - \psi'(A(z)))=0$, so the optimal $z^* = A^{-1}\circ(\psi')^{-1}(r^M)$. In \cref{sec:comparison}, we showed that the price under additive uniform advertising is lower than that under targeted advertising $p^M \leq p^U < p^*$ when the density $f$ is decreasing. We now extend this comparison to the multiplicative uniform advertising. Assume that there exists an optimal advertising plan under full flexibility (targeted advertising), a similar first-order argument as in \cref{sec:comparison} shows that the optimal intermediate interval form $(\underline p, p^*)$ satisfies 
\[
\oF(\underline p) = \int_{\underline p}^{p*} c_y(x,p^*) \de F(x)
\]
and there's no boundary solution since the cost of moving type $0$ is infinity, $p^* = c(\underline p , p^*)$. Therefore, 
\[
\oF\left(\tfrac{p^*}{\psi^{-1}(p^*)}\right) = \int_{p^*/\psi^{-1}(p^*)}^{p*} \psi'\left(\tfrac{p^*}{x}\right) \tfrac{1}{x}\de F(x).
\]

\begin{lemma}
    Suppose the initial type is uniform on $[0,1]$ and $\psi(t) = a(t-1)^2$ for some $a > 0$, then $p^M < p^D < p^*$. 
\end{lemma}
\begin{proof}
    When the initial type is uniform on $[0,1]$, $p^M = 1/2$ and $p^D = (\psi')^{-1}(r^M) p^M = (\tfrac{1}{8a}+1)/2 = \tfrac{1+8a}{16a} > 1/2$. By the first-order condition, the optimal price $p^*$ under targeted advertising satisfies
    \[
    1 - \tfrac{p^*}{\sqrt{p^*/a} + 1} - \int_{\tfrac{p^*}{\sqrt{p^*/a} + 1}}^{p^* \wedge 1} 2a(\tfrac{p^*}{x} - 1) \tfrac{1}{x} \de x = 0
    \]
    Let $G(p^*)$ denote the LHS of the previous expression as a function of $p^*$. Direct calculation shows that $G$ is strictly decreasing when $p^* < 1$. Note that 
    \[
    G(p^D) = 1 - \tfrac{p^D}{\sqrt{p^D/a} + 1} - 2a\left(\sqrt{p^D/a} - \ln (\sqrt{p^D/a}+1)\right) > 0
    \]
    for any $a > 0$, so $p^D < p^*$.
\end{proof}

Suppose the designer's objective is to maximize the ex-ante CS. Under multiplicative uniform advertising, the monopolist's optimal price is still $p^D = A(z) p^M \geq p^M$. Hence, to maximize the ex-ante CS, the designer would rather choose $z^* = 0$, so $p^D = p^M$. Similar argument as in \cref{prop:consumeroptimal_vs_uniform} shows that $p^* < p^M$, so 
\cref{prop:consumeroptimal_vs_uniform} (i) extends to the multiplicative case: restricting to multiplicative uniform advertising does not improve ex-ante CS.

\end{document}